\documentclass[twoside,11pt]{article}

\usepackage{jmlr2e}

\usepackage{alltt}
\usepackage[utf8]{inputenc}
\usepackage{amsmath}
\usepackage{amsfonts}
\usepackage{mathtools}
\usepackage{algorithm}
\usepackage{algcompatible}
\usepackage{tikz}
\usepackage{hyperref}
\usepackage{comment}

\newcommand{\vect}[1]{\boldsymbol{#1}}
\newcommand{\tp}[1]{{#1}^{\mathsf T}}

\newcommand{\eps}{\epsilon}
\newcommand{\pa}{\partial}

\renewcommand{\eps}{\varepsilon}
\renewcommand{\epsilon}{\varepsilon}
\renewcommand{\Sigma}{\varSigma}

\DeclareMathOperator{\arccot}{arccot}

\DeclareMathOperator{\sign}{sign}

\DeclareMathOperator{\diag}{diag}

\newtheorem{cor}{Corollary}
\newtheorem{prop}{Proposition}[section]

\newtheorem{dfn}{Definition}

\newtheorem{rk}{Remark}
\newcommand{\E}{\textrm E}

\newcommand{\Co}{\textrm{Cov}}
\DeclareMathAlphabet\mathbfcal{OMS}{cmsy}{b}{n}

\ifpdf
    \graphicspath{{./figure/PNG/}{./figure/PDF/}{./figure/}}
\else
    \graphicspath{{./figure/EPS/}{./figure/}}
\fi

\ShortHeadings{Spherical Augmentation}{Lan and Shahbaba}
\firstpageno{1}

\begin{document}

\title{Sampling constrained probability distributions\\ using Spherical Augmentation}

\author{\name Shiwei Lan \email s.lan@warwick.ac.uk \\
             \addr Department of Statistics\\
              University of Warwick\\
              Coventry CV4 7AL, UK
              \AND
              \name Babak Shahbaba \email babaks@uci.edu \\
              \addr Department of Statistics and Department of Computer Science\\
              University of California\\
              Irvine, CA 92697, USA}

\editor{XXX}

\maketitle
%\date

\begin{abstract}
Statistical models with constrained probability distributions are abundant in machine learning. Some examples include regression models with norm constraints (e.g., Lasso), probit, many copula models, and latent Dirichlet allocation (LDA). Bayesian inference involving probability distributions confined to constrained domains could be quite challenging for commonly used sampling algorithms. In this paper, we propose a novel augmentation technique that handles a wide range of constraints by mapping the constrained domain to a sphere in the augmented space. By moving freely on the surface of this sphere, sampling algorithms handle constraints implicitly and generate proposals that remain within boundaries when mapped back to the original space. Our proposed method, called {Spherical Augmentation}, provides a mathematically natural and computationally efficient framework for sampling from constrained probability distributions. We show the advantages of our method over state-of-the-art sampling algorithms, such as exact Hamiltonian Monte Carlo, using several examples including truncated Gaussian distributions, Bayesian Lasso, Bayesian bridge regression, reconstruction of quantized stationary Gaussian process, and LDA for topic modeling. 
\end{abstract}


\section{Introduction}
Many commonly used statistical models in Bayesian analysis involve high-dimensional
probability distributions confined to constrained domains. Some
examples include regression models with norm constraints (e.g., Lasso), probit, many copula models, and latent Dirichlet allocation (LDA). Very often, the resulting models are intractable and simulating samples for Monte Carlo
estimations is quite challenging \citep{pneal08,sherlock09, pneal12, brubaker12, pakman13}. Although the literature on improving the efficiency of computational methods for Bayesian inference is quite extensive \cite[see, for example,][]{neal96a, neal93, geyer92, mykland95, propp96, roberts97, gilks98, warnes01, freitas01, brockwell06, neal11, neal05, neal03, beal03, moller06, andrieu06, kurihara06, cappe08, craiu09, welling09, gelfand10, douc07, douc11, welling11, zhang11, ahmadian11, girolami11, hoffman11, beskos11, calderhead12, shahbaba14, ahn13, lan14a, ahn14}, these methods do not directly address the complications due to constrained target distributions. When dealing with such distributions, MCMC algorithms typically evaluate each proposal to ensure it is within the boundaries imposed by the constraints. Computationally, this is quite inefficient, especially in high dimensional problems where proposals are very likely to miss the constrained domain. Alternatively, one could map the original domain to the entire Euclidean space to remove the boundaries. This approach too is computationally inefficient since the sampler needs to explore a much larger space than needed. 

In this paper, we propose a novel method, called \emph{Spherical Augmentation}, for handling constraints involving norm inequalities (Figure \ref{fig:constraints}). Our proposed method augments the parameter space and maps the constrained domain to a sphere in the augmented space. The sampling algorithm explores the surface of this sphere. This way, it handles constraints implicitly and generates proposals that remain within boundaries when mapped back to the original space. While our method can be applied to all Metropolis-based sampling algorithms, we mainly focus on methods based on Hamiltonian Monte Carlo (HMC) \citep{duane87,neal11}. As discussed by \cite{neal11}, one could modify standard HMC such
that the sampler bounces off the boundaries by letting the potential energy go
to infinity for parameter values that violate the constraints. This creates
``energy walls'' at boundaries. This approach, henceforth called \emph{Wall
HMC}, has limited applications and tends to be computationally inefficient,
because the frequency of hitting and bouncing increases exponentially as
dimension grows. \cite{byrne13} discuss an alternative approach for situations where constrained domains can be identified as sub-manifolds. \cite{pakman13} follow the idea of Wall HMC and propose an exact
HMC algorithm specifically for truncated Gaussian distributions with
non-holonomic constraints. \cite{brubaker12} on the other hand propose a
modified version of HMC for handling holonomic constraint $c(\theta)=0$. All
these methods provide interesting solutions for specific types of constraints. In contrast, our proposed method offers a general and efficient framework applicable to a wide range of problems. 

The paper is structured as follows. Before presenting our methods, in Section \ref{prelim} we provide a
brief overview of HMC and one of its variants, namely, Lagrangian Monte Carlo (LMC) \citep{lan14a}. We then present the underlying idea of \emph{spherical augmentation}, first for two simple cases, ball type (Section \ref{ball}) and box type (Section \ref{box}) constraints, then for more general $q$-norm type constraints (Section \ref{q-norm}), as well as some functional constraints (Section \ref{funcon}). In Section \ref{SAMC}, we apply the spherical augmentation technique to HMC
(Section \ref{SphHMC}) and LMC (Section \ref{SphLMC}) for sampling from constrained target distributions. We evaluate our proposed methods using simulated and
real data in Section \ref{results}. Finally, Section \ref{discussion} is devoted to discussion and future directions.

\begin{figure}[t]
\vspace{-10pt}
\begin{center}
\includegraphics[width=5in, height=3.2in]{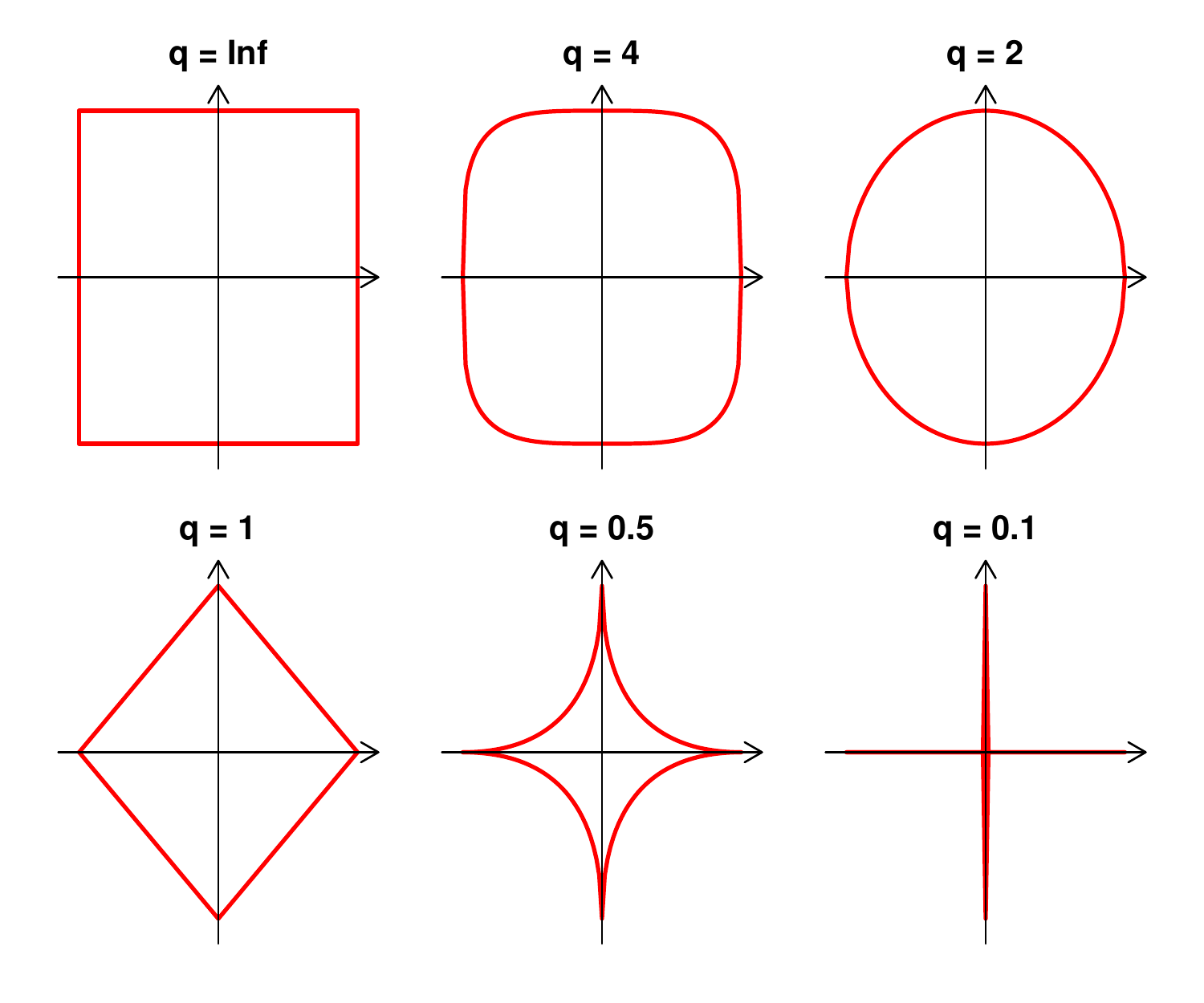}
\caption{$q$-norm constraints}
\vspace{-15pt}
\label{fig:constraints}
\end{center}
\end{figure}

\section{Preliminaries}\label{prelim}
\subsection{Hamiltonian Monte Carlo}\label{HMC}
HMC improves upon random walk Metropolis (RWM) by proposing states that are distant from the current
state, but nevertheless accepted with high probability. These distant proposals are found by numerically simulating Hamiltonian dynamics, whose state space consists of its \emph{position}, denoted by the vector $\vect\theta$, and its \emph{momentum}, denoted by the vector ${\bf p}$. Our objective is to sample from the continuous probability distribution of $\vect\theta$ with the density function $f(\vect\theta)$. It is common to assume that the fictitious momentum variable ${\bf p} \sim \mathcal N({\bf 0, M})$, where ${\bf M}$ is a symmetric, positive-definite matrix known as the \emph{mass matrix}, often set to the identity matrix ${\bf I}$ for convenience.

In this Hamiltonian dynamics, the \emph{potential energy}, $U(\vect\theta)$, is defined as minus the log density of $\vect\theta$ (plus any constant), that is $U(\vect\theta):=-\log f(\vect\theta)$; the \emph{kinetic energy}, $K({\bf p})$ for the auxiliary momentum variable ${\bf p}$ is set to be minus the log density of ${\bf p}$ (plus any constant). Then the total energy of the system, \emph{Hamiltonian} function, is defined as their sum,
\begin{equation}\label{Hamiltonian}
H(\vect\theta, {\bf p}) = U(\vect\theta) + K({\bf p})
\end{equation}
Given the Hamiltonian $H(\vect\theta, {\bf p})$, the system of $(\vect\theta, {\bf p})$ evolves according to the following \emph{Hamilton's equations},
\begin{align}
\begin{aligned}
&\dot{\vect\theta} && = && \nabla_{\bf p} H(\vect\theta, {\bf p}) && = && {\bf M}^{-1}{\bf p} \\
&\dot{\bf p} && = && -\nabla_{\vect\theta} H(\vect\theta, {\bf p}) && = && -\nabla_{\vect\theta} U(\vect\theta) 
\end{aligned}
\end{align}
%Note that since momentum is mass times velocity, ${\bf v} = {\bf M}^{-1}{\bf p}$ is regarded
%as velocity. Throughout this paper, we express the kinetic energy $K$ in terms of velocity, ${\bf v}$, instead of momentum, ${\bf p}$ \citep{beskos11,lan14a}.

%Hamiltonian dynamics have three important properties: 1) reversibility (the target distribution remains invariant), 2) conservation of the Hamiltonian (the acceptance probability is one), and 3) volume preservation (the determinant of the Jacobian matrix for the mapping is one). See Neal (2010) \cite{neal11} for more discussion.

In practice when the analytical solution to Hamilton's equations is not available, we need to numerically solve these equations by discretizing them, using some small time step $\epsilon$. For the sake of accuracy and stability, a numerical method called \emph{leapfrog} is commonly used to approximate the Hamilton's equations \citep{neal11}. We usually solve the system for $L$ steps, with some step size, $\epsilon$, to propose a new state in the Metropolis algorithm, and accept or reject it according to the Metropolis acceptance probability. \citep[See][for more discussions]{neal11}.

\subsection{Lagrangian Monte Carlo}\label{LMC}
Although HMC explores the target distribution more efficiently than RWM, it does
not fully exploit its geometric properties of the
parameter space.
\cite{girolami11} propose Riemannian HMC (RHMC), which adapts to the
local Riemannian geometry of the target distribution by using a
position-specific mass matrix ${\bf M} = {\bf G}(\vect\theta)$. More
specifically, they set ${\bf G}(\vect\theta)$ to the Fisher information matrix.
In this paper, we mainly use \emph{spherical metric} instead to serve the purpose of constraint
handling. The proposed method can be viewed as an extension to this
approach since it explores the geometry of sphere.

Following the argument of \cite{amari00}, \cite{girolami11} define Hamiltonian dynamics on
the Riemannian manifold endowed with metric ${\bf G}(\vect\theta)$.
With the non-flat metic, the momentum vector becomes ${\bf p}|\vect\theta\sim
\mathcal N({\bf 0}, {\bf G}(\vect\theta))$ and the Hamiltonian is therefore
defined as follows:
\begin{equation}\label{rmhamiltonp}
H({\vect\theta}, {\bf p})  = \phi(\vect\theta) + \frac12 \tp{\bf p}{\bf G}(\vect\theta)^{-1}{\bf p},\quad \phi(\vect\theta):= U(\vect\theta) +\frac12 \log\det{\bf G}(\vect\theta)
\end{equation}
Unfortunately the resulting Riemannian manifold Hamiltonian dynamics becomes
non-separable since it contains products of $\vect\theta$ and ${\bf p}$, and 
the numerical integrator, \emph{generalized leapfrog}, is an \emph{implicit}
scheme that involves time-consuming fixed-point iterations.

\cite{lan14a} propose to change the variables ${\bf p}\mapsto {\bf v}:= {\bf
G}(\vect\theta)^{-1}{\bf p}$ and define an \emph{explicit} integrator for RHMC by
using the following equivalent {\it Lagrangian} dynamics:
\begin{align}\label{LD}
\dot{\vect\theta} & = {\bf v} \\
\dot{\bf v} & = -\tp{\bf v}\vect\Gamma(\vect\theta){\bf v} - {\bf G}(\vect\theta)^{-1} \nabla_{\vect\theta}\phi(\vect\theta)
\end{align}
where the \emph{velocity} ${\bf v}|\vect\theta \sim\mathcal N({\bf 0}, {\bf
G}(\vect\theta)^{-1})$. Here, $\vect\Gamma(\vect\theta)$ is the Christoffel Symbols
%of the second kind whose $(i,j,k)$-th element is
%$\Gamma_{ij}^k=\frac{1}{2}g^{km}(\pa_i g_{mj}+\pa_j g_{im}-\pa_m g_{ij})$ with
%$g^{km}$ being the $(k,m)$-th element of ${\bf G}(\vect\theta)^{-1}$.
derived from ${\bf G}(\vect\theta)$.

The proposed \emph{explicit} integrator is time
reversible but not volume preserving. 
Based on the change of variables theorem, one can adjust the acceptance
probability with Jacobian determinant to satisfy the detailed balance condition.
The resulting algorithm, \emph{Lagrangian Monte Carlo (LMC)}, is
shown to be more efficient than RHMC \citep[See][for more details]{lan14a}.

Throughout this paper, we express the kinetic energy $K$ in terms of velocity,
${\bf v}$, instead of momentum, ${\bf p}$ \citep{beskos11,lan14a}.

%%%%%%%%%%%%%%%%%%%%%%%%%%%%%%%%%%%%%%%%%%%%%%%%%%%%%%%%%%%%%%%%%%%%%%%%%%%%%%%%%%%%%

\section{Spherical Augmentation}\label{SA}
In this section, we introduce the \emph{spherical augmentation} technique for handling norm constraints implicitly. We start with two simple constraints: ball type (2-norm) and box type ($\infty$-norm). Then, we generalize the
methodology to arbitrary $q$-norm type constraints for $q>0$. Finally, we discuss some functional constraints
that can be reduced to norm constraints.

\begin{figure}[t]
\vspace{-20pt}
\begin{center}
\includegraphics[width=.8\textwidth, height=.55\textwidth]{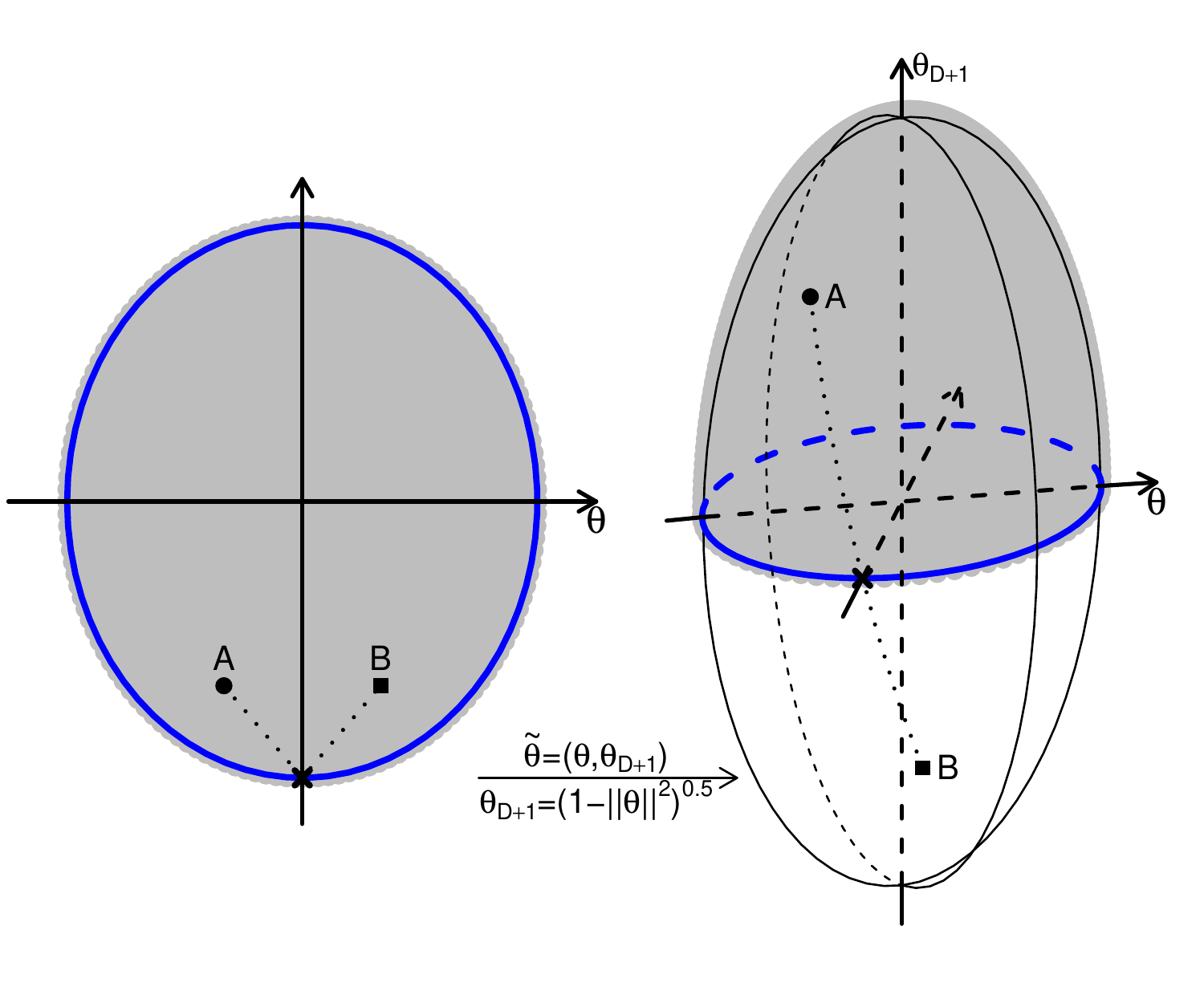}
\vspace{-10pt}
\caption{Transforming the unit ball $\mathcal B_{\bf 0}^D(1)$ to the sphere $\mathcal S^D$.}
\vspace{-20pt}
\label{fig:B2S}
\end{center}
\end{figure}
\subsection{Ball type constraints} \label{ball}
Consider probability distributions confined to
the $D$-dimensional unit ball $\mathcal B_{\bf 0}^D(1):=\{\vect\theta\in\mathbb R^D:
\Vert \vect\theta\Vert_2 =\sqrt{\sum_{i=1}^D \theta_i^2}\leq 1\}$. The
constraint is given by restricting the 2-norm of parameters: $\Vert
\vect\theta\Vert_2\leq 1$.

The idea of spherical augmentation is to augment the original $D$-dimensional
manifold of unit ball $\mathcal B_{\bf 0}^D(1)$ to a hyper-sphere $\mathcal S^D
:=\{\tilde{\vect\theta}\in \mathbb R^{D+1}: \Vert
\tilde{\vect\theta}\Vert_2=1\}$ in $(D+1)$-dimensional space.
This can be done by adding an auxiliary variable $\theta_{D+1}$ to the
original parameter $\vect\theta\in\mathcal B_{\bf 0}^D(1)$ to form an extended
parameter $\tilde{\vect\theta}=(\vect\theta,\theta_{D+1})$ such that $\theta_{D+1} = \sqrt{1-\Vert
\vect\theta\Vert_2^2}$.
Next, we identify the lower hemisphere $\mathcal S_-^D$ with the upper
hemisphere $\mathcal S_+^D$ by ignoring the sign of $\theta_{D+1}$.
This way, the domain of the target distribution is changed from the unit ball
$\mathcal B_{\bf 0}^D(1)$ to the $D$-dimensional sphere, $\mathcal S^D
:=\{\tilde{\vect\theta}\in \mathbb R^{D+1}: \Vert \tilde{\vect\theta}\Vert_2=1\}$, through the following transformation:
\begin{equation}\label{b2s}
T_{\mathcal B\to \mathcal S}:\, \mathcal B_{\bf 0}^D(1)\longrightarrow \mathcal S^D, \quad \vect\theta \mapsto \tilde{\vect\theta} = (\vect\theta, \pm\sqrt{1-\Vert \vect\theta\Vert_2^2})
\end{equation}
which can also be recognized as the coordinate map from the Euclidean coordinate
chart $\{\vect\theta,\mathcal B_{\bf 0}^D(1)\}$ to the manifold $\mathcal S^D$.

After collecting samples $\{\tilde{\vect\theta}\}$ using a sampling algorithm (e.g., HMC)
defined on the sphere, $\mathcal S^D$, we discard the last component
$\theta_{D+1}$ and obtain the samples $\{\vect\theta\}$ that automatically
satisfy the constraint $\Vert \vect\theta\Vert_2\leq 1$.
Note that the sign of $\theta_{D+1}$ does not affect our Monte Carlo estimates.
However, after applying the above transformation, we need to adjust our
estimates according to the change of variables theorem as
follows:
\begin{equation}\label{domainB2S}
\int_{\mathcal B_{\bf 0}^D(1)} f(\vect\theta) d\vect\theta_{\mathcal B} = \int_{\mathcal S_+^D} f(\tilde{\vect\theta}) \left|\frac{d\vect\theta_{\mathcal B}}{d\vect\theta_{\mathcal S_c}}\right| d\vect\theta_{\mathcal S_c}
\end{equation}
where $\left|\frac{d\vect\theta_{\mathcal B}}{d\vect\theta_{\mathcal
S_c}}\right|=|\theta_{D+1}|$ as shown in Corollary \ref{voladjB2S} in Appendix \ref{Metc}. Here,
$d\vect\theta_{\mathcal B}$ and $d\vect\theta_{\mathcal S_c}$ are
volume elements under the Euclidean metric and the \emph{canonical spherical
metric} respectively.

With the above transformation \eqref{b2s}, the resulting sampler is defined and moves freely on $\mathcal S^D$ while
implicitly handling the constraints imposed on the original parameters. As
illustrated in Figure \ref{fig:B2S}, the boundary of the constraint, i.e.,
$\Vert\vect\theta\Vert_2=1$, corresponds to the equator on the sphere $\mathcal
S^D$. Therefore, as the sampler moves on the sphere, e.g. from $A$ to $B$, passing
across the equator from one hemisphere to the other translates to ``bouncing back'' off the boundary in the original parameter space.

\begin{figure}[t]
\vspace{-20pt}
\begin{center}
\includegraphics[width=.8\textwidth, height=.5\textwidth]{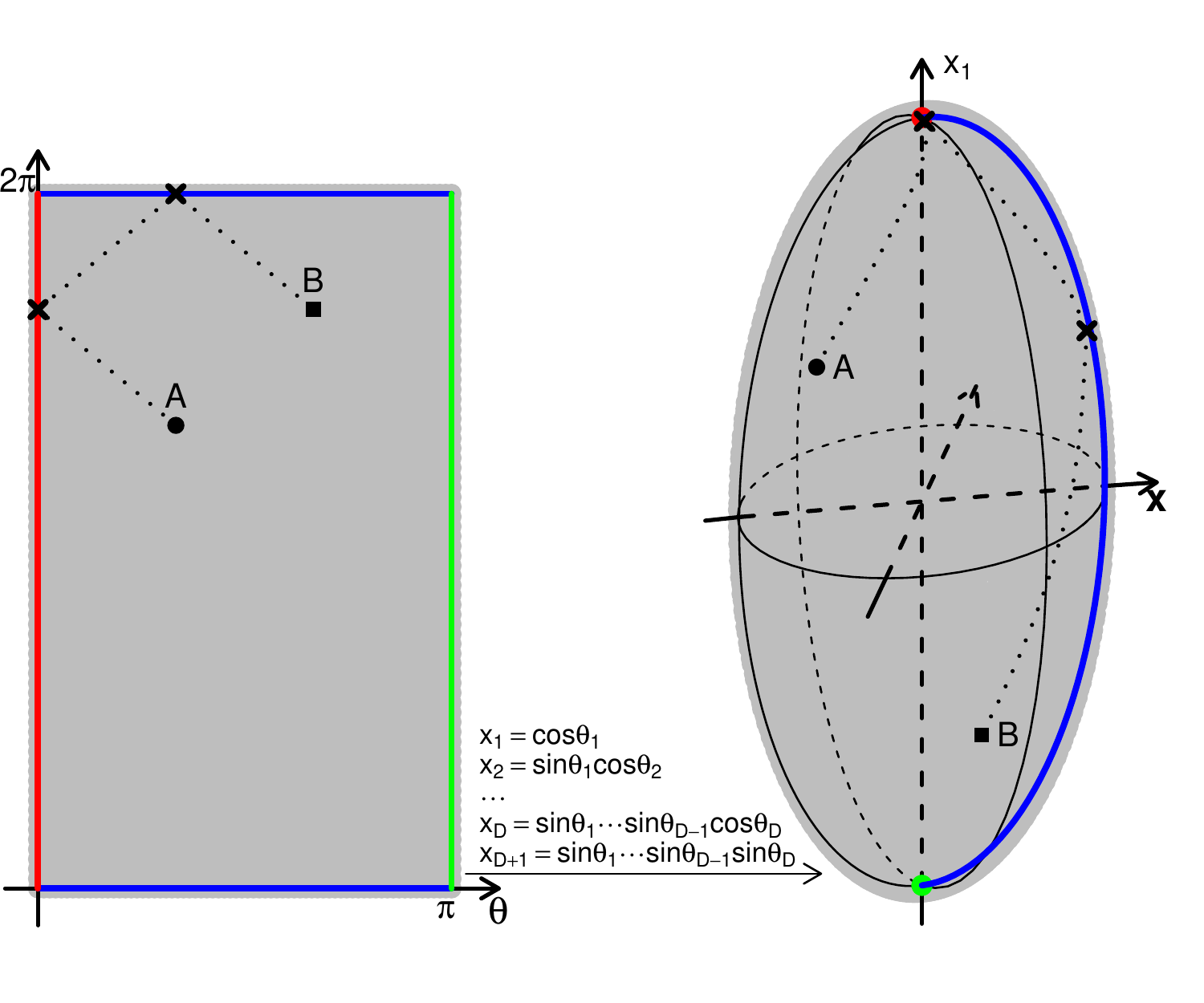}
\vspace{-10pt}
\caption{Transforming the hyper-rectangle $\mathcal R_{\bf 0}^D$ to the sphere $\mathcal S^D$.}
\vspace{-20pt}
\label{fig:R2S}
\end{center}
\end{figure}
\subsection{Box type constraints} \label{box}
Many constraints are given by both lower and upper bounds. Here we
focus on a special case that defines a hyper-rectangle 
$\mathcal R_{\bf 0}^D:=[0,\pi]^{D-1}\times [0,2\pi)$; other \emph{box} type constraints can be transformed to this hyper-rectangle. This constrained domain can be
mapped to the unit ball $\mathcal B_{\bf 0}^D(1)$ and thus reduces to the ball type
constraint discussed in Section \ref{ball}. However, a more natural approach is to use \emph{spherical} coordinates, which directly map the hyper-rectangle $\mathcal R_{\bf 0}^D$ to the sphere $\mathcal S^D$,
\begin{equation}\label{r2s}
T_{\mathcal R_{\bf 0}\to \mathcal S}:\, \mathcal R_{\bf 0}^D \longrightarrow \mathcal S^D, \quad \vect\theta\mapsto{\bf x},\;
x_d = \begin{cases}
\cos(\theta_d)\prod_{i=1}^{d-1}\sin(\theta_i), & d<D+1\\
\prod_{i=1}^D\sin(\theta_i), & d=D+1
\end{cases}
\end{equation}
Therefore, we use $\{\vect\theta, \mathcal R_{\bf 0}^D\}$ as the
spherical coordinate chart for the manifold $\mathcal S^D$.
Instead of being appended with an extra dimension as in Section \ref{ball},
here $\vect\theta\in\mathbb R^D$ is treated as the spherical coordinates of
the point ${\bf x}\in\mathbb R^{D+1}$ with $\Vert {\bf x}\Vert_2=1$.

After obtaining samples $\{\bf x\}$ on the sphere $\mathcal S^D$,
we transform them back to $\{\vect\theta\}$ in the original constrained domain $\mathcal  R_{\bf
0}^D$ using the following inverse mapping of \eqref{r2s}:
\begin{equation}\label{s2r}
T_{\mathcal S\to \mathcal R_0}: \mathcal S^D \longrightarrow \mathcal R_{\bf 0}^D,\; {\bf x}\mapsto\vect\theta,\;
\theta_d = \begin{dcases}\arccot\frac{x_d}{\sqrt{1-\sum_{i=1}^d x_i^2}},&d<D
%\\2\arccot\frac{x_D+\sqrt{x_{D+1}^2+x_D^2}}{x_{D+1}},&d=D
\\ \arccot\frac{x_D}{x_{D+1}}+\frac{\pi}{2}\sign(x_{D+1})(\sign(x_{D+1})-1),&d=D
\end{dcases}
\end{equation}
Similarly, we need to adjust the estimates based on the following change of
variables formula:
\begin{equation}\label{domainR2S}
\int_{\mathcal R_{\bf 0}^D} f(\vect\theta) d\vect\theta_{\mathcal R_{\bf 0}} = \int_{\mathcal S^D} f(\vect\theta) \left|\frac{d\vect\theta_{\mathcal R_{\bf 0}}}{d\vect\theta_{\mathcal S_r}}\right| d\vect\theta_{\mathcal S_r}
\end{equation}
where $\left|\frac{d\vect\theta_{\mathcal R_{\bf 0}}}{d\vect\theta_{\mathcal
S_r}}\right|=\prod_{d=1}^{D-1}\sin^{-(D-d)}(\theta_{d})$ as shown Proposition \ref{voladjR2S} in Appendix
\ref{Mets}. Here, $d\vect\theta_{\mathcal R_{\bf 0}}$ and
$d\vect\theta_{\mathcal S_r}$ are volume elements under the Euclidean metric and the \emph{round spherical metric} respectively.

With the above transformation \eqref{r2s}, we can derive sampling methods on the sphere to implicitly handle box type constraints. As illustrated in Figure \ref{fig:R2S}, the red vertical boundary of $\mathcal R_{\bf 0}^D$ collapses to the north pole of $\mathcal S^D$, while the green vertical boundary collapses to the south pole. Two blue horizontal boundaries are mapped to the same prime meridian of $\mathcal S^D$ shown in blue color. As the sampler moves freely on the sphere $\mathcal S^D$, the resulting samples automatically satisfy the original constraint after being transformed back to the original domain.

\begin{figure}[t]
\vspace{-20pt}
\begin{center}
\includegraphics[width=5in, height=3.5in]{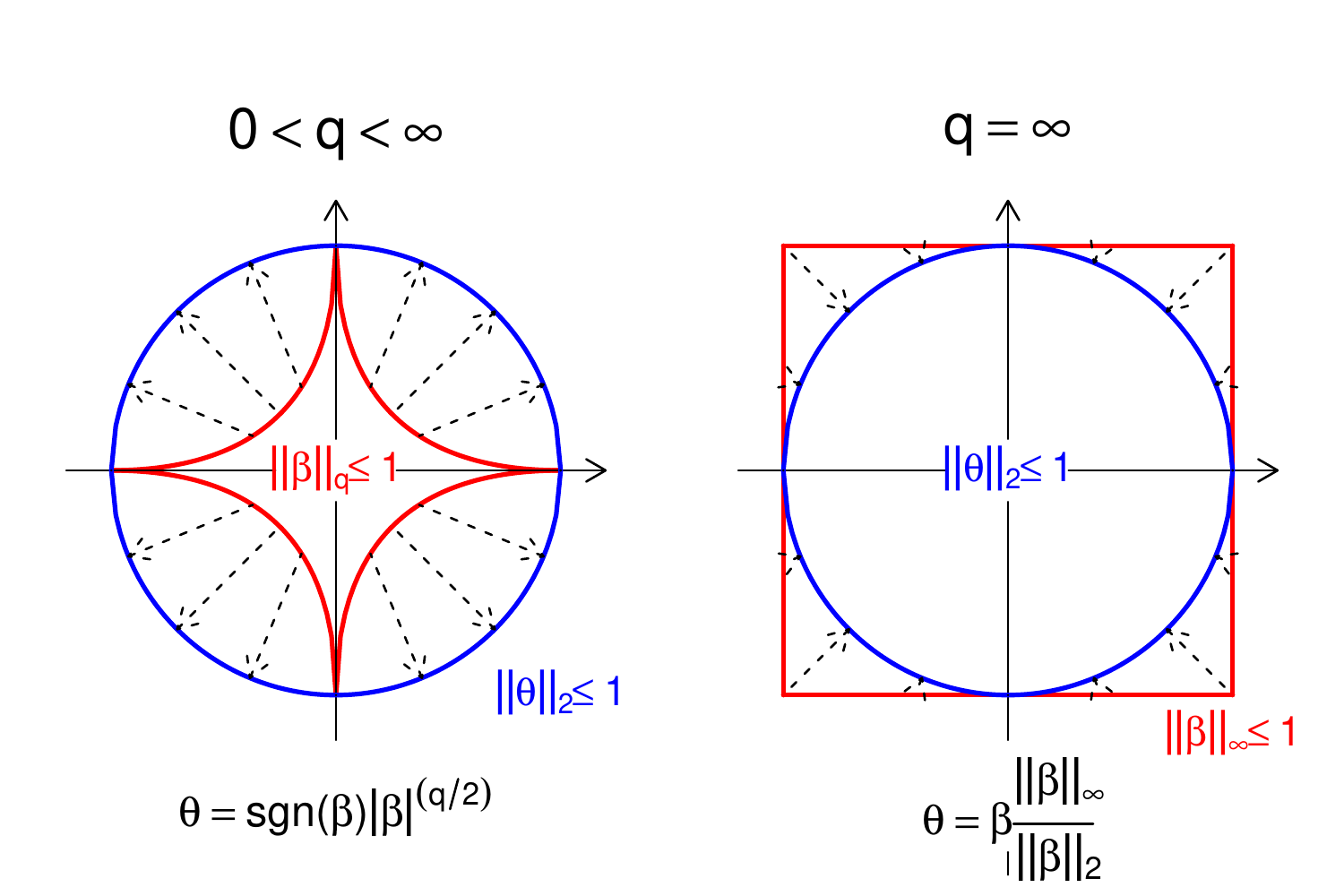}
\vspace{-10pt}
\caption{Transforming $q$-norm constrained domain to unit ball. Left: from unit
cube ${\mathcal C}^D$ to unit ball ${\mathcal B}_0^D(1)$; Right from general
$q$-norm domain $\mathcal Q^D$ to unit ball ${\mathcal B}_0^D(1)$.}
\vspace{-20pt}
\label{fig:changeofdomain}
\end{center}
\end{figure}

\subsection{General $q$-norm constraints} \label{q-norm}
The ball and box type constraints discussed in previous sections are in fact special cases
of more general $q$-norm constraints with $q$ set to 2 and $\infty$ respectively. In general, these constraints are expressed in terms of $q$-norm of the parameter vector $\vect\beta\in\mathbb R^D$,
\begin{equation}\label{qnorm}
\Vert \vect\beta\Vert_q =
\begin{cases}
(\sum_{i=1}^D |\beta_i|^q)^{1/q}, & q\in (0,+\infty)\\
\max_{1\leq i\leq D} |\beta_i|, & q=+\infty
\end{cases}
\end{equation}
This class of constraints is very common in statistics and machine learning.
For example, when $\vect\beta$ are regression parameters, $q=2$ corresponds to the ridge
regression and $q=1$ corresponds to Lasso \citep{tibshirani96}.

Denote the domain constrained by general $q$-norm as $\mathcal Q^D:=
\{\vect\beta\in \mathbb R^D: \Vert \vect\beta\Vert_q \leq 1\}$.
It could be quite challenging to sample probability distributions defined on $\mathcal Q^D$ (see Figure \ref{fig:constraints}).
To address this issue, we propose to transform $\mathcal Q^D$ to the unit ball
$\mathcal B_{\bf 0}^D(1)$ so that the method discussed in Section \ref{ball} can be applied.
As before, sampling methods defined on the sphere $\mathcal S^D$ generate samples
that automatically fall within $\mathcal B_{\bf 0}^D(1)$. Then we transform those samples
back to the $q$-norm domain, $\mathcal Q^D$, and adjust the estimates with the
following change of variables formula:
\begin{equation}\label{domainQ2S}
\int_{\mathcal Q^D} f(\vect\beta) d\vect\beta_{\mathcal Q} = \int_{\mathcal S_+^D} f(\tilde{\vect\theta}) \left|\frac{d\vect\beta_{\mathcal Q}}{d\vect\theta_{\mathcal S_c}}\right| d\vect\theta_{\mathcal S_c}
\end{equation}
where $\left|\frac{d\vect\beta_{\mathcal Q}}{d\vect\theta_{\mathcal
S_c}}\right|=\left|\frac{d\vect\beta_{\mathcal Q}}{d\tp{\vect\theta}_{\mathcal B}}\right| \left|\frac{d\vect\theta_{\mathcal B}}{d\vect\theta_{\mathcal
S_c}}\right|=\left|\frac{d\vect\beta_{\mathcal Q}}{d\tp{\vect\theta}_{\mathcal B}}\right| |\theta_{D+1}|$.
In the following, we introduce the bijective mappings between $\mathcal Q^D$ and
$\mathcal B_{\bf 0}^D(1)$ and specify the associated Jacobian determinants $\left|\frac{d\vect\beta_{\mathcal Q}}{d\tp{\vect\theta}_{\mathcal B}}\right|$.

\subsubsection{Norm constraints with $q\in (0,+\infty)$}\label{q-reg}
For $q\in (0,+\infty)$, $q$-norm domain $\mathcal Q^D$ can be transformed to
the unit ball $\mathcal B_0^D(1)$ bijectively via the following map (illustrated by
 the left panel of Figure \ref{fig:changeofdomain}):
\begin{equation}\label{q2b}
T_{\mathcal Q\to\mathcal B}:\, \mathcal Q^D \rightarrow {\mathcal B}_0^D(1), \quad \beta_i \mapsto \theta_i = \mathrm{sgn}(\beta_i)|\beta_i|^{q/2}
\end{equation}
The Jacobian determinant of $T_{\mathcal B\to\mathcal
Q}$ is $\left|\frac{d\vect\beta_{\mathcal Q}}{d\tp{\vect\theta}_{\mathcal
B}}\right| = \left(\frac{2}{q}\right)^D
\left(\prod_{i=1}^{D}|\theta_i|\right)^{2/q-1}$. See Appendix \ref{Jacobian} for more details.

%The following proposition gives the weights needed for the transformation from
%$\mathcal Q^D$ to $\mathcal B_{\bf 0}^D(1)$.
%\begin{prop}
%The Jacobian determinant (weight) of $T_{\mathcal B\to\mathcal Q}$ is as follows:
%\begin{equation}
%|dT_{\mathcal S\to\mathcal Q}| = \left(\frac{2}{q}\right)^D \left(\prod_{i=1}^{D}|\theta_i|\right)^{2/q-1}
%\end{equation}
%\end{prop}
%\begin{proof}
%Note
%\begin{equation*}
%T_{\mathcal B\to\mathcal Q}:\, \vect\theta\mapsto \vect\beta=\mathrm{sgn}(\vect\theta)|\vect\theta|^{2/q}
%\end{equation*}
%The Jacobian matrix for $T_{\mathcal B\to\mathcal Q}$ is
%\begin{equation*}
%\frac{d\vect\beta}{d\tp{\vect\theta}} = \frac{2}{q}\mathrm{diag}(|\vect\theta|^{2/q-1})
%\end{equation*}
%Therefore the Jacobian determinant of $T_{\mathcal B\to \mathcal Q}$ is
%\begin{equation*}
%|dT_{\mathcal B\to\mathcal Q}|
%= \left|\frac{d\vect\beta}{d\tp{\vect\theta}}\right| = \left(\frac{2}{q}\right)^D
%\left(\prod_{i=1}^{D}|\theta_i|\right)^{2/q-1}
%\end{equation*}
%\end{proof}

\subsubsection{Norm constraints with $q = +\infty$}\label{q-inf}
When $q = +\infty$, the norm inequality defines a unit \emph{hypercube},
${\mathcal C}^D:=[-1,1]^D=\{\vect\beta\in\mathbb R^D: \Vert \vect\beta
\Vert_{\infty} \leq 1\}$, from which the more general form, \emph{hyper-rectangle},
${\mathcal R}^D:=\{\vect\beta\in \mathbb R^D: {\bf l} \leq \vect\beta\leq {\bf
u} \}$, can be obtained by proper shifting and scaling.
The unit hypercube ${\mathcal C}^D$ can be transformed to its inscribed
unit ball $\mathcal B_{\bf 0}^D(1)$ through the following map (illustrated by
 the right panel of Figure \ref{fig:changeofdomain}):
\begin{equation}\label{c2b}
T_{\mathcal C\to \mathcal B}:\, [-1,1]^D \rightarrow {\mathcal B}_0^D(1), \quad \vect\beta \mapsto \vect\theta = \vect\beta \frac{\Vert \vect\beta\Vert_{\infty}}{\Vert \vect\beta\Vert_2} 
\end{equation}
The Jacobian determinant of $T_{\mathcal B\to\mathcal R}$ is
$\left|\frac{d\vect\beta_{\mathcal R}}{d\tp{\vect\theta}_{\mathcal
B}}\right|=\frac{\Vert\vect\theta\Vert_2^D}{\Vert\vect\theta\Vert_{\infty}^D} \prod_{i=1}^D \frac{u_i-l_i}{2}$. More details can be found in Appendix \ref{Jacobian}.

\subsection{Functional constraints}\label{funcon}
Many statistical problems involve functional constraints. For example, \cite{pakman13} discuss linear and quadratic constraints for multivariate Gaussian distributions. Since the target distribution is truncated Gaussian, 
Hamiltonian dynamics can be exactly simulated and the
boundary-hitting time can be analytically obtained. However, finding the hitting time and reflection trajectory is computationally expensive. Some constraints of
this type can be handled by the spherical augmentation method more efficiently. Further, our method can be use for sampling from a wide range of distributions beyond Gaussian.

\begin{figure}[t]
\begin{center}
\includegraphics[width=.9\textwidth,height=0.45\textwidth]{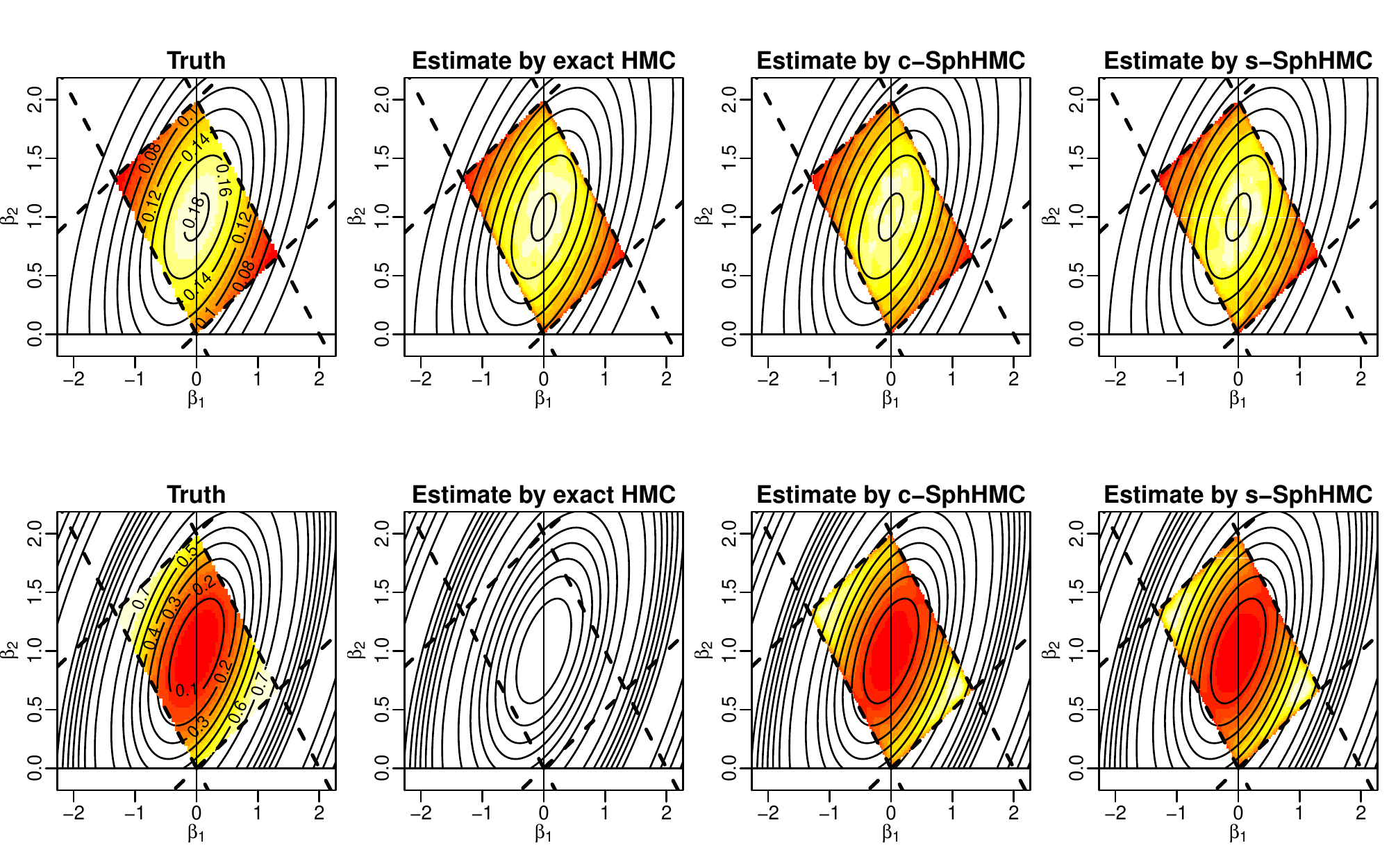}
\caption[Linear constraints]{Sampling from a Gaussian distribution (first row) and a damped sine wave distribution (second row) with linear constraints. First column shows the true distributions. The exact HMC method of \cite{pakman13} is shown in the second column. The last two columns show our proposed methods.
\label{fig:density_lin}}
\end{center}
\end{figure}

\subsubsection{Linear constraints}
In general, $M$ linear constraints can be written as ${\bf
 l}\leq {\bf A}\vect\beta\leq {\bf u}$, where ${\bf A}$ is $M\times D$ matrix,
 $\vect\beta$ is a $D$-vector, and the boundaries ${\bf l}$ and ${\bf u}$ are both $M$-vectors.
Here, we assume $M=D$ and ${\bf A}_{D\times D}$ is invertible. (Note that we generally
do not have ${\bf A}^{-1}{\bf l}\leq \vect\beta\leq {\bf A}^{-1}{\bf u}$.) Instead of sampling $\vect\beta$ directly, we sample $\vect\eta:={\bf
A}\vect\beta$ with the box type constraint: ${\bf l}\leq \vect\eta\leq {\bf u}$.
Now we can apply our proposed method to sample $\vect\eta$ and transform it back to $\vect\beta={\bf A}^{-1}\vect\eta$.
In this process, we use the following change of variables formula:
\begin{equation}\label{domainLinX}
\int_{{\bf l}\leq {\bf A}\vect\beta\leq {\bf u}} f(\vect\beta) d\vect\beta = \int_{{\bf l}\leq \vect\eta\leq {\bf u}} f(\vect\eta) \left|\frac{d\vect\beta}{d\vect\eta}\right| d\vect\eta
\end{equation}
where $\left|\frac{d\vect\beta}{d\vect\eta}\right|=|{\bf A}|^{-1}$.

Figure \ref{fig:density_lin} illustrates that both exact HMC \citep{pakman13}
and HMC with spherical augmentation can handle linear constraints, here ${\bf l}={\bf 0}$, ${\bf A}=
\begin{bmatrix}-0.5&1\\1&1\end{bmatrix}$ and ${\bf u}={\bf 2}$, imposed on a 2d Gaussian
distribution $\mathcal N(\vect\mu,\vect\Sigma)$ with
$\vect\mu=\begin{bmatrix}0\\1\end{bmatrix}$ and $\vect\Sigma=\begin{bmatrix}1&0.5\\0.5&1\end{bmatrix}$ (first row). However, the exact HMC is not applicable to more complicated distributions such as the damped sine wave distribution (second row in Figure \ref{fig:density_lin}) with the following density:
\begin{equation}
f(\vect\beta) \propto \frac{\sin^2 Q(\vect\beta)}{Q(\vect\beta)}, \quad Q(\vect\beta)=\frac12\tp{(\vect\beta-\vect\mu)}\vect\Sigma^{-1}(\vect\beta-\vect\mu)
\end{equation}
However, it is worth noting that for truncated Gaussian distributions, the exact HMC method of \cite{pakman13} can handle a wider range of linear constraints compared to our method. 

\subsubsection{Quadratic constraints}
General quadratic constraints can be written as $l\leq \tp{\vect\beta}{\bf A}\vect\beta+\tp{\bf b}\vect\beta \leq u$, where $l, u>0$
are scalars. We assume ${\bf A}_{D\times D}$ symmetric and positive definite.
By spectrum theorem, we have the decomposition ${\bf A}={\bf Q}\vect\Sigma\tp{\bf
Q}$, where ${\bf Q}$ is an orthogonal matrix and $\vect\Sigma$ is a diagonal
matrix of eigenvalues of ${\bf A}$. By shifting and scaling, $\vect\beta\mapsto
\vect\beta^*=\sqrt{\vect\Sigma}\tp{\bf Q}(\vect\beta+\frac{1}{2}{\bf
A}^{-1}{\bf b})$, we only need to consider the \emph{ring} type constraints for
$\vect\beta^*$,
\begin{equation}\label{ring}
\circledcirc:\, l^*\leq \Vert \vect\beta^*\Vert_2^2=\tp{(\vect\beta^*)}\vect\beta^* \leq u^*,\quad l^*=l+\frac{1}{4}\tp{\bf b}{\bf A}^{-1}{\bf b},\; u^*=u+\frac{1}{4}\tp{\bf b}{\bf
A}^{-1}{\bf b}
\end{equation}
which can be mapped to the unit ball as follows:
\begin{equation}\label{r2b}
T_{\circledcirc\to\mathcal B}:\, \mathcal B^D_{\bf 0}(\sqrt{u^*})\backslash \mathcal B^D_{\bf 0}(\sqrt{l^*}) \longrightarrow \mathcal B^D_{\bf 0}(1), \quad \vect\beta^* \mapsto \vect\theta=\frac{\vect\beta^*}{\Vert
\vect\beta^*\Vert_2}
\frac{\Vert \vect\beta^*\Vert_2-\sqrt{l^*}}{\sqrt{u^*}-\sqrt{l^*}}
\end{equation}
We can then apply our proposed method in Section \ref{ball} to obtain samples
$\{\vect\theta\}$ in $\mathcal B^D_{\bf 0}(1)$ and transform them back to the
original domain with the following inverse operation of \eqref{r2b}:
\begin{equation}\label{b2r}
T_{\mathcal B\to \circledcirc}:\, \mathcal B^D_{\bf 0}(1) \longrightarrow \mathcal B^D_{\bf 0}(\sqrt{u^*})\backslash \mathcal B^D_{\bf 0}(\sqrt{l^*}), \quad \vect\theta \mapsto
\vect\beta^*=\dfrac{\vect\theta}{\Vert\vect\theta\Vert_2} ((\sqrt{u^*}-\sqrt{l^*})\Vert\vect\theta\Vert_2+\sqrt{l^*})
\end{equation}
In this process, we need the change of variables formula
\begin{equation}\label{domainQuadX}
\int_{l\leq \tp{\vect\beta}{\bf A}\vect\beta+\tp{\bf b}\vect\beta \leq u} f(\vect\beta) d\vect\beta = \int_{\mathcal S_+^D} f(\vect\theta) \left|\frac{d\vect\beta}{d\vect\theta_{\mathcal S_c}}\right| d\vect\theta_{\mathcal S_c}
\end{equation}
where $\left|\frac{d\vect\beta}{d\vect\theta_{\mathcal
S_c}}\right|=\left|\frac{d\vect\beta}{d\tp{(\vect\beta^*)}}\right|
\left|\frac{d\vect\beta^*}{d\tp{\vect\theta}_{\mathcal B}}\right| \left|\frac{d\vect\theta_{\mathcal B}}{d\vect\theta_{\mathcal S_c}}\right| = |{\bf A}|^{-\frac12}\alpha^{D-1}(\alpha-\sqrt{l^*}) |\theta_{D+1}|$, $\alpha=\sqrt{u^*}+(1/\Vert\vect\theta\Vert_2-1)\sqrt{l^*}$.

\subsubsection{More general constraints}
We close this section with some comments on more general types of constraints. In some problems, several parameters might be unconstrained, and the type of constraints might be vary across the constrained parameters. In such cases, we could group the parameters into blocks and update each block separately using the methods discussed in this section. When dealing with one-sided constraints, e.g. $\theta_i\geq l_i$, one can map the constrained domain to the whole space and sample the unconstrained parameter $\theta_i^*$, where $\theta_i=|\theta_i^*|+l_i$. Alternatively, the one-sided constraint $\theta_i\geq l_i$ can be changed to a two-sided constraint for $\theta_i^*\in (0,1)$ by setting $\theta_i=-\log\theta_i^*+l_i$.

%%%%%%%%%%%%%%%%%%%%%%%%%%%%%%%%%%%%%%%%%%%%%%%%%%%%%%%%%%%%%%%%%%%%%%%%%%%%%%%%%%%%%
\section{Monte Carlos with Spherical Augmentation}\label{SAMC}
In this section, we show how the idea of spherical augmentation can be used to improve Markov Chain Monte Carlo methods applied to constrained probability distributions. In particular, we focus on two state-of-the-art sampling algorithms, namely Hamiltonian Monte Carlo\citep{duane87,neal11}, and Lagrangian Monte Carlo\citep{lan14a}. Note however that our proposed method is generic so its application goes beyond these two algorithms.

\subsection{Common settings}
Throughout this section, we denote the original parameter vector as $\vect\beta$, the constrained domain as
$\mathcal D$, the coordinate vector of sphere $\mathcal S^D$ as $\vect\theta$. All the change of variables formulae presented in the previous section can be summarized as
\begin{equation}\label{domainX}
\int_{\mathcal D} f(\vect\beta) d\vect\beta_{\mathcal D} = \int_{\mathcal S} f(\vect\theta) \left|\frac{d\vect\beta_{\mathcal D}}{d\vect\theta_{\mathcal S}}\right| d\vect\theta_{\mathcal S}
\end{equation}
where $\left|\frac{d\vect\beta_{\mathcal D}}{d\vect\theta_{\mathcal S}}\right|$
is the Jacobian determinant of the mapping $T:\mathcal S\longrightarrow\mathcal
D$ and $d\vect\theta_{\mathcal S}$ is some spherical measure.

For energy based MCMC algorithms like HMC, RHMC and LMC, we need to investigate
the change of energy under the above transformation.
The original potential energy function
$U(\vect\beta)=-\log f(\vect\beta)$ should be transformed to the following $\phi(\vect\theta)$
\begin{equation}\label{potentialX}
\phi(\vect\theta) = -\log f(\vect\theta) - \log\left|\frac{d\vect\beta_{\mathcal D}}{d\vect\theta_{\mathcal S}}\right| = U(\vect\beta(\vect\theta)) - \log\left|\frac{d\vect\beta_{\mathcal D}}{d\vect\theta_{\mathcal S}}\right| 
\end{equation}
Consequently the total energy $H(\vect\beta,{\bf v})$ in \eqref{Hamiltonian}
becomes
\begin{equation}\label{HamiltonianX}
H(\vect\theta,{\bf v}) = \phi(\vect\theta) + \frac12 \langle {\bf v}, {\bf v}\rangle_{{\bf G}_{\mathcal S}(\vect\theta)}
\end{equation}
The gradient of potential energy $U$, metric and natural gradient (preconditioned gradient)
under the new coordinate system $\{\vect\theta,\mathcal S^D\}$ can be calculated as follows
\begin{align}
\nabla_{\vect\theta} U(\vect\theta) &= \frac{d\tp{\vect\beta}}{d\vect\theta} \nabla_{\vect\beta} U(\vect\beta)\\
{\bf G}_{\mathcal S}(\vect\theta) &= \frac{d\tp{\vect\beta}}{d\vect\theta} {\bf G}_{\mathcal D}(\vect\beta) \frac{d\vect\beta}{d\tp{\vect\theta}}\\
{\bf G}_{\mathcal S}(\vect\theta)^{-1} \nabla_{\vect\theta} U(\vect\theta) &= \left[\frac{d\tp{\vect\beta}}{d\vect\theta}\right]^{-1} {\bf G}_{\mathcal D}(\vect\beta)^{-1} \nabla_{\vect\beta} U(\vect\beta)
\end{align}

\subsection{Spherical Hamiltonian Monte Carlo}\label{SphHMC}
We define HMC on the sphere $\mathcal S^D$ in two different
coordinate systems: the \emph{Cartesian coordinate} and the \emph{spherical coordinate}.
The former is applied to ball type constraints or those that could be converted
to ball type constraints; the later is more suited for box type constraints.
Besides the merit of implicitly handling constraints, HMC on sphere can take
advantage of the splitting technique \citep{beskos11, shahbaba14, byrne13} to
further improve its computational efficiency.

\subsubsection{Spherical HMC in the Cartesian coordinate}\label{SphHMCc}
We first consider HMC for the target distribution with density $f(\vect\theta)$ defined
on the unit ball $\mathcal B_{\bf 0}^D(1)$ endowed with the Euclidean metric ${\bf I}$. The potential energy is defined as $U(\vect\theta) := -\log f(\vect\theta)$.
Associated with the auxiliary variable ${\bf v}$ (i.e., velocity), we define the kinetic energy
$K({\bf v})=\frac{1}{2}\tp{\bf v}{\bf I}{\bf v}$ for ${\bf v}\in T_{\vect\theta} \mathcal B_{\bf 0}^D(1)$, which is a $D$-dimensional vector sampled from the tangent space of $\mathcal B_{\bf 0}^D(1)$. Therefore, the Hamiltonian is defined on $\mathcal B_{\bf 0}^D(1)$ as 
\begin{equation}\label{HamiltonianI}
H(\vect\theta,{\bf v}) = U(\vect\theta) + K({\bf v}) = U(\vect\theta) + \frac{1}{2}\tp{\bf v}{\bf I}{\bf v} 
\end{equation}

Under the transformation $T_{\mathcal B\to
\mathcal S}$ in \eqref{b2s}, the above Hamiltonian \eqref{HamiltonianI} on
$\mathcal B_{\bf 0}^D(1)$ will be changed to the follwing Hamiltonian
$H(\tilde{\vect\theta}, \tilde{\bf v})$ on $\mathcal S^D$ as in \eqref{HamiltonianX}:
\begin{equation}\label{HamiltonianSc}
H(\tilde{\vect\theta}, \tilde{\bf v}) = \phi(\tilde{\vect\theta}) +
\frac{1}{2}\tp{\bf v} {\bf G}_{\mathcal S_c}(\vect\theta){\bf v} = U(\tilde{\vect\theta}) - \log\left|\frac{d\vect\theta_{\mathcal B}}{d\vect\theta_{\mathcal S_c}}\right| + \frac{1}{2}\tp{\bf v} {\bf G}_{\mathcal S_c}(\vect\theta){\bf v}
\end{equation}
where the potential energy $U(\tilde{\vect\theta})=U(\vect\theta)$ (i.e., the
distribution is fully defined in terms of the original parameter $\vect\theta$,
which are the first $D$ elements of $\tilde{\vect\theta}$), and ${\bf G}_{\mathcal S_c}(\vect\theta)={\bf I}_D+\vect\theta\tp{\vect\theta}/(1-\Vert\vect\theta\Vert_{2}^2)$ is the \emph{canonical spherical metric}.

Viewing $\{\vect\theta, \mathcal B_{\bf 0}^D(1)\}$ as the
Euclidean coordinate chart of manifold $(\mathcal S^D,{\bf G}_{\mathcal
S_c}(\vect\theta))$, we have the logorithm of volume adjustment,
$\log\left|\frac{d\vect\theta_{\mathcal B}}{d\vect\theta_{\mathcal
S_c}}\right|=-\frac12\log|{\bf G}_{\mathcal S_c}|=\log|\theta_{D+1}|$ (See
Appendix \ref{Metc}). The last two terms in Equation \eqref{HamiltonianSc} is the minus log
density of ${\bf v}|\vect\theta\sim\mathcal N({\bf 0},{\bf G}_{\mathcal
S_c}(\vect\theta)^{-1})$ \citep[See][for more details]{girolami11,lan14a}.
However, the derivative of log volume adjustment, $\theta_{D+1}^{-1}$,
contributes an extremely large component to the gradient of energy around the
equator ($\theta_{D+1}=0$), which in turn increases the numerical error in the
discretized Hamiltonian dynamics. For the purpose of numerical stability, we
instead consider the following \emph{partial} Hamiltonian $H^*(\tilde{\vect\theta}, \tilde{\bf
v})$ and leave the volume adjustment as weights to adjust the estimation of
integration \eqref{domainX}:
\begin{equation}\label{hamiltonianSc}
H^*(\tilde{\vect\theta}, \tilde{\bf v}) = U(\vect\theta) + \frac{1}{2}\tp{\bf v} {\bf G}_{\mathcal S_c}(\vect\theta){\bf v}
\end{equation}

If we extend the velocity as $\tilde{\bf v} = ({\bf v}, v_{D+1})$ with
$v_{D+1}=-\tp{\vect\theta}{\bf v}/\theta_{D+1}$, then $\tilde{\bf v}$ falls in
the tangent space of the sphere, $T_{\tilde{\vect\theta}}\mathcal S^D:=\{\tilde{\bf v}\in \mathbb R^{D+1}|\tp{\tilde{\vect\theta}} \tilde{\bf v}=0\}$.
Therefore, $\tp{\bf v} {\bf G}_{\mathcal S_c}(\vect\theta){\bf v}=\tp{\tilde{\bf
v}}\tilde{\bf v}$.
As a result, the partial Hamiltonian \eqref{hamiltonianSc} can be recognized as the
standard Hamiltonian \eqref{HamiltonianI} in the augmented $(D+1)$ dimensional space
\begin{equation}\label{augHamiltonianIc}
H^*(\tilde{\vect\theta}, \tilde{\bf v}) = U(\tilde{\vect\theta}) + K(\tilde{\bf v})
= U(\tilde{\vect\theta}) +\frac12 \tp{\tilde{\bf v}}\tilde{\bf v}
\end{equation}
This is due to the energy invariance presented as Proposition \ref{enginv} in Appendix \ref{geomS}. Now we can sample the velocity ${\bf v}\sim \mathcal
N({\bf 0},{\bf G}_{\mathcal S_c}(\vect\theta)^{-1})$ and set $\tilde{\bf v}=\begin{bmatrix}{\bf I}
\\-\tp{\vect\theta} /\theta_{D+1}\end{bmatrix}{\bf v}$. Alternatively, since $\Co[\tilde{\bf v}] = \begin{bmatrix}{\bf I} \\-\tp{\vect\theta}
/\theta_{D+1}\end{bmatrix} {\bf G}_{\mathcal S_c}(\vect\theta)^{-1}
\begin{bmatrix}{\bf I} -\vect\theta/\theta_{D+1}\end{bmatrix} = {\bf I}_{D+1} -
\tilde{\vect\theta} \tp{\tilde{\vect\theta}}$ is idempotent, we can sample
$\tilde{\bf v}$ by $({\bf I}_{D+1} - \tilde{\vect\theta}
\tp{\tilde{\vect\theta}}){\bf z}$ with ${\bf z}\sim\mathcal N({\bf 0},{\bf I}_{D+1})$.

The Hamiltonian function \eqref{hamiltonianSc} can be used to define the Hamiltonian dynamics on the Riemannian manifold $(\mathcal S^D, {\bf G}_{\mathcal S_c}(\vect\theta))$ in terms of $(\vect\theta, {\bf p})$, or equivalently as the following Lagrangian dynamics in terms of $(\vect\theta, {\bf v})$ \citep{lan14a}:
\begin{align}\label{sphLD}
\begin{aligned}
&\dot{\vect\theta} && = && {\bf v}\\
&\dot{\bf v} && = && -\tp{\bf v} \vect\Gamma_{\mathcal S_c}(\vect\theta){\bf v} - {\bf G}_{\mathcal S_c}(\vect\theta)^{-1} \nabla_{\vect\theta} U(\vect\theta)
\end{aligned}
\end{align}
where $\vect\Gamma_{\mathcal S_c}(\vect\theta)$ are the Christoffel symbols of second kind derived from ${\bf G}_{\mathcal S_c}(\vect\theta)$.
The Hamiltonian \eqref{hamiltonianSc} is preserved under Lagrangian dynamics \eqref{sphLD}. \citep[See][for more discussion]{lan14a}.

\cite{byrne13} split the Hamiltonian \eqref{hamiltonianSc} as follows:
\begin{equation}
H^*(\tilde{\vect\theta}, \tilde{\bf v}) = U(\vect\theta)/2 + \frac{1}{2}\tp{\bf v} {\bf G}_{\mathcal S_c}(\vect\theta){\bf v} + U(\vect\theta)/2
\end{equation}
However, their approach requires the manifold to be embedded in the Euclidean
space. To avoid this assumption, instead of splitting the Hamiltonian dynamics of $(\vect\theta, {\bf p})$, we
split the corresponding Lagrangian dynamics \eqref{sphLD} in terms of $(\vect\theta, {\bf v})$ as follows (See Appendix \ref{splitHS} for more details):

\noindent
\begin{subequations}
\begin{minipage}{.5\textwidth}
\begin{align}\label{sphLD:U}
\begin{cases}
\begin{aligned}
&\dot{\vect\theta} && = && {\bf 0}\\
&\dot{\bf v} && = && -\frac{1}{2} {\bf G}_{\mathcal S_c}(\vect\theta)^{-1} \nabla_{\vect\theta} U(\vect\theta)
\end{aligned}
\end{cases}
\end{align}
\end{minipage}
\begin{minipage}{.5\textwidth}
\begin{align}\label{sphLD:K}
\begin{cases}
\begin{aligned}
&\dot{\vect\theta} && = && {\bf v}\\
&\dot{\bf v} && = && -\tp{\bf v}\vect\Gamma_{\mathcal S_c}(\vect\theta){\bf v}
\end{aligned}
\end{cases}
\end{align}
\end{minipage}
\end{subequations}

Note that the first dynamics \eqref{sphLD:U} only involves updating velocity
$\tilde{\bf v}$ in the tangent space $T_{\tilde{\vect\theta}}\mathcal S^D$ and
has the following solution (see Appendix \ref{splitHS} for more details):
\begin{align}\label{evolv}
\begin{aligned}
&\tilde{\vect\theta}(t) && = && \tilde{\vect\theta}(0)\\
&\tilde{\bf v}(t) && = && \tilde{\bf v}(0) -\frac{t}{2}\left(\begin{bmatrix}{\bf I}_D\\ \tp{\bf 0}\end{bmatrix} -\tilde{\vect\theta}(0) \tp{\vect\theta(0)}\right) \nabla_{\vect\theta} U(\vect\theta(0))
\end{aligned}
\end{align}
The second dynamics \eqref{sphLD:K} only involves the kinetic energy and has the
geodesic flow that is a \emph{great circle} (orthodrome or Riemannian circle) on
the sphere $\mathcal S^D$ as its analytical solution (See Appendix \ref{GEODS}
for more details):
\begin{align}\label{gcir}
\begin{aligned}
&\tilde{\vect\theta}(t) && = && \tilde{\vect\theta}(0) \cos(\Vert \tilde{\bf v}(0)\Vert_{2} t) + \frac{\tilde{\bf v}(0)}{\Vert \tilde{\bf v}(0)\Vert_{2}} \sin(\Vert \tilde{\bf v}(0)\Vert_{2} t)\\
&\tilde{\bf v}(t) && = && -\tilde{\vect\theta}(0) \Vert \tilde{\bf v}(0)\Vert_{2} \sin(\Vert \tilde{\bf v}(0)\Vert_{2} t) + \tilde{\bf v}(0) \cos(\Vert \tilde{\bf v}(0)\Vert_{2} t)
\end{aligned}
\end{align}
This solution defines an evolution, denoted as $g_t:(\vect\theta(0),{\bf v}(0))\mapsto(\vect\theta(t),{\bf v}(t))$.
Both \eqref{evolv} and \eqref{gcir} are symplectic. Due to the explicit
formula for the geodesic flow on sphere, the second dynamics in \eqref{sphLD:K} is simulated exactly. Therefore, updating $\tilde{\vect\theta}$ does not involve discretization error so we can use large step sizes. This could lead to improved computational efficiency. 
%Since this step is in fact a rotation on sphere, we set the trajectory length to be $2\pi/D$ and randomize the number of leapfrog steps to avoid periodicity. 
Because this step is in fact a rotation on sphere, it can generate proposals that are far away from the current state. 
Algorithm \ref{Alg:cSphHMC} shows the steps for implementing this approach, henceforth called \emph{Spherical HMC in the Cartesian coordinate (c-SphHMC)}.
It can be shown that the integrator in the algorithm has order 3 local error and order 2 global error (See the details in Appendix \ref{Err}).

%\begin{tiny}
\begin{algorithm}[t]
\caption{Spherical HMC in the Cartesian coordinate (c-SphHMC)}
\label{Alg:cSphHMC}
\begin{algorithmic}
\STATE Initialize $\tilde{\vect\theta}^{(1)}$ at current $\tilde{\vect\theta}$
after transformation $T_{\mathcal D\to\mathcal S}$
\STATE Sample a new velocity value $\tilde{\bf v}^{(1)}\sim \mathcal N({\bf 0},{\bf I}_{D+1})$
\STATE Set $\tilde{\bf v}^{(1)} \leftarrow \tilde{\bf v}^{(1)} - \tilde{\vect\theta}^{(1)} \tp{(\tilde{\vect\theta}^{(1)})} \tilde{\bf v}^{(1)}$
\STATE Calculate $H(\tilde{\vect\theta}^{(1)},\tilde{\bf v}^{(1)})=U(\vect\theta^{(1)}) + K(\tilde{\bf v}^{(1)})$ 
\FOR{$\ell=1$ to $L$}
\STATE $\tilde{\bf v}^{(\ell+\frac{1}{2})} = \tilde{\bf v}^{(\ell)}-\frac{\eps}{2} \left(\begin{bmatrix}{\bf I}_D\\ \tp{\bf 0}\end{bmatrix} -\tilde{\vect\theta}^{(\ell)} \tp{(\vect\theta^{(\ell)})}\right) \nabla_{\vect\theta} U(\vect\theta^{(\ell)})$
\STATE $\tilde{\vect\theta}^{(\ell+1)} = \tilde{\vect\theta}^{(\ell)} \cos(\Vert \tilde{\bf v}^{(\ell+\frac{1}{2})}\Vert \eps) + \frac{\tilde{\bf v}^{(\ell+\frac{1}{2})}}{\Vert \tilde{\bf v}^{(\ell+\frac{1}{2})}\Vert} \sin(\Vert \tilde{\bf v}^{(\ell+\frac{1}{2})}\Vert \eps)$
\STATE $\tilde{\bf v}^{(\ell+\frac{1}{2})} \leftarrow -\tilde{\vect\theta}^{(\ell)}\Vert \tilde{\bf v}^{(\ell+\frac{1}{2})}\Vert \sin(\Vert \tilde{\bf v}^{(\ell+\frac{1}{2})}\Vert \eps) + \tilde{\bf v}^{(\ell+\frac{1}{2})} \cos(\Vert \tilde{\bf v}^{(\ell+\frac{1}{2})}\Vert \eps)$
\STATE $\tilde{\bf v}^{(\ell+1)} = \tilde{\bf v}^{(\ell+\frac{1}{2})}-\frac{\eps}{2} \left(\begin{bmatrix}{\bf I}_D\\ \tp{\bf 0}\end{bmatrix} -\tilde{\vect\theta}^{(\ell+1)} \tp{(\vect\theta^{(\ell+1)})}\right) \nabla_{\vect\theta} U(\vect\theta^{(\ell+1)})$
\ENDFOR
\STATE Calculate $H(\tilde{\vect\theta}^{(L+1)},\tilde{\bf v}^{(L+1)})=U(\vect\theta^{(L+1)}) + K(\tilde{\bf v}^{(L+1)})$
\STATE Calculate the acceptance probability $\alpha =\min\{1, \exp[-H(\tilde{\vect\theta}^{(L+1)},\tilde{\bf v}^{(L+1)})+H(\tilde{\vect\theta}^{(1)},\tilde{\bf v}^{(1)})] \}$
\STATE Accept or reject the proposal according to $\alpha$ for the next state $\tilde{\vect\theta}'$
\STATE Calculate $T_{\mathcal S\to\mathcal D}(\tilde{\vect\theta}')$ and the
corresponding weight $|dT_{\mathcal S\to\mathcal D}|$
\end{algorithmic}
\end{algorithm}
%\end{tiny}

\subsubsection{Spherical HMC in the spherical coordinate}\label{SphHMCs}
Now we define HMC on the sphere $\mathcal S^D$ in the spherical coordinate $\{\vect\theta, \mathcal R_{\bf 0}^D\}$.
The natural metric on the sphere $\mathcal S^D$ induced by the coordinate
mapping \eqref{r2s} is the \emph{round spherical metric}\footnote{Note, $\tp{\bf
v}{\bf G}_{\mathcal S_r}(\vect\theta){\bf v}\leq \Vert{\bf v}\Vert_2^2\leq \Vert\tilde{\bf v}\Vert_2^2 = \tp{\bf v}{\bf G}_{\mathcal S_c}(\vect\theta){\bf v}$.}, ${\bf G}_{\mathcal S_r}(\vect\theta)=\diag[1,\sin^2(\theta_1),\cdots,\prod_{d=1}^{D-1}\sin^2(\theta_d)]$.

As in Section \ref{SphHMCc}, we start with the usual Hamiltonian
$H(\vect\theta,{\bf v})$ defined on $(\mathcal R_{\bf 0}^D,{\bf I})$ as in
\eqref{HamiltonianI} with ${\bf v}\in T_{\vect\theta} \mathcal R_{\bf 0}^D$.
Under the transformation $T_{\mathcal R_{\bf 0}\to\mathcal S}: \vect\theta\mapsto{\bf
x}$ in \eqref{r2s}, Hamiltonian \eqref{HamiltonianI} on $\mathcal R_{\bf 0}^D$
is changed to the following Hamiltonian $H({\bf x},\dot{\bf x})$ as in
\eqref{HamiltonianX}:
\begin{equation}\label{HamiltonianSs}
H({\bf x},\dot{\bf x}) = \phi({\bf x}) +
\frac{1}{2}\tp{\bf v} {\bf G}_{\mathcal S_c}(\vect\theta){\bf v} = U({\bf x}) - \log\left|\frac{d\vect\theta_{\mathcal R_{\bf 0}}}{d\vect\theta_{\mathcal S_r}}\right| + \frac{1}{2}\tp{\bf v} {\bf G}_{\mathcal S_r}(\vect\theta){\bf v}
\end{equation}
where the potential energy $U({\bf x})=U(\vect\theta({\bf x}))$ and ${\bf
G}_{\mathcal S_r}(\vect\theta)$ is the round spherical metric.

As before, the logorithm of volume adjustment is
$\log\left|\frac{d\vect\theta_{\mathcal R_{\bf 0}}}{d\vect\theta_{\mathcal
S_r}}\right|=-\frac12\log|{\bf G}_{\mathcal
S_r}|=-\sum_{d=1}^{D-1}(D-d)\log\sin(\theta_{d})$ (See Appendix \ref{Mets}).
The last two terms in Equation \eqref{HamiltonianSs} is the minus log
density of ${\bf v}|\vect\theta\sim\mathcal N({\bf 0},{\bf G}_{\mathcal
S_r}(\vect\theta)^{-1})$. Again, for numerical stability we
consider the following partial Hamiltonian $H^*({\bf x},\dot{\bf x})$ and leave the
volume adjustment as weights to adjust the estimation of integration \eqref{domainX}:
\begin{equation}\label{hamiltonianSs}
H^*({\bf x},\dot{\bf x}) = U(\vect\theta) + \frac{1}{2} \tp{\bf v} {\bf G}_{\mathcal S_r}(\vect\theta){\bf v}
\end{equation}

Taking derivative of $T_{\mathcal R_{\bf 0}\to \mathcal S}:
\vect\theta\mapsto{\bf x}$ in \eqref{r2s} with respect to time $t$ we have
\begin{equation}\label{dr2s}
\dot x_d =
\begin{dcases}
%\frac{d}{dt}\left[\cos\theta_d\prod_{i=1}^{d-1}\sin(\theta_i)\right]=
%-v_d\prod_{i=1}^d\sin(\theta_i) + \cos\theta_d\sum_{i<d}v_i\cot(\theta_i)\prod_{i=1}^{d-1}\sin(\theta_i)=
[-v_d\tan(\theta_d) + \sum_{i<d}v_i\cot(\theta_i)] x_d, & d< D+1\\
%\sum_{i<D+1}v_i\cot(\theta_i)\prod_{i=1}^D\sin(\theta_i)=
\sum_{i<D+1}v_i\cot(\theta_i) x_{D+1}, & d=D+1
\end{dcases}
\end{equation}
We can show that $\tp{{\bf x}(\vect\theta)}\dot{\bf x}(\vect\theta,{\bf v})=0$;
that is, $\dot{\bf x}\in T_{\bf x} \mathcal S^D$.
Taking derivative of $T_{\mathcal S \to \mathcal R_{\bf 0}}:
{\bf x}\mapsto\vect\theta$ in \eqref{s2r} with respect to time $t$ yields
\begin{equation}\label{ds2r}
v_d := \dot\theta_d = 
\begin{dcases}
%-\cot(\theta_d)\left[\frac{\dot x_d}{x_d}+\cos^2(\theta_d)\sum_{i=1}^{d-1}\frac{x_i\dot x_i}{x_d^2}\right],&d<D
-\frac{x_d}{\sqrt{1-\sum_{i=1}^d x_i^2}}\left[\frac{\dot x_d}{x_d}+\frac{\sum_{i=1}^{d-1}x_i\dot x_i}{1-\sum_{i=1}^{d-1}x_d^2}\right],&d<D
\\ \frac{x_D\dot x_{D+1}-\dot x_Dx_{D+1}}{x_D^2+x_{D+1}^2},&d=D
\end{dcases}
\end{equation}
Further, we have $\tp{\bf v} {\bf G}_{\mathcal S_r}(\vect\theta){\bf
v}=\tp{\dot{\bf x}}\dot{\bf x}$. Therefore, the partial Hamiltonian
\eqref{hamiltonianSs} can be recognized as the standard Hamiltonian
\eqref{HamiltonianI} in the augmented $(D+1)$ dimensional space, which is
again explained by the energy invariance Proposition \ref{enginv} (See more details in Appendix
\ref{geomS})
\begin{equation}\label{augHamiltonianIs}
H^*({\bf x},\dot{\bf x}) = U({\bf x}) + K(\dot{\bf x})
= U({\bf x}) +\frac12 \tp{\dot{\bf x}}\dot{\bf x}
\end{equation}

Similar to the method discussed in Section \ref{SphHMCc}, we split the Hamiltonian
\eqref{hamiltonianSs}, $H^*(\tilde{\vect\theta}, \tilde{\bf v}) = U(\vect\theta)/2 + \frac{1}{2}\tp{\bf v} {\bf G}_{\mathcal S_r}(\vect\theta){\bf v} + U(\vect\theta)/2$, and its corresponding Lagrangian dynamics \eqref{sphLD} as follows:\\
\begin{subequations}
\begin{minipage}{.5\textwidth}
\begin{align}\label{sphLD:1}
\begin{cases}
\begin{aligned}
&\dot{\vect\theta} && = && {\bf 0}\\
&\dot{\bf v} && = && -\frac{1}{2} {\bf G}_{\mathcal S_r}(\vect\theta)^{-1} \nabla_{\vect\theta} U(\vect\theta)
\end{aligned}
\end{cases}
\end{align}
\end{minipage}
\begin{minipage}{.5\textwidth}
\begin{align}\label{sphLD:2}
\begin{cases}
\begin{aligned}
&\dot{\vect\theta} && = && {\bf v}\\
&\dot{\bf v} && = && -\tp{\bf v}\vect\Gamma_{\mathcal S_r}(\vect\theta){\bf v}
\end{aligned}
\end{cases}
\end{align}
\end{minipage}
\end{subequations}

The first dynamics \eqref{sphLD:1} involves updating the velocity ${\bf v}$ only.
However, the diagonal term of ${\bf G}_{\mathcal S_r}(\vect\theta)^{-1}$,
$\prod_{i=1}^{d-1}\sin^{-2}(\theta_i)$ increases exponentially fast as dimension
grows. This will cause the velocity updated by \eqref{sphLD:1} to have extremely large components. To
avoid such issue, we use small time vector
$\vect\eps=[\eps,\eps^2,\cdots,\eps^D]$, instead of scalar $\eps$, in updating
Equation \eqref{sphLD:1}.
The second dynamics \eqref{sphLD:2} describes the same geodesic flow on the
sphere $\mathcal S^D$ as \eqref{sphLD:K} but in the spherical coordinate
$\{\vect\theta,\mathcal R_0^D\}$. Therefore it should have the same solution as
\eqref{gcir} expressed in $\{\vect\theta,\mathcal R_0^D\}$.
To obtain this solution, we first apply
$\tilde T_{\mathcal R_{\bf 0}\to \mathcal S}: (\vect\theta(0),{\bf v}(0))\mapsto ({\bf
x}(0),\dot{\bf x}(0))$, which consists of \eqref{r2s}\eqref{dr2s}.
Then, we use $g_t$ in \eqref{gcir} to evolve $({\bf x}(0),\dot{\bf x}(0))$ for
some time $t$ to find $({\bf x}(t),\dot{\bf x}(t))$. Finally, we use $\tilde T_{\mathcal S\to
\mathcal R_{\bf 0}}: ({\bf x}(t),\dot{\bf x}(t))\mapsto(\vect\theta(t),{\bf
v}(t))$, composite of \eqref{s2r}\eqref{ds2r}, to go back to $\mathcal R_0^D$.

%\begin{tiny}
\begin{algorithm}[t]
\caption{Spherical HMC in the spherical coordinate (s-SphHMC)}
\label{Alg:sSphHMC}
\begin{algorithmic}
\STATE Initialize $\vect\theta^{(1)}$ at current $\vect\theta$ after
transformation $T_{\mathcal D\to\mathcal S}$
\STATE Sample a new velocity value ${\bf v}^{(1)}\sim \mathcal N({\bf 0},{\bf I}_D)$
\STATE Set $v_d^{(1)} \leftarrow v_d^{(1)}\prod_{i=1}^{d-1}\sin^{-1}(\theta_i^{(1)})$,\;$d=1,\cdots,D$\
\STATE Calculate $H(\vect\theta^{(1)},{\bf v}^{(1)})=U(\vect\theta^{(1)}) + K({\bf v}^{(1)})$ 
\FOR{$\ell=1$ to $L$}
\STATE $v_d^{(\ell+\frac{1}{2})} = v_d^{(\ell)}-\frac{\eps^d}{2}
\frac{\pa}{\pa\theta_d} U(\vect\theta^{(\ell)}) \prod_{i=1}^{d-1}\sin^{-2}(\theta_i^{(\ell)})$,\;$d=1,\cdots,D$
\STATE $({\vect\theta}^{(\ell+1)},{\bf v}^{(\ell+\frac{1}{2})}) \leftarrow \tilde T_{\mathcal S\to \mathcal R_{\bf 0}}\circ g_{\eps}\circ\tilde T_{\mathcal R_{\bf 0}\to \mathcal S} ({\vect\theta}^{(\ell)},{\bf v}^{(\ell+\frac{1}{2})})$
\STATE $v_d^{(\ell+1)} = v_d^{(\ell+\frac{1}{2})}-\frac{\eps^d}{2} \frac{\pa}{\pa\theta_d} U(\vect\theta^{(\ell+1)}) \prod_{i=1}^{d-1}\sin^{-2}(\theta_i^{(\ell+1)})$,\;$d=1,\cdots,D$
\ENDFOR
\STATE Calculate $H(\vect\theta^{(L+1)},{\bf v}^{(L+1)})=U(\vect\theta^{(L+1)}) + K({\bf v}^{(L+1)})$
\STATE Calculate the acceptance probability $\alpha =\min\{1, \exp[-H(\vect\theta^{(L+1)},{\bf v}^{(L+1)})+H(\vect\theta^{(1)},{\bf v}^{(1)})] \}$
\STATE Accept or reject the proposal according to $\alpha$ for the next state ${\vect\theta}'$
\STATE Calculate $T_{\mathcal S\to\mathcal D}(\vect\theta')$ and the
corresponding weight $|dT_{\mathcal S\to\mathcal D}|$
\end{algorithmic}
\end{algorithm}
%\end{tiny}

Algorithm \ref{Alg:sSphHMC} summarizes the steps for this method, called \emph{Spherical HMC in the spherical
coordinate (s-SphHMC)}. In theory, the hyper-rectangle $\mathcal R_{\bf 0}^D$ can be used as a base type 
(as the unit ball $\mathcal B_{\bf 0}^D(1)$ does) for general $q$-norm constraints
for which s-SphHMC can be applied. This is because $q$-norm domain $\mathcal Q^D$ can be
bijectively mapped to the hypercube $\mathcal C^D$, and thereafter to $\mathcal R_{\bf 0}^D$. 
However the involved Jacobian matrix is rather complicated and s-SphHMC used in this way is not as efficient as c-SphHMC. Therefore, we use s-SphHMC only for box type constraints.

\subsection{Spherical LMC on probability simplex}\label{SphLMC}
A large class of statistical models involve defining probability distributions on the \emph{simplex} $\Delta^K$,
\begin{equation}\label{probsplx}
\Delta^K:=\{\vect\pi\in\mathbb R^D|\;\pi_k\geq 0, \sum_{k=1}^K \pi_d=1\}
\end{equation}
As an example, we consider \emph{latent Dirichlet allocation (LDA)} \citep{blei03}, which is a
hierarchical Bayesian model commonly used to model document topics.
This type of constraints can be viewed as a special case of the 1-norm constraint, discussed in Section \ref{q-reg}, by identifying the first orthant (all positive components) with the others. Then, the c-SphHMC algorithm \ref{Alg:cSphHMC} can be applied to generate samples
$\{\vect\theta\}$ on the sphere $\mathcal S^{K-1}$. These samples can be transformed as
$\{\vect\theta^2\}$ and mapped back to the simplex $\Delta^K$. %Since the main application considered here, i.e., LDA, is high dimensional and requires some online learning process, we need to modify c-SphHMC for LDA.

In what follows, we show that Fisher metric on the root space of simplex,
$\sqrt{\Delta}^K:=\{\vect\theta\in\mathcal S^{K-1}|\theta_k\geq 0,\,\forall\, k=1,\cdots,K\}$ 
(i.e. the first orthant of the sphere $\mathcal S^{K-1}$),
is the same as the canonical spherical metric ${\bf G}_{\mathcal S_c}(\vect\theta)$ up to a constant.
In this sense, it is more natural to define the sampling algorithms on the sphere
$\mathcal S^{K-1}$. We start with the toy example discussed in \cite{patterson13}. Denote the observed data as ${\bf x}=\{x_i\}_{i=1}^N$, where each data point belongs to one of the $K$ categories with probability $p(x_i=k|\vect\pi)=\pi_k$. We assume a Dirichlet prior on $\vect\pi$:
$p(\vect\pi)\propto \prod_{k=1}^K \pi_k^{\alpha_k-1}$. The posterior distribution is $p(\vect\pi|{\bf x})\propto\prod_{k=1}^K \pi_k^{n_k+\alpha_k-1}$, where
$n_k=\sum_{i=1}^NI(x_i=k)$ counts the points $x_i$ in category $k$.
Denote ${\bf n}=\tp{[n_1,\cdots,n_K]}$ and $n:=|{\bf n}|=\sum_{k=1}^Kn_k$.
For inference, we need to sample from the posterior distribution $p(\vect\pi|{\bf x})$ defined on the probability
simplex.

The Fisher information matrix is a function of $\vect\pi_{-K}$ (here, `$-K$' means all but the $K$-th components) and is calculated as follows:
\begin{align}
\begin{aligned}
{\bf G_F}(\vect\pi_{-K}) &= -\E[\nabla^2\log p({\bf x}|\vect\pi_{-K})]\\
&= -\E[\nabla^2 (\tp{\bf n}_{-K}\log(\vect\pi_{-K}) + (n-\tp{\bf n}_{-K}{\bf 1})\log(1-\tp{\vect\pi}_{-K}{\bf 1}))]\\
&= -\E[\nabla ({\bf n}_{-K}/\vect\pi_{-K} - {\bf 1}(n-\tp{\bf n}_{-K}{\bf 1})/(1-\tp{\vect\pi}_{-K}{\bf 1}))]\\
&= -\E[ -\diag({\bf n}_{-K}/\vect\pi_{-K}^2) - {\bf 1}\tp{\bf 1}(n-\tp{\bf n}_{-K}{\bf 1})/(1-\tp{\vect\pi}_{-K}{\bf 1})^2]\\
&= n[\diag(1/\vect\pi_{-K}) +{\bf 1}\tp{\bf 1}/\pi_K]
\end{aligned}
\end{align}
Now we use $T_{\Delta\to\sqrt{\Delta}}:\vect\pi\mapsto \vect\theta=\sqrt{\vect\pi}$ to
map the simplex to the sphere (the first orthant). Note that $\frac{d\vect\pi_{-K}}{d\tp{\vect\theta}_{-K}}=2\diag(\vect\theta_{-K})$. Therefore, we have a proper metric on $\sqrt{\Delta}^K$ as follows:
\begin{equation}
{\bf G}_{\sqrt\Delta}(\vect\theta) 
%&=\E[\nabla_{\vect\theta_{-K}}\log p({\bf x}|\vect\theta_{-K}) \tp{\nabla_{\vect\theta_{-K}}\log p({\bf x}|\vect\theta_{-K})}]&\\
= \frac{d\tp{\vect\pi}_{-K}}{d\vect\theta_{-K}} {\bf G_F}(\vect\pi_{-K}) \frac{d\vect\pi_{-K}}{d\tp{\vect\theta}_{-K}}
%\\&= 4n\diag(\vect\theta_{-K}) [\diag(1/\vect\pi_{-K}) +{\bf 1}\tp{\bf 1}/\pi_K] \diag(\vect\theta_{-K})&\\
= 4n[{\bf I}_{K-1} +\vect\theta_{-K}\tp{\vect\theta_{-K}}/\theta_K^2]
= 4n{\bf G}_{\mathcal S_c}(\vect\theta)
\end{equation}
where the scalar $4n$ properly scales the metric in high dimensional data
intensive models. In LDA particularly, $n$ could be the number of words counted
in the selected documents. Hence, we use ${\bf G}_{\sqrt\Delta}(\vect\theta)$
instead of ${\bf G}_{\mathcal S_c}(\vect\theta)$. We refer to the resulting method as \emph{Spherical Lagrangian Monte Carlo (SphLMC)}.

\begin{figure}[t]
\begin{center}
\includegraphics[width=.45\textwidth,height=0.4\textwidth]{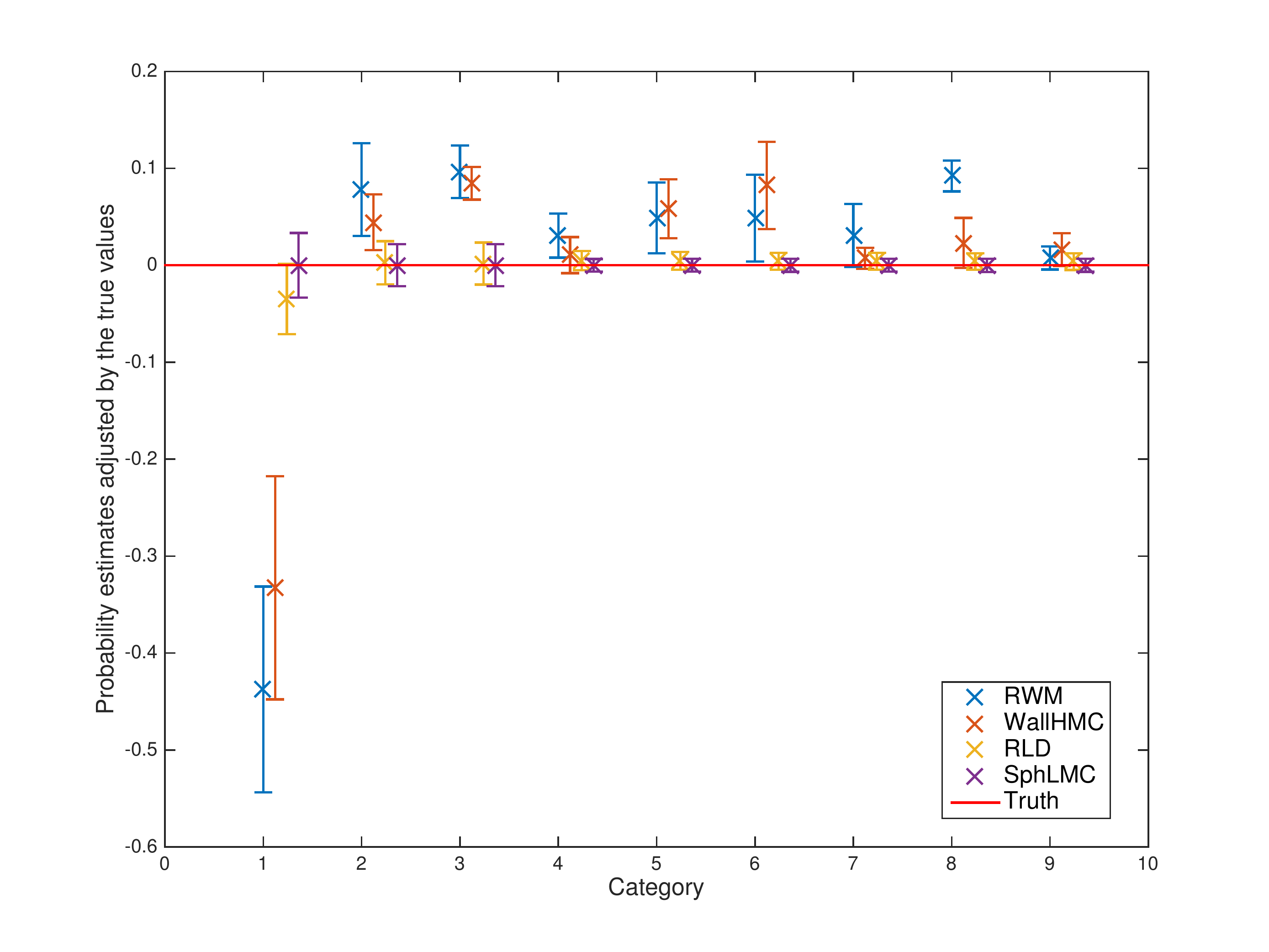}
\includegraphics[width=.54\textwidth,height=0.4\textwidth]{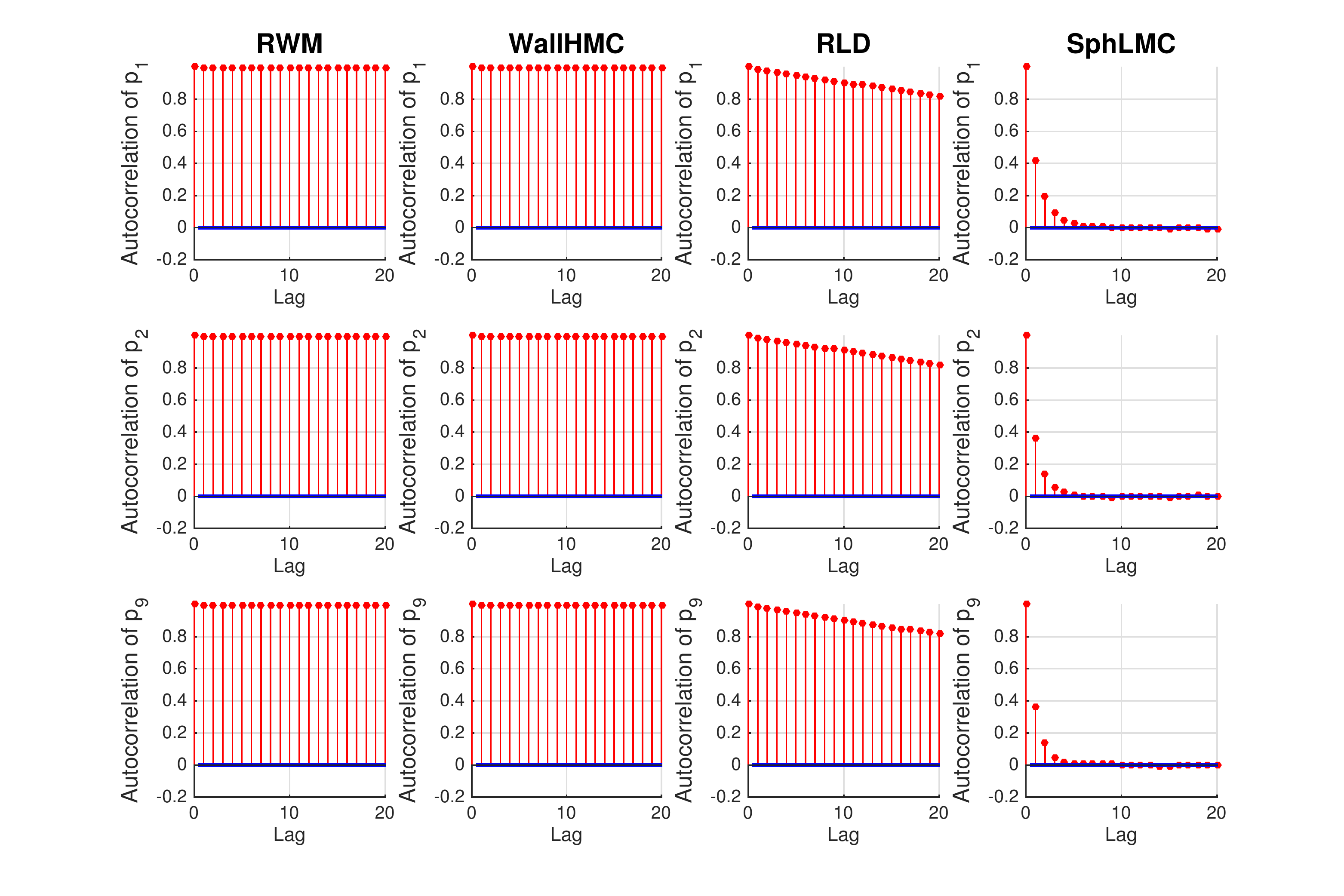}
\caption[Dirichlet-Multinomial distribution]{Dirichlet-Multinomial model: probability estimates (left) and autocorrelation function for MCMC samples (right).}
\label{fig:D-M}
\end{center}
\end{figure}

Recall that in the development of Spherical HMC algorithms, we decided to omit the log volume adjustment term, 
$\log\left|\frac{d\vect\beta_{\mathcal D}}{d\vect\theta_{\mathcal S}}\right|$, in the partial Hamiltonian
\eqref{hamiltonianSc} and \eqref{hamiltonianSs}, and regard it as the weight to
adjust the estimate of \eqref{domainX} or resample. This is not feasible if the LDA model is going to be used in an online setting. Therefore, we use $\phi(\vect\theta)$ in \eqref{potentialX}, as opposed to $U(\vect\theta)$ to avoid the re-weighting step. 

To illustrate our proposed method, we consider the toy example discussed above. For this problem, \cite{patterson13} propose a Riemannian Langevin Dynamics (RLD) method, but use an expanded-mean parametrization to map the simplex to the whole space. As mentioned above, this approach (i.e., expanding the parameter space) might not be efficient in general. This is illustrated in Figure \ref{fig:D-M}. Here, we set
$\alpha=0.5$ and run RMW, WallHMC, RLD, and SphLMC\footnote{Note, the natural gradient in \eqref{evolv} to
update $\tilde{\bf v}$ is $\begin{bmatrix}{\bf I}_{K-1}\\-\tp{\vect\theta}_{-K}/\theta_K\end{bmatrix}{\bf G}_{\sqrt\Delta}(\vect\theta)^{-1}\nabla_{\vect\theta_{-K}}\phi(\vect\theta_{-K})
=[({\bf n}+\alpha-0.5)/\vect\theta-\vect\theta*|{\bf n}+\alpha-0.5|]/(2n)$.} 
for $1.1\times 10^5$ iterations; we discard the first $10^4$ samples. As we can see in Figure \ref{fig:D-M},
compared to alternative algorithms, our SphLMC method provides better probability estimates (left panel). Further, SphLMC generates samples with a substantially lower autocorrelation (right panel).

%In the above example\footnote{Note, the natural gradient in \eqref{evolv} to
%update $\tilde{\bf v}$ is
%$\begin{bmatrix}{\bf I}_{K-1}\\-\tp{\vect\theta}_{-K}/\theta_K\end{bmatrix}{\bf G}_{\sqrt\Delta}(\vect\theta)^{-1}\nabla_{\vect\theta_{-K}}\phi(\vect\theta_{-K})
%=[({\bf n}+\alpha-0.5)/\vect\theta-\vect\theta*|{\bf n}+\alpha-0.5|]/(2n)$.}, we set
%$\alpha=0.5$ and run RMWt, WallHMC, RLD, and SphLMC for $1.1\times 10^5$
%iterations and burn in the first $10^4$. Unlike \cite{patterson13}, we do not thin samples. 
%From the results summarized in Figure \ref{fig:D-M},
%one can tell that RWM and WallHMC may not converge as they don't give correct
%estimates. All the algorithms but SphLMC give highly autocorrelated samples
%which yield estimates with large variance.

% We might want to move this to Discussion.
%We close the section with some comment on the limitation of SphLMC. More specifically, some
%choice of the prior ($\mathrm{Dir}(\alpha)$ with $\alpha<1$), e.g. $\alpha=0.1$,
%may result in a very high peak near the boundary of the simplex. This
%sometimes creates difficulty for SphLMC to explore the peak (mode) region
%because the gradient of energy may contain some extremely large component near the
%boundary. Such `boundary peak' issue could be alleviated by large scaling factor
%$n$ in the metric ${\bf G}_{\sqrt\Delta}(\vect\theta)$.

%%%%%%%%%%%%%%%%%%%%%%%%%%%%%%%%%%%%%%%%%%%%%%%%%%%%%%%%%%%%%%%%%%%%%%%%%%%%%%%%%%%%%

\section{Experimental results} \label{results}
In this section, we evaluate our proposed methods using simulated and real data. To this end, we compare their efficiency to that of RWM, Wall HMC, exact HMC \citep{pakman13}, and the Riemannian Langevin dynamics (RLD) algorithm proposed by \cite{patterson13} for LDA. We define efficiency in terms of time-normalized effective sample size (ESS). Given $N$ MCMC samples, for each parameter, we define $\mathrm{ESS}=N[1 + 2\Sigma_{k=1}^{K}\rho(k)]^{-1}$, where $\rho(k)$ is sample autocorrelation with lag $k$ \citep{geyer92}. We use the minimum ESS normalized by the CPU time, s (in seconds), as the overall measure of efficiency: $\min(\textrm{ESS})/\textrm{s}$. All computer codes are available online at \url{http://www.ics.uci.edu/~slan/SphHMC}.

\subsection{Truncated Multivariate Gaussian}
For illustration purpose, we start with a truncated bivariate Gaussian distribution,
\begin{equation*}
\binom{\beta_1}{\beta_2} \sim \mathcal N\left(\bf{0}, \begin{bmatrix} 1& .5\\ .5 & 1\end{bmatrix} \right), \qquad
0\leq \beta_1\leq 5, \quad 0\leq \beta_2\leq 1
\end{equation*}
This is box type constraint with the lower and upper limits as ${\bf l}=(0,0)$ and ${\bf u}=(5,1)$ respectively.
%The original rectangle domain can be mapped to the 2-dimensional unit sphere through the following transformation:
%\begin{equation*}
%T:\, [0,5]\times [0,1] \rightarrow \mathcal S^2, \quad \vect\beta  \mapsto \vect\beta' = (2\vect\beta-({\bf u}+{\bf l}))/({\bf u}-{\bf l}),\;
%\mapsto \vect\theta= \vect\beta' \frac{\Vert \vect\beta'\Vert_{\infty}}{\Vert \vect\beta'\Vert_2} \mapsto \tilde{\vect\theta} = \left(\vect\theta,\sqrt{1-\Vert
%\vect\theta\Vert_2^2}\right)
%\end{equation*}
The original rectangle domain can be mapped to 2d unit disc $\mathcal B_{\bf 0}^2(1)$ to use c-SphHMC, 
or mapped to 2d rectangle $\mathcal R_{\bf 0}^2$ where s-SphHMC can be directly applied.

\begin{figure}[t]
\begin{center}
\includegraphics[width=.9\textwidth,height=0.45\textwidth]{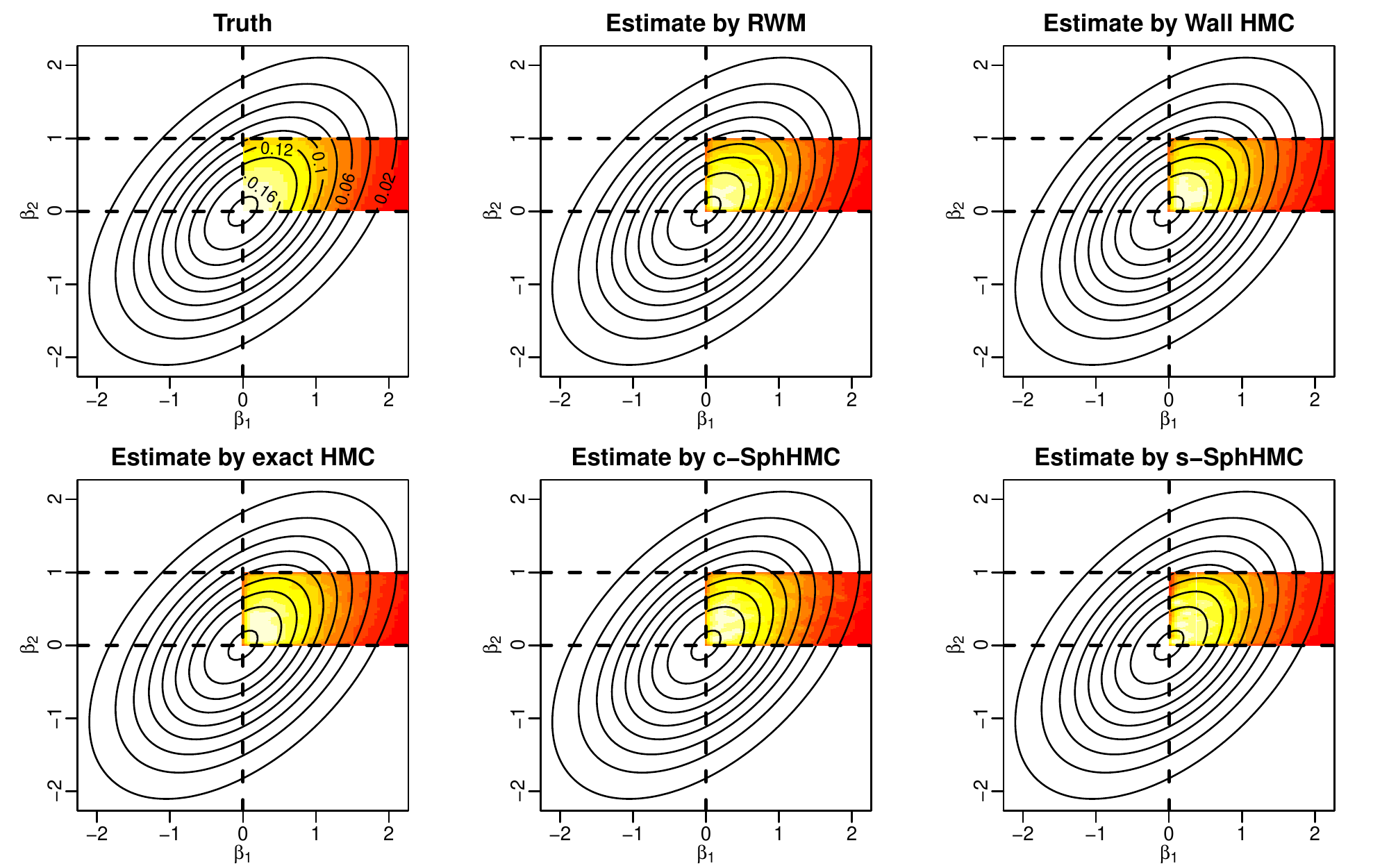}
\caption[Truncated Multivariate Gaussian]{Density plots of a truncated bivariate Gaussian using exact density function (upper leftmost) and MCMC samples from RWM, Wall HMC, exact HMC, c-SphHMC and s-SphHMC respectively.}
\label{fig:TMG-sph}
\end{center}
\end{figure}

The upper leftmost panel of Figure \ref{fig:TMG-sph} shows the heatmap based on the exact density function, and the other panels show the corresponding heatmaps based on MCMC samples from RWM, Wall HMC, exact HMC, c-SphHMC and s-SphHMC respectively. Table \ref{TMG-moments} compares the true mean and covariance of the above
truncated bivariate Gaussian distribution with the point estimates using $2\times 10^5$
($2\times 10^4$ for each of 10 repeated experiments with different random seeds) % comment it if you don't like it, and you can choose the alternative table as follows.
MCMC samples in each method. Overall, all methods estimate the mean and covariance reasonably well.

%%%% comment this table if you want to a simpler presentation below %%%%

% latex table generated in R 3.2.0 by xtable 1.7-4 package
% Wed May 20 18:05:22 2015
\begin{table}[ht]
\centering
\begin{tabular}{l|c|c}
  \hline
Method & Mean & Covariance \\ 
  \hline
Truth & $\begin{bmatrix}0.7906\\0.4889\end{bmatrix}$ & $\begin{bmatrix}0.3269&0.0172\\0.0172&0.08\end{bmatrix}$ \\ 
   \hline
RWM & $\begin{bmatrix}0.7796\pm0.0088\\0.4889\pm0.0034\end{bmatrix}$ & $\begin{bmatrix}0.3214\pm0.009&0.0158\pm0.001\\0.0158\pm0.001&0.0798\pm5e-04\end{bmatrix}$ \\ 
   \hline
Wall HMC & $\begin{bmatrix}0.7875\pm0.0049\\0.4884\pm8e-04\end{bmatrix}$ & $\begin{bmatrix}0.3242\pm0.0043&0.017\pm0.001\\0.017\pm0.001&0.08\pm3e-04\end{bmatrix}$ \\ 
   \hline
exact HMC & $\begin{bmatrix}0.7909\pm0.0025\\0.4885\pm0.001\end{bmatrix}$ & $\begin{bmatrix}0.3272\pm0.0026&0.0174\pm7e-04\\0.0174\pm7e-04&0.08\pm3e-04\end{bmatrix}$ \\ 
   \hline
c-SphHMC & $\begin{bmatrix}0.79\pm0.005\\0.4864\pm0.0016\end{bmatrix}$ & $\begin{bmatrix}0.3249\pm0.0045&0.0172\pm0.0012\\0.0172\pm0.0012&0.0801\pm0.001\end{bmatrix}$ \\ 
   \hline
s-SphHMC & $\begin{bmatrix}0.7935\pm0.0093\\0.4852\pm0.003\end{bmatrix}$ & $\begin{bmatrix}0.3233\pm0.0062&0.0202\pm0.0018\\0.0202\pm0.0018&0.0791\pm9e-04\end{bmatrix}$ \\ 
   \hline
\end{tabular}
\caption{Comparing the point estimates for the mean and covariance of a bivariate truncated Gaussian distribution using RWM, Wall HMC, exact HMC, c-SphHMC and s-SphHMC.} 
\label{TMG-moments}
\end{table}

%%%% alternative presentation of the table, please uncomment %%%%

%\begin{table}[ht]
%\centering
%\begin{tabular}{l|c|c}
%  \hline
%Method & Mean & Covariance \\ 
%  \hline
%Truth & $\begin{pmatrix}0.7906\\0.4889\end{pmatrix}$ & $\begin{pmatrix}0.3269&0.0172\\0.0172&0.08\end{pmatrix}$ \\ [9pt]
%   \hline
%RWM & $\begin{pmatrix}0.7796\\0.4889\end{pmatrix}$ & $\begin{pmatrix}0.3214&0.0158\\0.0158&0.0798\end{pmatrix}$ \\  [9pt]
%   \hline
%Wall HMC & $\begin{pmatrix}0.7875\\0.4884\end{pmatrix}$ & $\begin{pmatrix}0.3242&0.017\\0.017&0.08\end{pmatrix}$ \\  [9pt]
%   \hline
%exact HMC & $\begin{pmatrix}0.7909\\0.4885\end{pmatrix}$ & $\begin{pmatrix}0.3272&0.0174\\0.0174&0.08\end{pmatrix}$ \\  [9pt]
%   \hline
%c-SphHMC & $\begin{pmatrix}0.79\\0.4864\end{pmatrix}$ & $\begin{pmatrix}0.3249&0.0172\\0.0172&0.0801\end{pmatrix}$ \\  [9pt]
%   \hline
%s-SphHMC & $\begin{pmatrix}0.7935\\0.4852\end{pmatrix}$ & $\begin{pmatrix}0.3233&0.0202\\0.0202&0.07914\end{pmatrix}$ \\  [9pt]
%   \hline
%\end{tabular}
%\caption{Comparing the point estimates of mean and covariance matrix of a bivariate truncated Gaussian distribution using RWM, Wall HMC, exact HMC, c-SphHMC and s-SphHMC.} 
%\label{TMG-moments}
%\end{table}

To evaluate the efficiency of the above-mentioned methods, we repeat this experiment for higher dimensions, $D=10$, and $D=100$. As before, we set the mean to zero and set the $(i,j)$-th element of the covariance matrix to 
$\Sigma_{ij}=1/(1+|i-j|)$. Further, we impose the following constraints on the parameters,
\begin{equation*}
0\leq \beta_i\leq u_{i}
\end{equation*}
where $u_{i}$ (i.e., the upper bound) is set to 5 when $i=1$; otherwise, it is set to $0.5$.

For each method, we obtain $10^5$ MCMC samples after discarding the initial
$10^4$ samples. We set the tuning parameters of algorithms such that their
overall acceptance rates are within a reasonable range. As shown in Table
\ref{TMG-eff}, Spherical HMC algorithms are substantially more efficient than RWM and Wall
HMC. For RWM, the proposed states are rejected about $95\%$ of times due to
violation of the constraints. On average, Wall HMC bounces off the wall around 3.81
($L=2$) and 6.19 ($L=5$) times per iteration for $D=10$ and $D=100$ respectively.
Exact HMC is quite efficient for relatively low dimensional truncated Gaussian ($D=10$);
however it becomes very slow for higher dimensions ($D=100$). In contrast, by augmenting the parameter space, Spherical HMC algorithms handle the constraints in a more efficient way. Since s-SphHMC is more suited for box type constraints, it is substantially more efficient than c-SphHMC in this example.

% latex table generated in R 3.2.0 by xtable 1.7-4 package
% Wed May 20 18:05:22 2015
\begin{table}[ht]
\centering
\begin{tabular}{l|l|ccccc}
  \hline
Dim & Method & AP & s/iter & ESS(min,med,max) & Min(ESS)/s & spdup \\ 
  \hline
 & RWM & 0.62 & 5.72E-05 & (48,691,736) & 7.58 & 1.00 \\ 
   & Wall HMC & 0.83 & 1.19E-04 & (31904,86275,87311) & 2441.72 & 322.33 \\ 
  D= 10 & exact HMC & 1.00 & 7.60E-05 & (1e+05,1e+05,1e+05) & 11960.29 & 1578.87 \\ 
   & c-SphHMC & 0.82 & 2.53E-04 & (62658,85570,86295) & 2253.32 & 297.46 \\ 
   & s-SphHMC & 0.79 & 2.02E-04 & (76088,1e+05,1e+05) & 3429.56 & 452.73 \\ 
   \hline
 & RWM & 0.81 & 5.45E-04 & (1,4,54) & 0.01 & 1.00 \\ 
   & Wall HMC & 0.74 & 2.23E-03 & (17777,52909,55713) & 72.45 & 5130.21 \\ 
  D= 100 & exact HMC & 1.00 & 4.65E-02 & (97963,1e+05,1e+05) & 19.16 & 1356.64 \\ 
   & c-SphHMC & 0.73 & 3.45E-03 & (55667,68585,72850) & 146.75 & 10390.94 \\ 
   & s-SphHMC & 0.87 & 2.30E-03 & (74476,99670,1e+05) & 294.31 & 20839.43 \\ 
   \hline
\end{tabular}
\caption{Comparing the efficiency of RWM, Wall HMC, exact HMC, c-SphHMC and s-SphHMC in terms of sampling from truncated Gaussian distributions. AP is acceptance probability, s/iter is seconds per iteration, ESS(min,med,max) is the (minimum, median, maximum) effective sample size, and Min(ESS)/s is the time-normalized minimum ESS.} 
\label{TMG-eff}
\end{table}

\begin{figure}[t]
\begin{center}
\includegraphics[width=.8\textwidth,height=0.4\textwidth]{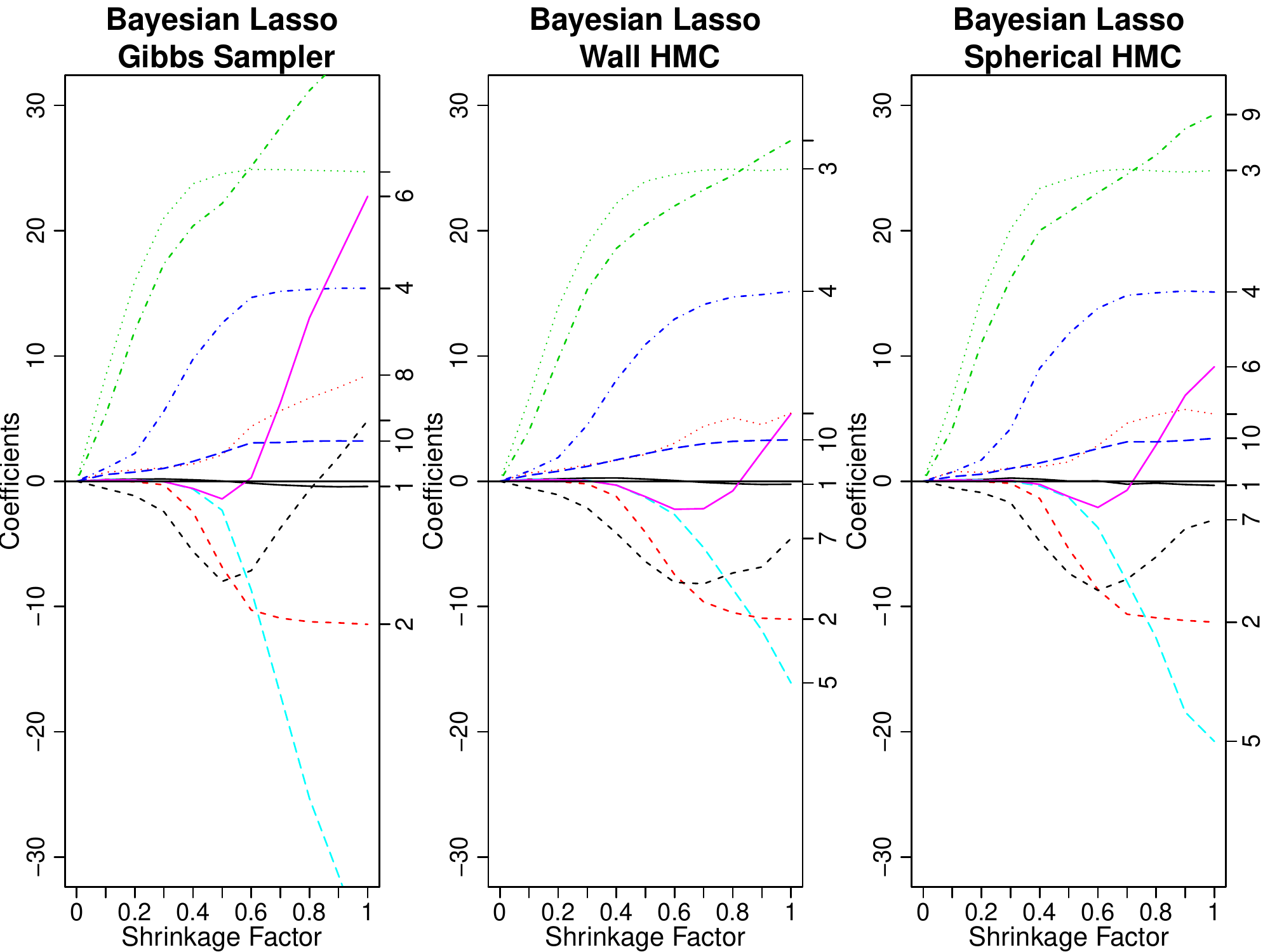}
\end{center}
\caption[Bayesian Lasso using different sampling algorithms]{Bayesian Lasso using three different sampling algorithms: Gibbs sampler (left), Wall HMC (middle) and Spherical HMC (right).}
\label{fig:lasso}
\end{figure}

\subsection{Bayesian Lasso}
In regression analysis, overly complex models tend to overfit the data. Regularized regression models control complexity by imposing a penalty on model parameters. By far, the most popular model in this group is \emph{Lasso} (least absolute shrinkage and selection operator) proposed by \cite{tibshirani96}. In this approach, the coefficients are obtained by minimizing the residual sum of squares (RSS) subject to a constraint on the magnitude of regression coefficients,
\begin{equation}
\min_{\Vert\beta\Vert_1\leq t} \mathrm{RSS}(\beta),  \qquad \mathrm{RSS}(\beta) := \sum_{i} (y_{i} - \beta_{0} - \tp{x}_{i} \beta)^2
\end{equation}
One could estimate the parameters by solving the following optimization problem:
\begin{equation}
\min_{\beta, \lambda} \mathrm{RSS}(\beta)  + \lambda \Vert\beta\Vert_1
\end{equation}
where $\lambda \ge 0$ is the regularization parameter. \cite{park08} and \cite{hans09} have proposed a Bayesian alternative method, called Bayesian Lasso, where the penalty term is replaced by a prior distribution of the form $P(\beta) \propto \exp(-\lambda |\beta|) $, which can be represented as a scale mixture of normal distributions \citep{west87}. This leads to a hierarchical Bayesian model with full conditional conjugacy; therefore, the Gibbs sampler can be used for inference. 

Our proposed spherical augmentation in this paper can directly handle the constraints in Lasso
models. That is, we can \emph{conveniently} use Gaussian priors for model
parameters, $\beta|\sigma^2 \sim \mathcal N(0,\sigma^2 I)$,  and let the sampler \emph{automatically} handle the constraint.
In particular, c-SphHMC can be used to sample posterior distribution of $\beta$ with the 1-norm constraint.
For this problem, we modify the Wall HMC algorithm, which was originally proposed for box type constraints \citep{neal11}.
See Appendix \ref{walldiamond} for more details.

\begin{figure}[t]
\begin{center}
\includegraphics[width=.8\textwidth,height=0.4\textwidth]{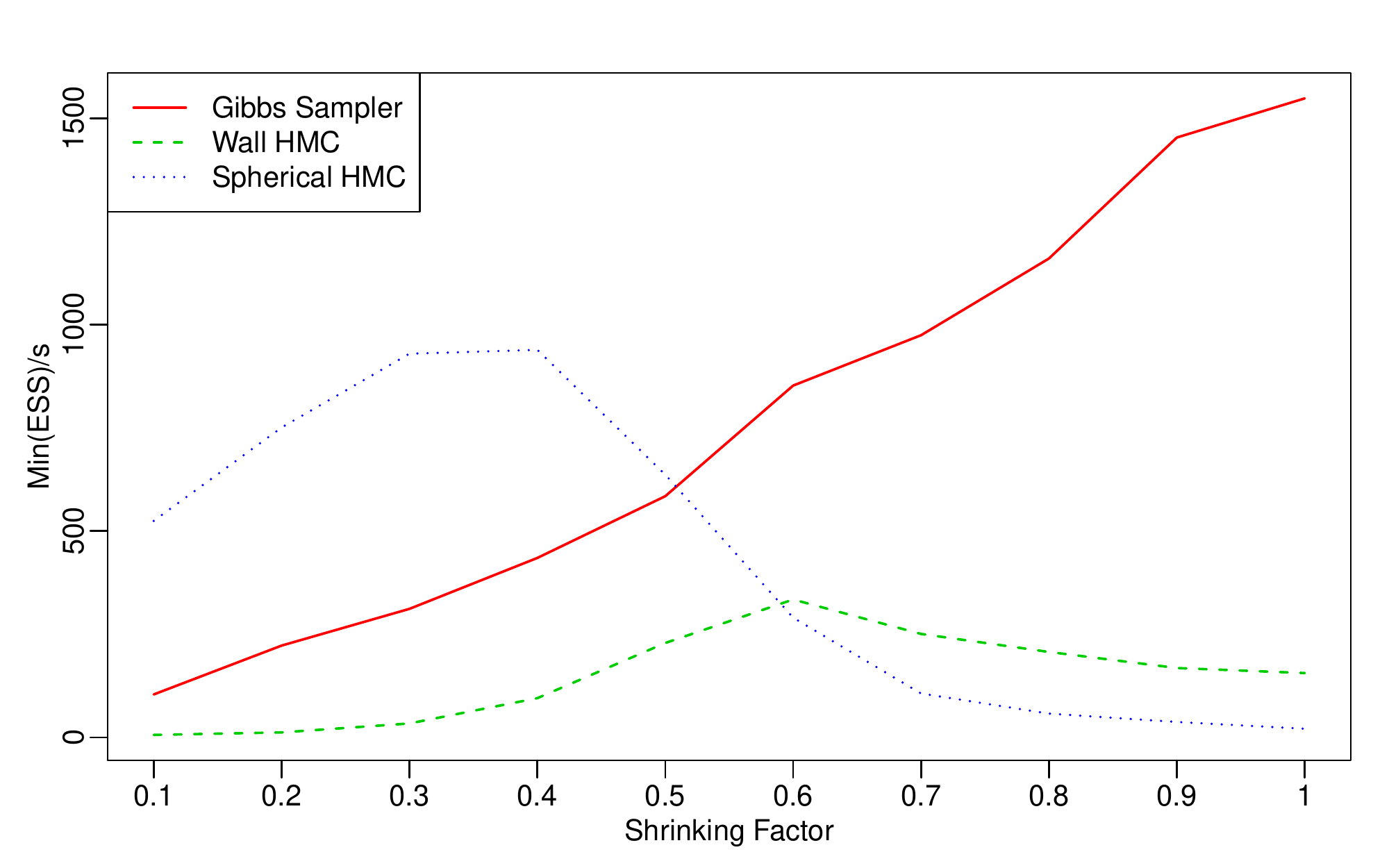}
\end{center}
\caption[Sampling Efficiency in Bayesian Lasso]{Sampling efficiency of different algorithms for Bayesian Lasso based on the diabetes dataset.}
\label{fig:bridgeff}
\end{figure}

We evaluate our method based on the diabetes data set (N=442, D=10) discussed in \cite{park08}. 
Figure \ref{fig:lasso} compares coefficient estimates given by the Gibbs sampler \citep{park08}, Wall HMC, and Spherical HMC respectively as the shrinkage factor $s:=\Vert \hat\beta^{\textrm{Lasso}}\Vert_1/\Vert \hat\beta^{\textrm{OLS}}\Vert_1$ changes from 0 to 1. Here, $\hat\beta^{\textrm{OLS}}$ denotes the estimates obtained by ordinary least squares (OLS) regression. For the Gibbs sampler, we choose different $\lambda$ so that the corresponding shrinkage factor $s$ varies from 0 to 1. For Wall HMC and Spherical HMC, we fix the number of leapfrog steps to 10 and set the trajectory length such that they both have comparable acceptance rates around 70\%.

Figure \ref{fig:bridgeff} compares the sampling efficiency of these three methods. As we impose tighter constraints (i.e., lower shrinkage factors $s$), Spherical HMC becomes substantially more efficient than the Gibbs sampler and Wall HMC.

\subsection{Bridge regression}
The Lasso model discussed in the previous section is in fact a member of a family of regression models called \emph{Bridge regression} \citep{frank93}, where the coefficients are obtained by minimizing the residual sum of squares subject to a constraint on the magnitude of regression coefficients as follows: 
\begin{equation}
\min_{\Vert\beta\Vert_q\leq t} \mathrm{RSS}(\beta),  \qquad \mathrm{RSS}(\beta) := \sum_{i} (y_{i} - \beta_{0} - \tp{x}_{i} \beta)^2
\end{equation}
For Lasso, $q=1$, which allows the model to force some of the coefficients to become exactly zero (i.e., become excluded from the model).
When $q=2$, this model is known as \emph{ridge regression}. Bridge regression is more flexible by allowing different $q$ norm constraints for different effects on shrinking the magnitude of parameters (See Figure \ref{fig:bridge}).

%As mentioned earlier, our Spherical HMC method can easily handle this type of constraints through the following transformation:
%\begin{equation*}
%T: \mathcal Q^D \rightarrow \mathcal S^D, \quad \beta_i\mapsto \beta_i' = \beta_i/ t \mapsto \theta_i= \mathrm{sgn}(\beta_i') |\beta_i'|^{q/2},\; \vect\theta \mapsto \tilde{\vect\theta} = \left(\vect\theta,\sqrt{1-\Vert \vect\theta\Vert_2^2}\right)
%\end{equation*}

While the Gibbs sampler method of \cite{park08} and \cite{hans09} is limited to Lasso, our approach can be applied to all bridge regression models with different $q$. To handle the general $q$-norm constraint, one can map the constrained domain to the unit ball by \eqref{q2b} and apply c-SphHMC. Figure \ref{fig:bridge} compares the parameter estimates of Bayesian Lasso to the estimates obtained from two Bridge regression models with $q=1.2$ and $q=0.8$ for the diabetes dataset \citep{park08} using our Spherical HMC algorithm. As expected, tighter constraints (e.g., $q=0.8$) would lead to faster shrinkage of regression parameters as we decrease $s$.

\begin{figure}[t]
\begin{center}
\includegraphics[width=.8\textwidth,height=0.4\textwidth]{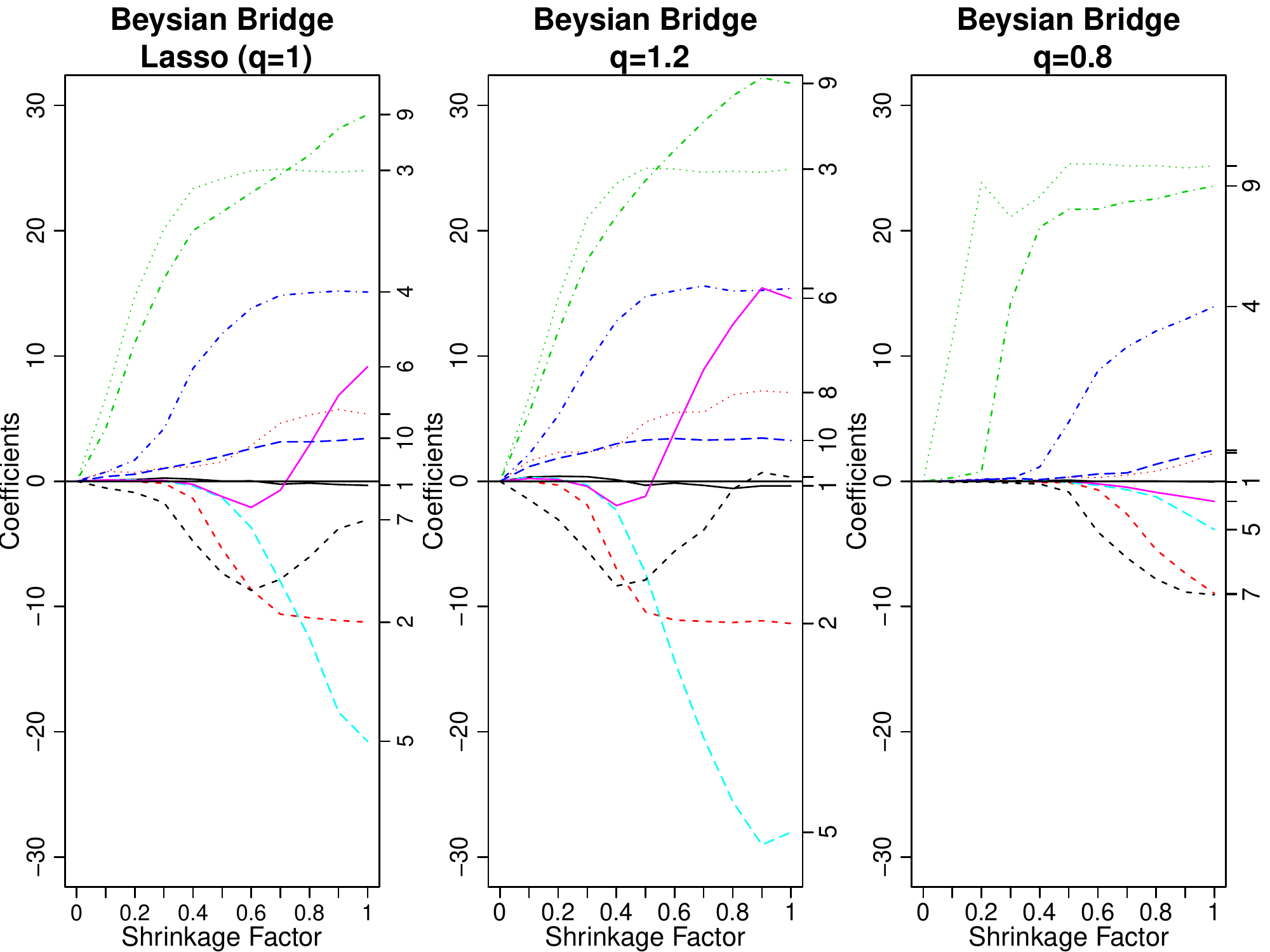}
\end{center}
\caption[Bayesian Bridge Regression by Spherical HMC]{Bayesian Bridge Regression by Spherical HMC: Lasso (q=1, left), q=1.2 (middle), and q=0.8 (right).}
\label{fig:bridge}
\end{figure}

\subsection{Reconstruction of quantized stationary Gaussian process}

We now investigate the example of reconstructing quantized stationary Gaussian
process discussed in \cite{pakman13}. Suppose we are given $N$ values of a
function $f(x_i), i=1,\cdots,N$, which takes discrete values from $\{q_k\}_{k=1}^K$.
We assume that this is a quantized projection of a sample $y(x_i)$ from a
stationary Gaussian process with a known translation-invariant covariance kernel
of the form $\Sigma_{ij}=K(|x_i-x_j|)$, and the quantization follows a known
rule of the form
\begin{equation}\label{trunct}
f(x_i) = q_k, \quad \textrm{if} \; z_k\leq y(x_i) < z_{k+1}
\end{equation}
The objective is to sample from the posterior distribution
\begin{equation}
p(y(x_1),\cdots,y(x_N)|f(x_1),\cdots,f(x_N)) \sim \mathcal N(0,\Sigma) \quad \textrm{trunctated\, by\, rule\, \eqref{trunct}}
\end{equation}
In this example, the function is sampled from a Gaussian process with the
following kernel
\begin{equation*}
K(|x_i-x_j|)=\sigma^2\exp\left\{-\frac{|x_i-x_j|^2}{2\eta^2}\right\},\quad \sigma^2=0.6,\; \eta^2=0.2
\end{equation*}
We sample $N=100$ points of $\{y(x_i)\}$ and quantize them with
\begin{equation*}
q_1=-0.75,\,q_2=-0.25,\,q_3=0.25,\,q_4=0.75,\;\; z_1=-\infty,\,z_2=-0.5,\,z_3=0,\,z_4=0.5,\,z_5=+\infty
\end{equation*}
This example involves two types of constraints: box type (two sided) constraints and one sided
constraints. In implementing our Spherical HMC algorithms, we transform the subspace formed by components
with both finite lower and upper limits into unit ball and map the subspace formed by components 
with one sided constraints to the whole space using absolute value (discussed at the end of Section \ref{SA}).

\begin{figure}[t]
\begin{center}
\includegraphics[width=.9\textwidth,height=0.45\textwidth]{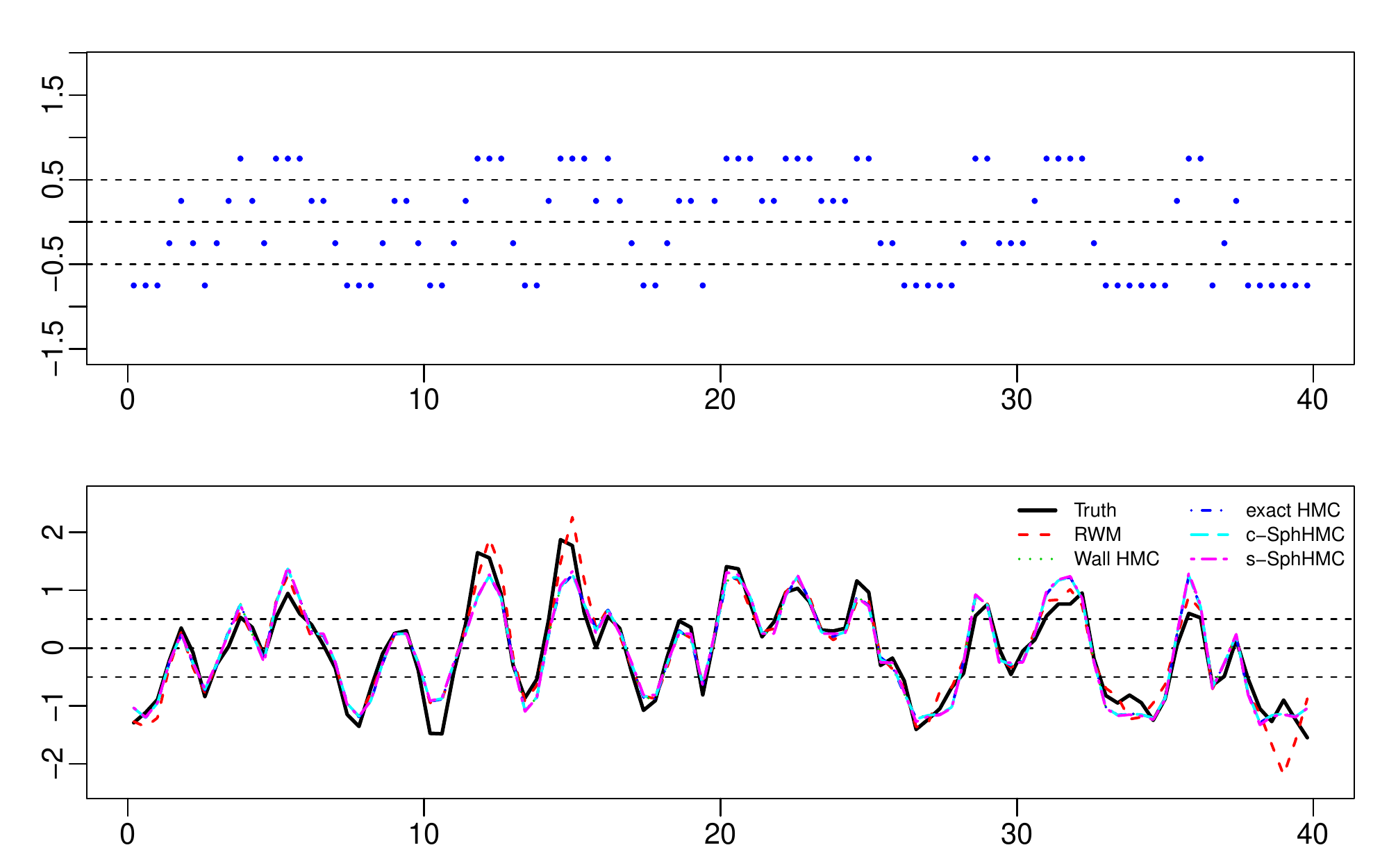}
\caption{Quantized stationary Gaussian process (upper) and the estimates of the
process (lower).
\label{fig:qsGP_est}}
\end{center}
\end{figure}

Figure \ref{fig:qsGP_est} shows the quantized Gaussian process (upper) and the
estimates (lower) with $10^5$ samples given by different MCMC algorithms. Overall, all the methods recover the truth well. Table \ref{qsGP-eff} summarizes the efficiency of sampling $1.1\times 10^5$ and
burning the first $10^4$ with RWM, Wall HMC, exact HMC, c-SphHMC and s-SphHMC. Exact
HMC generates more effective samples but takes much longer time even though
implemented in C. Spherical HMC algorithms outperform it in terms of time
normalized ESS. Interestingly, Wall HMC performs well in this example, even
better than exact HMC and c-SphHMC.

% latex table generated in R 3.2.0 by xtable 1.7-4 package
% Wed May 20 18:05:22 2015
\begin{table}[ht]
\centering
\begin{tabular}{l|ccccc}
  \hline
Method & AP & s/iter & ESS(min,med,max) & Min(ESS)/s & spdup \\ 
  \hline
RWM & 0.70 & 7.11E-05 & (2,9,35) & 0.22 & 1.00 \\ 
  Wall HMC & 0.69 & 9.94E-04 & (12564,24317,43876) & 114.92 & 534.48 \\ 
  exact HMC & 1.00 & 1.00E-02 & (72074,1e+05,1e+05) & 65.31 & 303.76 \\ 
  c-SphHMC & 0.72 & 1.73E-03 & (13029,26021,56445) & 68.44 & 318.32 \\ 
  s-SphHMC & 0.80 & 1.09E-03 & (14422,31182,81948) & 120.59 & 560.86 \\ 
   \hline
\end{tabular}
\caption{Comparing efficiency of RWM, Wall HMC, exact HMC, c-SphHMC and s-SphHMC in reconstructing a quantized stationary Gaussian process. AP is acceptance probability, s/iter is seconds per iteration, ESS(min,med,max) is the (minimal,median,maximal) effective sample size, and Min(ESS)/s is the minimal ESS per second.} 
\label{qsGP-eff}
\end{table}

% Maybe move to the discussion section.
%We want to comment that, as dimension increases, those MCMC algorithms relying
%on Metropolis correction will eventually suffer from the drop of acceptance due
%to the `curse of dimensionality'. To battle this, Spherical HMC algorithms are
%forced to take more leap frog steps with smaller step size, which may undermine
%their advantage over exact HMC. However, exact HMC generally has to spend 
%more time on solving wall hitting time and bouncing off from the boundary in higher dimensional space.

\begin{figure}[t]
\begin{center}
\includegraphics[width=1\textwidth,height=.6\textwidth]{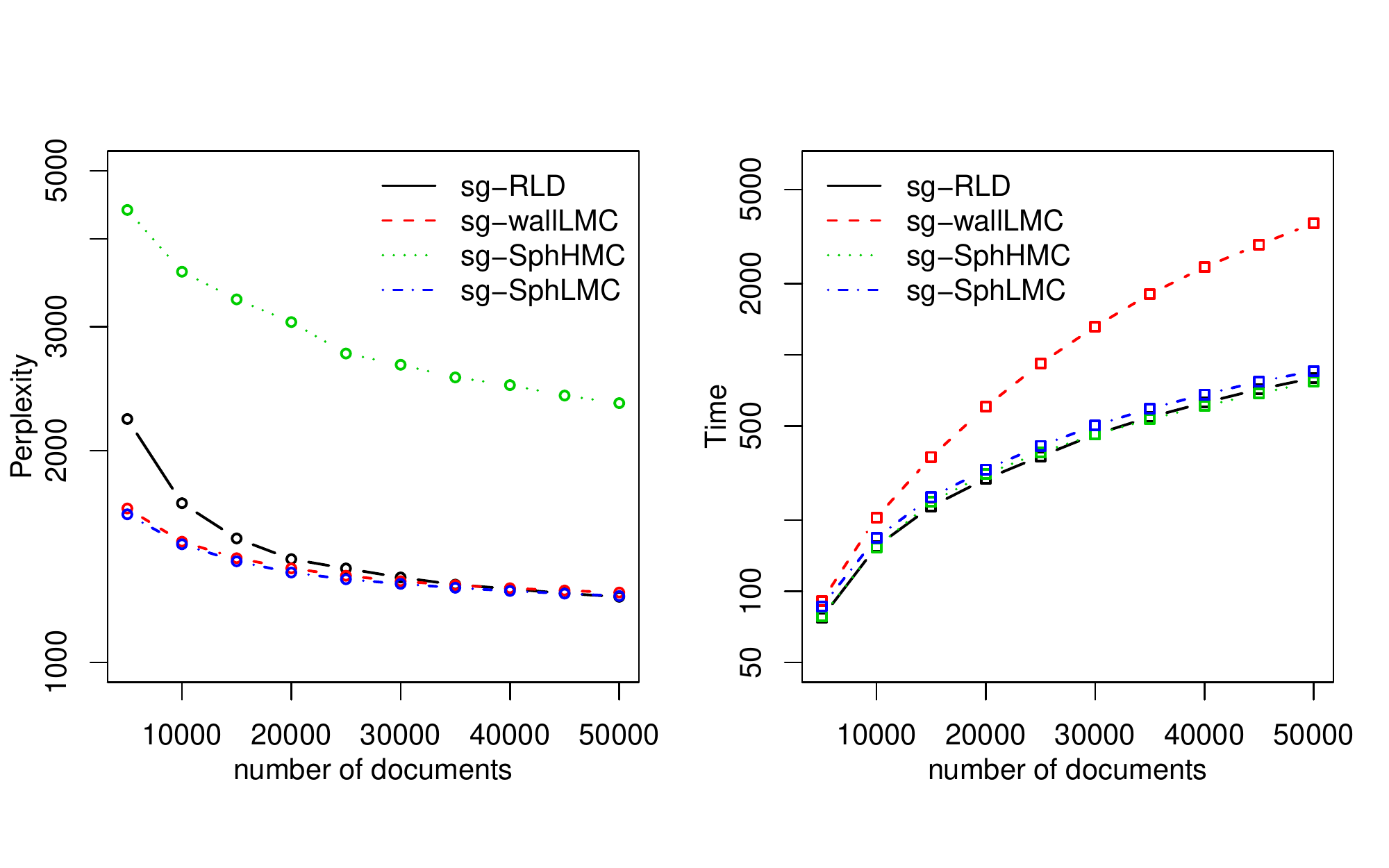}
\caption[Test-set perplexities and computation time]{Test-set perplexity and computation time (in log scale) based on the Wikipedia corpus.}
\label{fig:LDA-wiki}
\end{center}
\end{figure}

\subsection{LDA on Wikipedia corpus}

LDA \citep{blei03} is a popular hierarchical Bayesian model for topic modeling.
The model consists of $K$ topics with probabilities $\{\pi_k\}$ drawn from a symmetric Dirichlet prior
$\mathrm{Dir}(\beta)$. A document $d$ is modeled by a mixture of topics, with
mixing proportions $\eta_d\sim \mathrm{Dir}(\alpha)$. Document $d$ is assumed to be generated by
i.i.d. sampling of a topic assignment, $z_{di}$, from $\eta_d$ for each word $w_{di}$
in the document, and then drawing the word $w_{di}$ from the assigned topic with probability 
$\pi_{z_{di}}$ \citep{patterson13}. \cite{teh06} integrate out $\eta$ analytically to obtain the following semi-collapsed
distribution:
\begin{equation}
p(w,z,\pi|\alpha,\beta) = \prod_{d=1}^D\frac{\Gamma(K\alpha)}{\Gamma(K\alpha+n_{d\cdot\cdot})} \prod_{k=1}^K\frac{\Gamma(\alpha+n_{dk\cdot})}{\Gamma(\alpha)} \prod_{k=1}^K\frac{\Gamma(W\beta)}{\Gamma(\beta)^W} \prod_{w=1}^W \pi_{kw}^{\beta+n_{\cdot kw}-1}
\end{equation}
where $n_{dkw}=\sum_{i=1}^{N_d}\delta(w_{di}=w,z_{di}=k)$. Here, ``$\cdot$'' denotes
the summation over the corresponding index. Given $\pi$, the documents
are i.i.d so the above equation can be factorized as follows \citep{patterson13}:
\begin{equation}
p(w,z,\pi|\alpha,\beta) = p(\pi|\beta) \prod_{d=1}^D p(w_d,z_d|\alpha,\pi), \quad p(w_d,z_d|\alpha,\pi) = \prod_{k=1}^K\frac{\Gamma(\alpha+n_{dk\cdot})}{\Gamma(\alpha)} \prod_{w=1}^W \pi_{kw}^{n_{dkw}}
\end{equation}

To evaluate our proposed methods, we compare them with the state-of-the-art method of \cite{patterson13}. Their approach, called stochastic gradient Riemannian Langevin dynamics (sg-RLD) is an extension of the stochastic gradient Langevin dynamics (SGLD) proposed by \cite{welling11}. Because this approach uses mini-batches of data to approximate the gradient and omits the accept/reject step of Metropolis-Hastings while decreasing the step size, we follow the same procedure to make our methods comparable. Further, because Langevin dynamics can be regarded as a single step Hamiltonian dynamics \citep{neal11}, we set $L=1$. We refer the resulting algorithms as sg-SphHMC and sg-SphLMC, which are modified versions of our SphHMC and SphLMC algorithms. sg-SphLMC uses the following stochastic (natural) gradient (gradient preconditioned with metric)
\begin{equation}\label{sg4LMC}
g_{kw} = [(n^*_{kw} + \beta -1/2)/\theta_{kw} + \theta_{kw}(n^*_{k\cdot}+W(\beta-1/2))]/(2*n^*_{k\cdot}), \quad n^*_{kw} = \frac{|D|}{|D_t|}\sum_{d\in D_t}\E_{z_d|w_d,\theta,\alpha}[n_{dkw}]
\end{equation}
where $1/2$ comes from the logarithm of volume adjustment.
In contrast, the stochastic gradient for sg-SphHMC is $4g_{kw}n^*_{k\cdot}$ (See Section \ref{SphLMC}).
The expectation in Equation \eqref{sg4LMC} is
calculated using Gibbs sampling on the topic assignment in each document
separately, given the conditional distributions \citep{patterson13}
\begin{equation}
p(z_{di}=k|w_d,\theta,\alpha) = \frac{(\alpha+n^{\backslash i}_{dk\cdot})\pi_{kw_{di}}}{\sum_k (\alpha+n^{\backslash i}_{dk\cdot})\pi_{kw_{di}}}
\end{equation}
where $\backslash i$ means a count excluding the topic assignment variable
currently being updated. Step size is decreased according to $\eps_t=a(1+t/b)^{-c}$.

We use perplexity \citep{patterson13,wallach09} to compare the predictive performance of different methods in terms of the
probability they assign to unseen data, 
\begin{equation}
\mathrm{perp}(w_d|\mathcal W,\alpha,\beta) = \exp\left\{-\sum_{i=1}^{n_{d\cdot\cdot}}\log p(w_{di}|\mathcal W,\alpha,\beta)/n_{d\cdot\cdot}\right\}, \quad p(w_{di}|\mathcal W,\alpha,\beta) = \E_{\eta_d,\pi}[\sum_{k} \eta_{dk}\pi_{kw_{di}}]
\end{equation}
where $\mathcal W$ is the training set and $w_d$ is the hold-out sample. More specifically, we use the document completion approach \citep{wallach09}, which
partitions the test document $w_d$ into two sets, $w_d^{\textrm{train}}$ and $w_d^{\textrm{test}}$; we then use $w_d^{\textrm{train}}$ to estimate $n_d$ for the
test document and use $w_d^{\textrm{test}}$ to calculate perplexity.
%More specifically, it is estimated with samples $\{\pi\}$ by sg-RLD/sg-SphLMC,
%and samples for $z_d^{\textrm{train}}$ by Gibbs sampling the topic assignments
%on $w_d^{\textrm{train}}$ \citep{patterson13}
%\begin{equation}
%p(w_{di}|w_d^{\textrm{train}},\mathcal W,\alpha,\beta) = \E_{\pi|\mathcal W,\beta}[\E_{z_d^{\textrm{train}}|\pi,\alpha}[\sum_{k} \hat\eta_{dk}\pi_{kw_{di}}]],\quad \hat\eta_{dk} = p(z_{di}^{\textrm{test}}=k|z_d^{\textrm{train}},\alpha) = \frac{n_{dk\cdot}^{}\textrm{train}+\alpha}{n_{d\cdot\cdot}^{\textrm{train}}+K\alpha}
%\end{equation}

We train the model online using 50000 documents randomly downloaded from Wikipedia with the vocabulary of approximately 8000 words created from Project Gutenburg texts \citep{hoffman10}. The perplexity is evaluated on 1000 held-out documents. A mini-batch of 50 documents is used for updating the natural gradient for 4 algorithms:
sg-RLD, sg-wallLMC\footnote{The stochastic gradient for sg-wallLMC is $[(n^*_{kw} + \beta -1/2) + \pi_{kw}(n^*_{k\cdot}+W(\beta-1/2))]/n^*_{k\cdot}$},
sg-SphHMC%\footnote{The stochastic gradient for SphHMC is $4g_{kw}n^*_{k\cdot}$, where $g_{kw}$ is the stochastic gradient \eqref{sg4LMC}.}
 and sg-SphLMC.

Figure \ref{fig:LDA-wiki} compares the above methods in terms of their perplexities. For each method, we show the best performance over different settings (Settings for best performance are listed in Table \ref{LDA-parset}.). Both sg-wallLMC and sg-SphLMC have lower perplexity than sg-RLD at early stage, when relatively a small number of documents are used for training; as the number of training documents increases, the methods reach the same level of performance. As expected, sg-SphHMC does not perform well due to the absence of a proper scaling provided by the Fisher metric.
% However, sg-wallLMC takes excessively more time than others to complete each run due to the high frequent boundary hitting. 
%Parameter settings are listed in table \ref{LDA-parset}.
% latex table generated in R 3.2.0 by xtable 1.7-4 package
% Thu May 21 15:54:48 2015
%\begin{table}[ht]
%\centering
%\begin{tabular}{l|ccccccc}
%  \hline
%Algorithm & a & b & c & $\alpha$ & $\beta$ & K & Gibbs samples \\ 
%  \hline
% & 0.01 & 1000 & 0.6 & 0.01 & 0.0001 & 100 & 100 \\ 
%  sg-RLD & 0.01 & 1000 & 0.6 & 0.01 & 0.1000 & 100 & 100 \\ 
%   & 0.01 & 1000 & 0.6 & 0.01 & 0.5000 & 100 & 100 \\ 
%   \hline
% & 0.20 & 1000 & 2.0 & 0.01 & 0.5000 & 100 & 100 \\ 
%  sg-wallLMC & 0.20 & 1000 & 2.0 & 0.01 & 0.1000 & 100 & 100 \\ 
%   & 0.20 & 1000 & 2.2 & 0.01 & 0.0100 & 100 & 100 \\ 
%   \hline
% & 0.10 & 1000 & 0.6 & 0.01 & 0.0100 & 100 & 100 \\ 
%  sg-SphHMC & 0.01 & 1000 & 0.6 & 0.01 & 0.0100 & 100 & 100 \\ 
%   & 0.10 & 1000 & 1.0 & 0.01 & 0.5000 & 100 & 100 \\ 
%   \hline
% & 0.25 & 1000 & 1.5 & 0.01 & 0.5000 & 100 & 100 \\ 
%  sg-SphLMC & 0.25 & 1000 & 2.0 & 0.01 & 0.5000 & 100 & 100 \\ 
%   & 0.20 & 1000 & 1.5 & 0.01 & 0.5000 & 100 & 100 \\ 
%   \hline
%\end{tabular}
%\caption{Parameter settings for Wikipedia experiment.} 
%\label{LDA-parset}
%\end{table}

%%%% Parameter setting for the best performance %%%%
\begin{table}[ht]
\centering
\begin{tabular}{l|ccccccc}
  \hline
Algorithm & a & b & c & $\alpha$ & $\beta$ & K & Gibbs samples \\ 
  \hline
 sg-RLD & 0.01 & 1000 & 0.6 & 0.01 & 0.5000 & 100 & 100 \\ 
   \hline
sg-wallLMC & 0.20 & 1000 & 2.0 & 0.01 & 0.5000 & 100 & 100 \\ 
   \hline
  sg-SphHMC & 0.01 & 1000 & 0.6 & 0.01 & 0.0100 & 100 & 100 \\ 
   \hline
 sg-SphLMC & 0.25 & 1000 & 1.5 & 0.01 & 0.5000 & 100 & 100 \\ 
   \hline
\end{tabular}
\caption{Parameter settings for best performance in Wikipedia experiment.} 
\label{LDA-parset}
\end{table}

\section{Discussion} \label{discussion}

We have introduced a new approach, \emph{spherical augmentation}, for sampling
from constrained probability distributions. This method maps the constrained
domain to a sphere in an augmented space. Sampling algorithms can freely explore the surface of sphere to generate samples that remain within the constrained domain when mapped back to the original space. This way, our proposed method provides a mathematically natural and computationally efficient framework that
can be applied to a wide range of statistical inference problems with norm constraints. 

The augmentation approach proposed here is based on the change of variables theorem. We augment the
original $D$-dimensional space with one extra dimension by either inserting slack variables (c-SphHMC) or using embedding map (s-SphHMC),
%extending the parameter vector and auxiliary variables with one more component each. As a result, an extra kinetic energy $\frac12 v_{D+1}^2$
%is added to the total energy \eqref{Hamiltonian}. 
The augmented Hamiltonian is
the same under different representations \eqref{augHamiltonianIc}\eqref{augHamiltonianIs} due to the mathematical fact
that the energy is invariant to the choice of coordinates (Proposition \ref{enginv}).
To account for the change of geometry, a volume adjustment term needs to be used, either as a weight after obtaining all the samples (SphHMC) or as an added term to the total energy (SphLMC).

Our proposed method takes advantage of the splitting strategy to further improve computational efficiency. We split the Lagrangian dynamics and update velocity in the tangent space, rather
than momentum in the cotangent space. This implementation avoids the
requirement of embedding as in \cite{byrne13} and could be applied to
more general situations.

In developing Spherical HMC, we start with the standard HMC, using the Euclidean
metric ${\bf I}$ on unit ball ${\bf B}^D_0(1)$. Then, spherical geometry is introduced
to handle constraints. One possible future direction could be to directly start with RHMC/LMC, which use a more informative
metric (i.e., the Fisher metric ${\bf G_F}$), and then incorporate the spherical geometry
for the constraints. For example, a possible metric for the augmented space
could be ${\bf G_F} + \vect\theta\tp{\vect\theta}/\theta_{D+1}^2$. However,
under such a metric, we might not be able to find the geodesic flow
analytically, which could undermine the added benefit from using the Fisher
metric.

In future, we also intend to explore the possibility of applying the spherical augmentation to Elliptical Slice sampler \citep{murray10} in order to generalize it to Spherical Slice sampler (SSS). The resulting algorithm can be applied to truncated Gaussian process models. In general, we can extend our proposed methods to infinite
dimensional function spaces. This would involve the infinite dimensional manifold $\mathcal S^{\infty}:=\{f\in
L^2(\Omega)|\int f^2 d\mu=1\}$. In this setting it is crucial to ensure that the acceptance probability does not drop quickly
as dimension increases \citep{beskos11}. 

%It is also interesting to develop tune-free algorithms for spherical HMC \citep{hoffman11}.

% Acknowledgements should go at the end, before appendices and references

%\acks{
%We would like to thank Jeffrey Streets, Max Welling, and Alexander Ihler for helpful discussion. 
%SL is supported by XXX,
%BS is supported by NSF grant IIS-1216045 and NIH grant R01-AI107034.} 

%%%%%%%%%%%%%%%%%%%%%%%%%%%%%%%%%%%%%%%%%%%%%%%%%%%%%%%%%%%%%%%%%%%%%%%%%%%%%%%%%%%%%

\newpage

%%%% Appendix %%%%%
\appendix
%\section*{Appendix}

\section{Spherical Geometry}\label{geomS}
We first discuss the geometry of the $D$-dimensional sphere $\mathcal S^D:=\{\tilde{\vect\theta}
\in \mathbb R^{D+1}: \Vert \tilde{\vect\theta}\Vert_2= 1\}\hookrightarrow
\mathbb R^{D+1}$ under different coordinate systems, namely, the \emph{Cartesian coordinate} and the \emph{spherical coordinate}.
Since $\mathcal S^D$ can be embedded (injectively and differentiably mapped to)
in $\mathbb R^{D+1}$, we first introduce the concept of `induced metric'.
\begin{dfn}[induced metric]
If $\mathcal D^{d}$ can be embedded to $\mathcal M^{m}$ $(m>d)$ by $f: U\subset
\mathcal D\hookrightarrow\mathcal M$, then one can define the \emph{induced metric},
$g_{\mathcal D}$, on $T\mathcal D$ through the metric $g_{\mathcal M}$ defined on $T\mathcal M$:
\begin{equation}
g_{\mathcal D}(\vect\theta)({\bf u},{\bf v})=g_{\mathcal M}(f(\vect\theta))(df_{\vect\theta}({\bf u}),df_{\vect\theta}({\bf v})),\quad {\bf u},{\bf v}\in T_{\vect\theta}\mathcal D
\end{equation}
\end{dfn}
\begin{rk}
For any $f: U\subset \mathcal S^D\hookrightarrow\mathbb R^{D+1}$, we can define
the induced metric through dot product on $\mathbb R^{D+1}$. More specifically,
\begin{equation}
g_{\mathcal S}({\bf u},{\bf v}) = \tp{[(Df){\bf u}]} (Df){\bf v} = \tp{\bf u} [\tp{(Df)} (Df)] {\bf v}
\end{equation}
where $(Df)_{(D+1)\times D}$ is the Jacobian matrix of the mapping $f$.
A Metric induced from dot product on Euclidean space is called
a ``canonical metric''.
This observation leads to the following simple fact that lays down the
foundation of Spherical HMC.
\end{rk}

\begin{prop}[Energy invariance]\label{enginv}
Kinetic energy $\frac12\langle {\bf v},{\bf v}\rangle_{{\bf G}(\vect\theta)}$ is
invariant to the choice of coordinate systems.
\end{prop}
\begin{proof}
For any ${\bf v}\in T_{\vect\theta}\mathcal D$, suppose $\vect\theta(t)$ such
that $\dot{\vect\theta}(0)={\bf v}$. Denote the pushforward of ${\bf v}$ by
embedding map $f:\mathcal D\to \mathcal M$ as $\tilde{\bf v}:=f_*({\bf
v})=\frac{d}{dt}(f\circ\vect\theta)(0)$. Then we have
\begin{equation}
\frac12\langle {\bf v},{\bf v}\rangle_{{\bf G}(\vect\theta)} = \frac12 g_{\mathcal M}(f(\vect\theta))(\tilde{\bf v},\tilde{\bf v})
\end{equation}
That is, regardless of the form of the energy under a coordinate system, its value is the same as the one in the embedded manifold. In particular, when $\mathcal M=\mathbb R^{D+1}$, the right hand side simplifies to $\frac12 \Vert\tilde{\bf v}\Vert_2^2$.
\end{proof}

\subsection{Canonical metric in the Cartesian coordinate}\label{Metc}
Now consider the $D$-dimensional ball $\mathcal B_{\bf 0}^D(1):=\{\vect\theta\in \mathbb
R^D: \Vert \vect\theta\Vert_2\leq 1\}$. Here, $\{\vect\theta, \mathcal B_{\bf 0}^D(1)\}$ can be viewed as the Cartesian
coordinate system for ${\mathcal S}^D$.
The coordinate mapping $T_{\mathcal B\to\mathcal
S_+}:\vect\theta\mapsto\tilde{\vect\theta}=(\vect\theta,\theta_{D+1})$ in
\eqref{b2s} can be viewed as the embedding map into $\mathbb R^{D+1}$, and the
Jacobian matrix of $T_{\mathcal B\to\mathcal S_+}$ is $dT_{\mathcal B\to\mathcal S_+}=\frac{d\tilde{\vect\theta}}{d\tp{\vect\theta}}=\begin{bmatrix}{\bf I}_{D}\\-\tp{\vect\theta}/\theta_{D+1}\end{bmatrix}$.
Therefore the \emph{canonical metric} of ${\mathcal S}^D$ in the Cartesian
coordinate, ${\bf G}_{\mathcal S_c}(\vect\theta)$, is
\begin{equation}\label{GSc}
{\bf G}_{\mathcal S_c}(\vect\theta) = \tp{dT}_{\mathcal B\to\mathcal S_+} dT_{\mathcal B\to\mathcal S_+} = {\bf I}_D + \frac{\vect\theta \tp{\vect\theta}}{\theta_{D+1}^2} = {\bf I}_D + \frac{\vect\theta \tp{\vect\theta}}{1-\Vert \vect\theta\Vert_2^2}
\end{equation}

Another way to obtain the metric is through the first fundamental form $ds^{2}$
(i.e., squared infinitesimal length of a curve) for $\mathcal S^D$, which can be expressed in terms of the differential form $d\vect\theta$ and the canonical metric ${\bf G}_{\mathcal S_c}(\vect\theta)$,
\begin{equation*}
ds^{2} = \langle d\vect\theta, d\vect\theta \rangle_{{\bf G}_{\mathcal S_c}} = d\tp{\vect\theta}{\bf G}_{\mathcal S_c}(\vect\theta)d\vect\theta
\end{equation*}
On the other hand, $ds^{2}$ can also be obtained as follows \citep{spivak79-1}:
\begin{equation*}
ds^2 = \sum_{i=1}^{D+1} d\theta_i^2 = \sum_{i=1}^D d\theta_i^2 + (d(\theta_{D+1}(\vect\theta)))^2
= d\tp{\vect\theta} d\vect\theta + \frac{(\tp{\vect\theta} d\vect\theta)^2}{1-\Vert \vect\theta\Vert_2^2} = d\tp{\vect\theta} [{\bf I} + \vect\theta \tp{\vect\theta}/\theta_{D+1}^2] d\vect\theta
\end{equation*}
Equating the above two quantities yields the form of the canonical metric ${\bf
G}_{\mathcal S_c}(\vect\theta)$ as in Equation \eqref{GSc}.
This viewpoint provides a natural way to explain the length of tangent
vector. For any vector $\tilde{\bf v}=({\bf v},v_{D+1})\in
T_{\tilde{\vect\theta}}{\mathcal S}^D=\{\tilde {\bf v}\in \mathbb R^{D+1}:
\tp{\tilde{\vect\theta}}\tilde{\bf v}=0\}$, one could think of ${\bf
G}_{\mathcal S_c}(\vect\theta)$ as a mean to express the length of $\tilde{\bf
v}$ in terms of ${\bf v}$,
\begin{equation}
\tp{\bf v}{\bf G}_{\mathcal S_c}(\vect\theta){\bf v} = \Vert {\bf v}\Vert_2^2 + \frac{\tp{\bf v}\vect\theta \tp{\vect\theta} {\bf v}}{\theta_{D+1}^2} = \Vert {\bf v}\Vert_2^2 + \frac{
(-\theta_{D+1} v_{D+1})^2}{\theta_{D+1}^2} = \Vert {\bf v}\Vert_2^2 + v_{D+1}^2 = \Vert \tilde{\bf v}\Vert_2^2
\end{equation}
This indeed verifies the energy invariance Proposition \ref{enginv}.

The following proposition provides the analytic forms of the determinant and
the inverse of ${\bf G}_{\mathcal S_c}(\vect\theta)$.
\begin{prop}
The determinant and the inverse of the \emph{canonical metric} are as
follows
\begin{equation}
|{\bf G}_{\mathcal S_c}(\vect\theta)| = \theta_{D+1}^{-2}, \qquad {\bf G}_{\mathcal S_c}(\vect\theta)^{-1} = {\bf I}_D - \vect\theta \tp{\vect\theta}
\end{equation}
\end{prop}
\begin{proof}
The determinant of the canonical metric ${\bf G}_{\mathcal
S_c}(\vect\theta)$ is given by the matrix determinant lemma,
\begin{equation*}
|{\bf G}_{\mathcal S_c}(\vect\theta)| = \det \left[ {\bf I}_D + \frac{\vect\theta \tp{\vect\theta}}{\theta_{D+1}^2}\right] = 1+ \frac{\tp{\vect\theta} \vect\theta}{\theta_{D+1}^2} =
\frac{1}{\theta_{D+1}^2}
\end{equation*}
The inverse of ${\bf G}_{\mathcal S_c}(\vect\theta)$ is obtained by the Sherman-Morrison-Woodbury formula \citep{golub96}
\begin{equation*}
{\bf G}_{\mathcal S_c}(\vect\theta)^{-1} = \left[ {\bf I}_D + \frac{\vect\theta \tp{\vect\theta}}{\theta_{D+1}^2} \right]^{-1} = {\bf I}_D - \frac{\vect\theta
\tp{\vect\theta}/\theta_{D+1}^2}{1+\tp{\vect\theta}\vect\theta/\theta_{D+1}^2} = {\bf I}_D - \vect\theta \tp{\vect\theta}
\end{equation*}
\end{proof}
\begin{cor}\label{voladjB2S}
The volume adjustment of changing measure in \eqref{domainB2S} is
\begin{equation}
\left|\frac{d\vect\theta_{\mathcal B}}{d\vect\theta_{\mathcal S_c}}\right|=|{\bf G}_{\mathcal S_c}(\vect\theta)|^{-\frac12}=|\theta_{D+1}|
\end{equation}
\end{cor}
\begin{proof}
Canonical measure can be defined through the Riesz representation theorem by using a
positive linear functional on the space $C_0(\mathcal S^D)$ of compactly
supported continuous functions on $\mathcal S^D$ \citep{spivak79-1,docarmo92}.
More precisely, there is a unique positive Borel measure $\mu_c$ such that for
(any) coordinate chart $(\mathcal B_{\bf 0}^D(1),T_{\mathcal B\to\mathcal S_+})$,
\begin{equation*}
\int_{\mathcal S_+^D} f(\tilde{\vect\theta}) d\vect\theta_{\mathcal S_c} = \int_{\mathcal B_{\bf 0}^D(1)} f(\vect\theta) \sqrt{|{\bf G}_{\mathcal S_c}(\vect\theta)|} d\vect\theta_{\mathcal B}
\end{equation*}
where $\mu_c=d\vect\theta_{\mathcal S_c}$, and $d\vect\theta_{\mathcal B}$ is the
Euclidean measure. Therefore we have
\begin{equation*}
\left|\frac{d\vect\theta_{\mathcal S_c}}{d\vect\theta_{\mathcal B}}\right| =
|{\bf G}_{\mathcal S_c}(\vect\theta)|^{\frac12} = |\theta_{D+1}|^{-1}
\end{equation*}
Alternatively, $\left|\frac{d\vect\theta_{\mathcal B}}{d\vect\theta_{\mathcal S_c}}\right|=|\theta_{D+1}|$.
\end{proof}

\subsection{Geodesic on a sphere in the Cartesian coordinate}\label{GEODS}
To find the geodesic on a sphere, we need to solve the following equations:
\begin{align}
\dot{\vect\theta} & = {\bf v} \label{geodS:1}\\
\dot{\bf v} & = -\tp{\bf v}\vect\Gamma_{\mathcal S_c}(\vect\theta){\bf v} \label{geodS:2}
\end{align}
for which we need to calculate the Christoffel symbols, $\vect\Gamma_{\mathcal S_c}(\vect\theta)$, first. Note that the $(i,j)$-th element of ${\bf G}_{\mathcal S_c}$ is $g_{ij} = \delta_{ij} +
\theta_i\theta_j/\theta_{D+1}^2$, and the $(i,j,k)$-th element of $d{\bf G}_{\mathcal S_c}$
is $g_{ij,k} = (\delta_{ik}\theta_j + \theta_i\delta_{jk})/\theta_{D+1}^2 + 2\theta_i\theta_j\theta_k/\theta_{D+1}^4$.
Therefore
\begin{equation*}
\begin{split}
\Gamma_{ij}^k &= \frac{1}{2}g^{kl}[g_{lj,i}+g_{il,j}-g_{ij,l}]\\
& = \frac{1}{2}(\delta^{kl}\!\!-\!\theta^k\theta^l)[(\delta_{li}\theta_j \!+\!
\theta_l\delta_{ji})/\theta_{D+1}^2 + (\delta_{ij}\theta_l +
\theta_i\delta_{lj})/\theta_{D+1}^2 - (\delta_{il}\theta_j +
\theta_i\delta_{jl})/\theta_{D+1}^2 + 2\theta_i\theta_j\theta_l/\theta_{D+1}^4]\\
& = (\delta^{kl}-\theta^k\theta^l)\theta_l/\theta_{D+1}^2[\delta_{ij}+\theta_i\theta_j/\theta_{D+1}^2]\\
& = \theta_k[\delta_{ij}+\theta_i\theta_j/\theta_{D+1}^2] = [{\bf G}_{\mathcal S_c}(\vect\theta)\otimes \vect\theta]_{ijk}
\end{split}
\end{equation*}
Using these results, we can write Equation \eqref{geodS:2} as $\dot{\bf v} = -\tp{\bf v} {\bf G}_{\mathcal S_c}(\vect\theta){\bf v}\vect\theta = -\Vert\tilde{\bf
v}\Vert_2^2\vect\theta$. Further, we have
\begin{alignat*}{4}
\dot \theta_{D+1} = & \frac{d}{dt} \sqrt{1-\Vert \vect\theta\Vert_{2}^{2}} && = -\frac{\tp{\vect\theta}}{\theta_{D+1}}\dot{\vect\theta} && = v_{D+1}\\
\dot v_{D+1} = &  -\frac{d}{dt} \frac{\tp{\vect\theta}{\bf v}}{\theta_{D+1}} && = -\frac{\tp{\dot{\vect\theta}} {\bf v}+ \tp{\vect\theta}\dot{\bf v}}{\theta_{D+1}} + \frac{\tp{\vect\theta}{\bf v}}{\theta_{D+1}^{2}}\dot\theta_{D+1} && = -\Vert\tilde{\bf v}\Vert_2^2\theta_{D+1}
\end{alignat*}
Therefore, we can rewrite the geodesic equations \eqref{geodS:1}\eqref{geodS:2} with augmented components as
\begin{align}
\dot{\tilde{\vect\theta}} & = \tilde{\bf v} \label{geodSx:1}\\
\dot{\tilde{\bf v}} & = -\Vert\tilde{\bf v}\Vert_2^2 \tilde{\vect\theta} \label{geodSx:2}
\end{align}
Multiplying both sides of Equation \eqref{geodSx:2} by $\tp{\tilde{\bf v}}$ to obtain $\frac{d}{dt}\Vert \tilde{\bf v}\Vert_2^2=0$, we can solve the above system of differential equations as follows: 
\begin{align*}
\tilde{\vect\theta}(t) & = \tilde{\vect\theta}(0) \cos(\Vert \tilde{\bf v}(0)\Vert_{2} t) + \frac{\tilde{\bf v}(0)}{\Vert \tilde{\bf v}(0)\Vert_{2}} \sin(\Vert \tilde{\bf v}(0)\Vert_{2} t)\\
\tilde{\bf v}(t) & = -\tilde{\vect\theta}(0) \Vert \tilde{\bf v}(0)\Vert_{2} \sin(\Vert \tilde{\bf v}(0)\Vert_{2} t) + \tilde{\bf v}(0) \cos(\Vert \tilde{\bf v}(0)\Vert_{2} t)
\end{align*}

\subsection{Round metric in the spherical coordinate}\label{Mets}
Consider the $D$-dimensional hyper-rectangle $\mathcal R_{\bf
0}^D:=[0,\pi]^{D-1}\times [0,2\pi)$ and the corresponding spherical coordinate
system, $\{\vect\theta, \mathcal R_{\bf 0}^D\}$, for $\mathcal S^D$.
The coordinate mapping $T_{\mathcal R_{\bf 0}\to \mathcal S}: \vect\theta\mapsto{\bf x},\;
x_d = \cos(\theta_d)\prod_{i=1}^{d-1}\sin(\theta_i),\, d=1,\cdots, D+1$,
($\theta_{D+1}=0$) can be viewed as the embedding map into $\mathbb R^{D+1}$, and the
Jacobian matrix of $T_{\mathcal R_{\bf 0}\to\mathcal S}$ is $\frac{d{\bf
x}}{d\tp{\vect\theta}}$ with the $(d,j)$-th element $[-\tan(\theta_d)\delta_{dj}+\cot(\theta_j)I(j<d)]x_d$.
The induced metric of ${\mathcal S}^D$ in the spherical coordinate is
called \emph{round metric}, denoted as ${\bf G}_{\mathcal S_r}(\vect\theta)$,
whose $(i,j)$-th element is as follows
\begin{equation}\label{GSs}
\begin{aligned}
&{\bf G}_{\mathcal S_r}(\vect\theta)_{ij}\\
&= \sum_{d=1}^{D+1} [-\tan(\theta_d)\delta_{di}+\cot(\theta_j)I(i<d)][-\tan(\theta_d)\delta_{dj}+\cot(\theta_j)I(j<d)]x_d^2\\
&= \tan^2(\theta_i)\delta_{ij}x_i^2 - \tan(\theta_i)\cot(\theta_j)I(i>j)x_i^2 - \tan(\theta_j)\cot(\theta_i)I(i<j)x_j^2 + \cot(\theta_i)\cot(\theta_j)\sum_{d>\max\{i,j\}} x_d^2\\
&=\begin{dcases}
-\tan(\theta_j)\cot(\theta_i)x_j^2 + \cot(\theta_i)\cot(\theta_j)\sum_{d>j} x_d^2
%=-\tan(\theta_j)\cot(\theta_i)x_j^2 + \cot(\theta_i)\cot(\theta_j) x_j^2/\cot^2(\theta_j)
=0,& i<j\\
\tan^2(\theta_i)x_i^2 + \cot^2(\theta_i) \sum_{d>i} x_d^2
= (\tan^2(\theta_i)+1)x_i^2
= \prod_{i=1}^{d-1}\sin^2(\theta_i),& i=j
\end{dcases}\\
&= \prod_{i=1}^{d-1}\sin^2(\theta_i) \delta_{ij}
\end{aligned}
\end{equation}
Therefore, ${\bf G}_{\mathcal S_r}(\vect\theta)=\diag[1,\sin^2(\theta_1),\cdots,\prod_{d=1}^{D-1}\sin^2(\theta_d)]$.
Another way to obtain ${\bf G}_{\mathcal S_r}(\vect\theta)$ is through the coordinate
change:
\begin{equation}
{\bf G}_{\mathcal S_r}(\vect\theta) =  \frac{d\tp{\vect\theta}_{\mathcal
S_c}}{d\vect\theta_{\mathcal S_r}}{\bf G}_{\mathcal S_c}(\vect\theta)
\frac{d\vect\theta_{\mathcal S_c}}{d\tp{\vect\theta}_{\mathcal S_r}}
\end{equation}
Similar to Corollary \eqref{voladjB2S}, we have
\begin{prop}\label{voladjR2S}
The volume adjustment of changing measure in \eqref{domainR2S} is
\begin{equation}
\left|\frac{d\vect\theta_{\mathcal R_{\bf 0}}}{d\vect\theta_{\mathcal S_r}}\right|=|{\bf G}_{\mathcal S_r}(\vect\theta)|^{-\frac12}=\prod_{d=1}^{D-1}\sin^{-(D-d)}(\theta_{d})
\end{equation}
\end{prop}

%\begin{comment}

\section{Jacobian of the transformation between $q$-norm domains}\label{Jacobian}
The following proposition gives the weights needed for the transformation from
$\mathcal Q^D$ to $\mathcal B_{\bf 0}^D(1)$.
\begin{prop}
The Jacobian determinant (weight) of $T_{\mathcal B\to\mathcal Q}$ is as follows:
\begin{equation}
|dT_{\mathcal S\to\mathcal Q}| = \left(\frac{2}{q}\right)^D \left(\prod_{i=1}^{D}|\theta_i|\right)^{2/q-1}
\end{equation}
\end{prop}
\begin{proof}
Note
\begin{equation*}
T_{\mathcal B\to\mathcal Q}:\, \vect\theta\mapsto \vect\beta=\mathrm{sgn}(\vect\theta)|\vect\theta|^{2/q}
\end{equation*}
The Jacobian matrix for $T_{\mathcal B\to\mathcal Q}$ is
\begin{equation*}
\frac{d\vect\beta}{d\tp{\vect\theta}} = \frac{2}{q}\mathrm{diag}(|\vect\theta|^{2/q-1})
\end{equation*}
Therefore the Jacobian determinant of $T_{\mathcal B\to \mathcal Q}$ is
\begin{equation*}
|dT_{\mathcal B\to\mathcal Q}|
= \left|\frac{d\vect\beta}{d\tp{\vect\theta}}\right| = \left(\frac{2}{q}\right)^D
\left(\prod_{i=1}^{D}|\theta_i|\right)^{2/q-1}
\end{equation*}
\end{proof}

The following proposition gives the weights needed for the change of domains from $\mathcal R^D$ to $\mathcal B_{\bf 0}^D(1)$.
\begin{prop}
The Jacobian determinant (weight) of $T_{\mathcal B\to\mathcal R}$ is as follows:
\begin{equation}
|dT_{\mathcal B\to\mathcal R}| = \frac{\Vert\vect\theta\Vert_2^D}{\Vert\vect\theta\Vert_{\infty}^D} \prod_{i=1}^D \frac{u_i-l_i}{2}
\end{equation}
\end{prop}
\begin{proof}
First, we note 
\begin{equation*}
T_{\mathcal B\to\mathcal R} = T_{\mathcal C\to\mathcal R}\circ T_{\mathcal
B\to\mathcal C}:\, \vect\theta\mapsto \vect\beta'=\vect\theta
\frac{\Vert\vect\theta\Vert_2}{\Vert\vect\theta\Vert_{\infty}}\mapsto \vect\beta=\frac{{\bf u}-{\bf l}}{2}\vect\beta'+\frac{{\bf u}+{\bf l}}{2}
\end{equation*}
The corresponding Jacobian matrices are
\begin{alignat*}{3}
T_{\mathcal B\to\mathcal C}:\,& \frac{d\vect\beta'}{d\tp{\vect\theta}} && =\, \frac{\Vert\vect\theta\Vert_2}{\Vert\vect\theta\Vert_{\infty}}\left[{\bf I} + \vect\theta
\left(\frac{\tp{\vect\theta}}{\Vert\vect\theta\Vert_2^2} - \frac{\tp{\bf e}_{\arg\max|\vect\theta|}}{\vect\theta_{\arg\max|\vect\theta|}}\right)\right]\\
T_{\mathcal C\to\mathcal R}:\,& \frac{d\vect\beta}{d\tp{(\vect\beta')}} && =\, \mathrm{diag}\left(\frac{{\bf u}-{\bf l}}{2}\right)
\end{alignat*}
where ${\bf e}_{\arg\max|\vect\theta|}$ is a vector with $(\arg\max|\vect\theta|)$-th element 1 and all others 0.
Therefore, 
\begin{equation*}
|dT_{\mathcal B\to\mathcal R}| = |dT_{\mathcal C\to\mathcal R}|\, |dT_{\mathcal
B\to\mathcal C}| =  \left|\frac{d\vect\beta}{d\tp{(\vect\beta')}}\right| \left|\frac{d\vect\beta'}{d\tp{\vect\theta}}\right| = 
\frac{\Vert\vect\theta\Vert_2^D}{\Vert\vect\theta\Vert_{\infty}^D} \prod_{i=1}^D \frac{u_i-l_i}{2}
\end{equation*}
\end{proof}

%\end{comment}

\section{Splitting Hamiltonian (Lagrangian) dynamics on $\mathcal S^D$}\label{splitHS}
Splitting the Hamiltonian dynamics and its usefulness in improving HMC is a well-studied topic of research \citep{leimkuhler04, shahbaba14, byrne13}. Splitting the Lagrangian dynamics (used in our approach), on the other hand, has not been discussed in the literature, to the best of our knowledge. Therefore, we prove the validity of our splitting method by starting with the well-understood method of splitting Hamiltonian \citep{byrne13},
\begin{equation*}
H^*(\vect\theta,{\bf p}) = \frac{1}{2} U(\vect\theta) + \frac{1}{2} \tp{\bf p} {\bf G}_{\mathcal S_c}(\vect\theta)^{-1}{\bf p} + \frac{1}{2}U(\vect\theta)
\end{equation*}
The corresponding systems of differential equations,\\
\begin{subequations}
\begin{minipage}{.5\textwidth}
\begin{align*}
\begin{cases}
\begin{aligned}
&\dot{\vect\theta} && = && {\bf 0}\\
&\dot{\bf p} && = && -\frac{1}{2} \nabla_{\vect\theta} U(\vect\theta)
\end{aligned}
\end{cases}
\end{align*}
\end{minipage}
\begin{minipage}{.5\textwidth}
\begin{align*}
\begin{cases}
\begin{aligned}
&\dot{\vect\theta} && = && {\bf G}_{\mathcal S_c}(\vect\theta)^{-1}{\bf p}\\
&\dot{\bf p} && = && -\frac{1}{2} \tp{\bf p}{\bf G}_{\mathcal S_c}(\vect\theta)^{-1} d{\bf G}_{\mathcal S_c}(\vect\theta) {\bf G}_{\mathcal S_c}(\vect\theta)^{-1}{\bf p}
\end{aligned}
\end{cases}
\end{align*}
\end{minipage}
\end{subequations}
can be written in terms of Lagrangian dynamics in $(\vect\theta, {\bf v})$ as follows:\\
\begin{subequations}
\begin{minipage}{.5\textwidth}
\begin{align*}
\begin{cases}
\begin{aligned}
&\dot{\vect\theta} && = && {\bf 0}\\
&\dot{\bf v} && = && -\frac{1}{2} {\bf G}_{\mathcal S_c}(\vect\theta)^{-1} \nabla_{\vect\theta} U(\vect\theta)
\end{aligned}
\end{cases}
\end{align*}
\end{minipage}
\begin{minipage}{.5\textwidth}
\begin{align*}
\begin{cases}
\begin{aligned}
&\dot{\vect\theta} && = && {\bf v}\\
&\dot{\bf v} && = && -\tp{\bf v}\vect\Gamma_{\mathcal S_c}(\vect\theta){\bf v}
\end{aligned}
\end{cases}
\end{align*}
\end{minipage}
\end{subequations}
We have solved the second dynamics (on the right) in Section \ref{GEODS}. 
To solve the first dynamics, we note that
\begin{alignat*}{4}
\dot \theta_{D+1} = & \frac{d}{dt} \sqrt{1-\Vert \vect\theta\Vert_{2}^{2}} && = -\frac{\tp{\vect\theta}}{\theta_{D+1}}\dot{\vect\theta} && = 0\\
\dot v_{D+1} = & -\frac{d}{dt} \frac{\tp{\vect\theta}{\bf v}}{\theta_{D+1}} && = -\frac{\tp{\dot{\vect\theta}} {\bf v}+ \tp{\vect\theta}\dot{\bf v}}{\theta_{D+1}} +
\frac{\tp{\vect\theta}{\bf v}}{\theta_{D+1}^{2}}\dot\theta_{D+1} && = \frac{1}{2} \frac{\tp{\vect\theta}}{\theta_{D+1}} {\bf G}_{\mathcal S_c}(\vect\theta)^{-1} \nabla_{\vect\theta} U(\vect\theta)
\end{alignat*}
Therefore, we have
\begin{align*}
\tilde{\vect\theta}(t) & = \tilde{\vect\theta}(0)\\
\tilde{\bf v}(t) & = \tilde{\bf v}(0) -\frac{t}{2} \begin{bmatrix} {\bf I}\\ -\frac{\tp{\vect\theta(0)}}{\theta_{D+1}(0)}\end{bmatrix} [{\bf I}-\vect\theta(0)\tp{\vect\theta(0)}]
\nabla_{\vect\theta} U(\vect\theta)
\end{align*}
where $\begin{bmatrix} {\bf I}\\ -\frac{\tp{\vect\theta(0)}}{\theta_{D+1}(0)}\end{bmatrix} [{\bf I}-\vect\theta(0)\tp{\vect\theta(0)}]=\begin{bmatrix} {\bf
I}-\vect\theta(0)\tp{\vect\theta(0)}\\ -\theta_{D+1}(0) \tp{\vect\theta(0)}\end{bmatrix} = \begin{bmatrix} {\bf I}\\ \tp{\bf 0}\end{bmatrix} - \tilde{\vect\theta}(0) \tp{\vect\theta(0)}$.

Finally, we note that $\Vert\tilde{\vect\theta}(t)\Vert_2=1$ if $\Vert\tilde{\vect\theta}(0)\Vert_2=1$ and $\tilde{\bf v}(t)\in T_{\tilde{\vect\theta}(t)} \mathcal S_c^D$ if $\tilde{\bf v}(0)\in T_{\tilde{\vect\theta}(0)} \mathcal S_c^D$.

\section{Error analysis of Spherical HMC}\label{Err}
Following \cite{leimkuhler04}, we now show that the discretization error $e_n=\Vert {\bf z}(t_n) - {\bf z}^{(n)}\Vert =\Vert ({\vect\theta}(t_n),{\bf v}(t_n)) - ({\vect\theta}^{(n)},{\bf v}^{(n)})\Vert$ (i.e. the difference between the true solution and the numerical solution) is $\mathcal O(\eps^3)$ locally and $\mathcal O(\eps^2)$ globally, where $\eps$ is the discretization step size. Here, we assume that ${\bf f}({\vect\theta},{\bf v}):= {\bf v}^{\textsf T}\vect\Gamma({\vect\theta}) {\bf v} + {\bf G}({\vect\theta})^{-1} \nabla_{\vect\theta}U({\vect\theta})$ is smooth; hence, ${\bf f}$ and its derivatives are uniformly bounded as ${\bf z}=({\vect\theta},{\bf v})$ evolves within finite time duration $T$.
We expand the true solution ${\bf z}(t_{n+1})$ at $t_n$:
\begin{equation}
\begin{split}
{\bf z}(t_{n+1}) & = {\bf z}(t_n)+\dot {\bf z}(t_n)\eps +\frac{1}{2}\ddot {\bf z}(t_n)\eps^2 +\mathcal O(\eps^3)\\
& = \begin{bmatrix}{\vect\theta}(t_n)\\{\bf v}(t_n)\end{bmatrix} + \begin{bmatrix}{\bf v}(t_n)\\-{\bf f}({\vect\theta}(t_n),{\bf v}(t_n))\end{bmatrix}\eps + \frac12 \begin{bmatrix}-{\bf f}({\vect\theta}(t_n),{\bf v}(t_n))\\-\dot{\bf f}({\vect\theta}(t_n),{\bf v}(t_n))\end{bmatrix}\eps^2 +\mathcal O(\eps^3)\\
%& = \begin{bmatrix}{\vect\theta}(t_n)\\{\bf v}(t_n)\end{bmatrix} + \begin{bmatrix}{\bf v}(t_n)\\-{\bf f}({\vect\theta}(t_n),{\bf v}(t_n))\end{bmatrix}\eps + \frac12 \begin{bmatrix}-{\bf f}({\vect\theta}(t_n),{\bf v}(t_n))\\-\frac{\pa {\bf f}}{\pa {\vect\theta}}{\bf v}(t_n)+\frac{\pa {\bf f}}{\pa {\bf v}}{\bf f}({\vect\theta}(t_n),{\bf v}(t_n))\end{bmatrix}\eps^2 +\mathcal O(\eps^3)
\end{split}
\end{equation}
We first consider Spherical HMC in the Cartesian coordinate, where ${\bf f}({\vect\theta},{\bf v})=\Vert\tilde{\bf v}\Vert^2\vect\theta + [{\bf I}-\vect\theta\tp{\vect\theta}] \nabla_{\vect\theta} U(\vect\theta)$.
From Equation \eqref{evolv} we have
\begin{equation}\label{vnormv}
\begin{split}
{\bf v}^{(n+1/2)} &= {\bf v}^{(n)} -\frac{\eps}{2} ({\bf I}-\vect\theta^{(n)}\tp{(\vect\theta^{(n)})}) \nabla_{\vect\theta} U(\vect\theta^{(n)})\\
\Vert \tilde{\bf v}^{(n+1/2)}\Vert^2 &= \Vert \tilde{\bf v}^{(n)}\Vert^2 -\eps \tp{({\bf v}^{(n)})} \nabla_{\vect\theta} U(\vect\theta^{(n)}) + \mathcal O(\eps^2)
\end{split}
\end{equation}
Now we expand Equation \eqref{gcir} using Taylor series as follows:
\begin{align*}
\begin{aligned}
\vect\theta^{(n+1)} &= \vect\theta^{(n)}[1-\Vert \tilde{\bf v}^{(n+1/2)}\Vert^2\eps^2/2+\mathcal O(\eps^4)] + {\bf v}^{(n+1/2)}\eps[1 - \Vert \tilde{\bf v}^{(n+1/2)}\Vert^2\eps^2/3! + \mathcal O(\eps^4)]\\
{\bf v}^{(n+3/4)} & = -\vect\theta^{(n)} \Vert \tilde{\bf v}^{(n+1/2)}\Vert^2 \eps[1 - \Vert \tilde{\bf v}^{(n+1/2)}\Vert^2\eps^2/3! + \mathcal O(\eps^4)] + {\bf v}^{(n+1/2)} [1-\Vert \tilde{\bf v}^{(n+1/2)}\Vert^2\eps^2/2+\mathcal O(\eps^4)]
\end{aligned}
\end{align*}
Substituting \eqref{vnormv} in the above equations yields 
\begin{align*}
\begin{aligned}
\vect\theta^{(n+1)} &= \vect\theta^{(n)} + {\bf v}^{(n+1/2)}\eps - \vect\theta^{(n)}\Vert \tilde{\bf v}^{(n+1/2)}\Vert^2\eps^2/2 +\mathcal O(\eps^3)\\
& = \vect\theta^{(n)} +  {\bf v}^{(n)}\eps - \frac12 {\bf f}({\vect\theta}^{(n)},{\bf v}^{(n)}) \eps^2 + \mathcal O(\eps^3)\\
{\bf v}^{(n+3/4)} & =  {\bf v}^{(n+1/2)} -\vect\theta^{(n)} \Vert \tilde{\bf v}^{(n+1/2)}\Vert^2 \eps - {\bf v}^{(n+1/2)} \Vert \tilde{\bf v}^{(n+1/2)}\Vert^2\eps^2/2 +\mathcal O(\eps^3)\\
 &= {\bf v}^{(n)} - [({\bf I}-\vect\theta^{(n)}\tp{(\vect\theta^{(n)})}) \nabla_{\vect\theta} U(\vect\theta^{(n)})/2 + \vect\theta^{(n)} \Vert \tilde{\bf v}^{(n)}\Vert^2]\eps\\
 &\phantom{= {\bf v}^{(n)}\;} + [\vect\theta^{(n)} \tp{({\bf v}^{(n)})} \nabla_{\vect\theta} U(\vect\theta^{(n)}) - {\bf v}^{(n)} \Vert \tilde{\bf v}^{(n)}\Vert^2/2]\eps^2 + \mathcal O(\eps^3)
\end{aligned}
\end{align*}
With the above results, we have
\begin{equation*}
\begin{split}
{\bf v}^{(n+1)} =& {\bf v}^{(n+3/4)} -\frac{\eps}{2} ({\bf I}-\vect\theta^{(n+1)}\tp{(\vect\theta^{(n+1)})}) \nabla_{\vect\theta} U(\vect\theta^{(n+1)})\\
%=& {\bf v}^{(n)} - [({\bf I}-\vect\theta^{(n)}\tp{(\vect\theta^{(n)})}) \nabla_{\vect\theta} U(\vect\theta^{(n)})/2 + \vect\theta^{(n)} \Vert \tilde{\bf v}^{(n)}\Vert^2]\eps + [\vect\theta^{(n)} \tp{({\bf v}^{(n)})} \nabla_{\vect\theta} U(\vect\theta^{(n)}) - {\bf v}^{(n)} \Vert \tilde{\bf v}^{(n)}\Vert^2/2]\eps^2\\
%& -\frac12\{[{\bf I}-\vect\theta^{(n)}\tp{(\vect\theta^{(n)})} - (\vect\theta^{(n)}\tp{({\bf v}^{(n)})}+{\bf v}^{(n)}\tp{(\vect\theta^{(n)})})\eps +\mathcal O(\eps^2) ] [\nabla_{\vect\theta} U(\vect\theta^{(n)}) +\nabla_{\vect\theta}^2 U(\vect\theta^{(n)}) (\vect\theta^{(n+1)}-\vect\theta^{(n)}) +\mathcal O(\eps^2) ] \} \eps\\
=& {\bf v}^{(n)} - {\bf f}({\vect\theta}^{(n)},{\bf v}^{(n)})\eps + [\vect\theta^{(n)} \tp{({\bf v}^{(n)})} \nabla_{\vect\theta} U(\vect\theta^{(n)}) - {\bf v}^{(n)} \Vert \tilde{\bf v}^{(n)}\Vert^2/2]\eps^2\\
& - \frac12[ ({\bf I}-\vect\theta^{(n)}\tp{(\vect\theta^{(n)})}) \nabla_{\vect\theta}^2 U(\vect\theta^{(n)}) {\bf v}^{(n)} - (\vect\theta^{(n)}\tp{({\bf v}^{(n)})}+{\bf v}^{(n)}\tp{(\vect\theta^{(n)})}) \nabla_{\vect\theta} U(\vect\theta^{(n)}) ]\eps^2 + \mathcal O(\eps^3)\\
=& {\bf v}^{(n)} - {\bf f}({\vect\theta}^{(n)},{\bf v}^{(n)})\eps - \frac12 \dot{\bf f}({\vect\theta}^{(n)},{\bf v}^{(n)})\eps^2 + \mathcal O(\eps^3)
\end{split}
\end{equation*}
where for the last equality we need to show $\tp{({\bf v}^{(n)})} \nabla_{\vect\theta} U(\vect\theta^{(n)})= -2\frac{d}{dt}\Vert\tilde{\bf v}^{(n)}\Vert^2$.
This can be proved as follows:
\begin{equation*}
\begin{split}
\frac{d}{dt}\Vert\tilde{\bf v}\Vert^2 &=  \frac{d}{dt} [ \Vert\tilde{\bf v}\Vert^2 + v_{D+1}^2] = 2[-\tp{\bf v}{\bf f} + v_{D+1}\dot v_{D+1}]\\
&= 2\left[-\tp{\bf v}{\bf f} +\left(-\frac{\tp{\dot{\vect\theta}} {\bf v}+ \tp{\vect\theta}\dot{\bf v}}{\theta_{D+1}} + \frac{\tp{\vect\theta}{\bf v}}{\theta_{D+1}^{2}}\dot\theta_{D+1}\right)v_{D+1}\right]\\%=2\left[-\tp{\bf v}{\bf f} +\left(-\frac{\Vert{\bf v}\Vert^2- \tp{\vect\theta}{\bf f}}{\theta_{D+1}} - \frac{v_{D+1}^2}{\theta_{D+1}}\right)v_{D+1}\right]\\
&= -2\left[\tp{\left({\bf v} -\frac{v_{D+1}}{\theta_{D+1}}\vect\theta\right)}{\bf f}  + \frac{v_{D+1}}{\theta_{D+1}} \Vert\tilde{\bf v}\Vert^2 \right]\\%= -2\left[\tp{\left({\bf v} -\frac{v_{D+1}}{\theta_{D+1}}\vect\theta\right)}[\vect\theta \Vert\tilde{\bf v}\Vert^2 + [{\bf I}-\vect\theta\tp{\vect\theta}] \nabla_{\vect\theta} U(\vect\theta)]  + \frac{v_{D+1}}{\theta_{D+1}} \Vert\tilde{\bf v}\Vert^2 \right]\\
&= -2\left[\left(\tp{\bf v}\vect\theta -\frac{v_{D+1}}{\theta_{D+1}}(\Vert\vect\theta\Vert^2-1)\right)\Vert\tilde{\bf v}\Vert^2 + \tp{\left({\bf v} -\frac{v_{D+1}}{\theta_{D+1}}\vect\theta\right)}[{\bf I}-\vect\theta\tp{\vect\theta}] \nabla_{\vect\theta} U(\vect\theta)] \right]\\
&= -2\left[\tp{\bf v}\nabla_{\vect\theta} U(\vect\theta) + \left(-\tp{\bf v}\vect\theta -\frac{v_{D+1}}{\theta_{D+1}}(1-\Vert\vect\theta\Vert^2)\right)\tp{\vect\theta}\nabla_{\vect\theta} U(\vect\theta)\right]
=-2 \tp{\bf v}\nabla_{\vect\theta} U(\vect\theta)
\end{split}
\end{equation*}
Therefore we have
\begin{equation}
{\bf z}^{(n+1)} := \begin{bmatrix}{\vect\theta}^{(n+1)}\\{\bf v}^{(n+1)}\end{bmatrix} = \begin{bmatrix}{\vect\theta}^{(n)}\\{\bf v}^{(n)}\end{bmatrix} +\begin{bmatrix}{\bf v}^{(n)}\\- {\bf f}({\vect\theta}^{(n)},{\bf v}^{(n)})\end{bmatrix}\eps + \frac12 \begin{bmatrix}-{\bf f}({\vect\theta}^{(n)},{\bf v}^{(n)})\\-\dot{\bf f}({\vect\theta}^{(n)},{\bf v}^{(n)})\end{bmatrix}\eps^2 + \mathcal O(\eps^3)
\end{equation}
The local error is
\begin{equation}
\begin{split}
e_{n+1} & =\Vert {\bf z}(t_{n+1}) - {\bf z}^{(n+1)}\Vert \\
%& = \left\Vert \begin{bmatrix}{\vect\theta}(t_n)-{\vect\theta}^{(n)}\\{\bf v}(t_n)-{\bf v}^{(n)}\end{bmatrix} + \begin{bmatrix}{\bf v}(t_n)-{\bf v}^{(n)}\\-[{\bf f}({\vect\theta}(t_n),{\bf v}(t_n))- {\bf f}({\vect\theta}^{(n)},{\bf v}^{(n)})]\end{bmatrix}\eps + \frac12 \begin{bmatrix}-[{\bf f}({\vect\theta}(t_n),{\bf v}(t_n))- {\bf f}({\vect\theta}^{(n)},{\bf v}^{(n)})]\\-[\dot{\bf f}({\vect\theta}(t_n),{\bf v}(t_n))- \dot{\bf f}({\vect\theta}^{(n)},{\bf v}^{(n)})]\end{bmatrix}\eps^2 + \mathcal O(\eps^3) \right\Vert\\
& = \left\Vert \begin{bmatrix}{\vect\theta}(t_n)-{\vect\theta}^{(n)}\\{\bf v}(t_n)-{\bf v}^{(n)}\end{bmatrix} + \begin{bmatrix}{\bf v}(t_n)-{\bf v}^{(n)}\\-[{\bf f}(t_n)- {\bf f}^{(n)}]\end{bmatrix}\eps + \frac12 \begin{bmatrix}-[{\bf f}(t_n)- {\bf f}^{(n)}]\\-[\dot{\bf f}(t_n)- \dot{\bf f}^{(n)}]\end{bmatrix}\eps^2 + \mathcal O(\eps^3) \right\Vert\\
& \leq (1+M_1\eps+M_2\eps^2) e_n + \mathcal O(\eps^3)
\end{split}
\end{equation}
where $M_k = c_k\sup_{t\in [0,T]}\Vert \nabla^k {\bf f}({\vect\theta}(t),{\bf v}(t))\Vert,\; k=1,2$ for some constants $c_k>0$. Accumulating the local errors by iterating the above inequality for $L=T/\eps$ steps provides the following global error:
\begin{equation}
\begin{split}
e_{L+1} & \leq (1+M_1\eps+M_2\eps^2) e_L + \mathcal O(\eps^3) \leq (1+M_1\eps+M_2\eps^2)^2 e_{L-1} + 3\mathcal O(\eps^3)\leq \cdots\\
& \leq (1+M_1\eps+M_2\eps^2)^L e_1 + L\mathcal O(\eps^3) \leq (e^{M_1T}+T)\eps^2 \to 0,\quad as\; \eps \to 0
\end{split}
\end{equation}

For Spherical HMC in the spherical coordinate, we conjecture that the integrator of Algorithm \ref{Alg:sSphHMC} still has order 3 local error and order 2 global error. One can follow the same argument as above to verify this.

\section{Bounce in diamond: Wall HMC for 1-norm constraint}\label{walldiamond}
\begin{figure}[t]
\vspace{-20pt}
\begin{center}
\includegraphics[width=5in, height=3.5in]{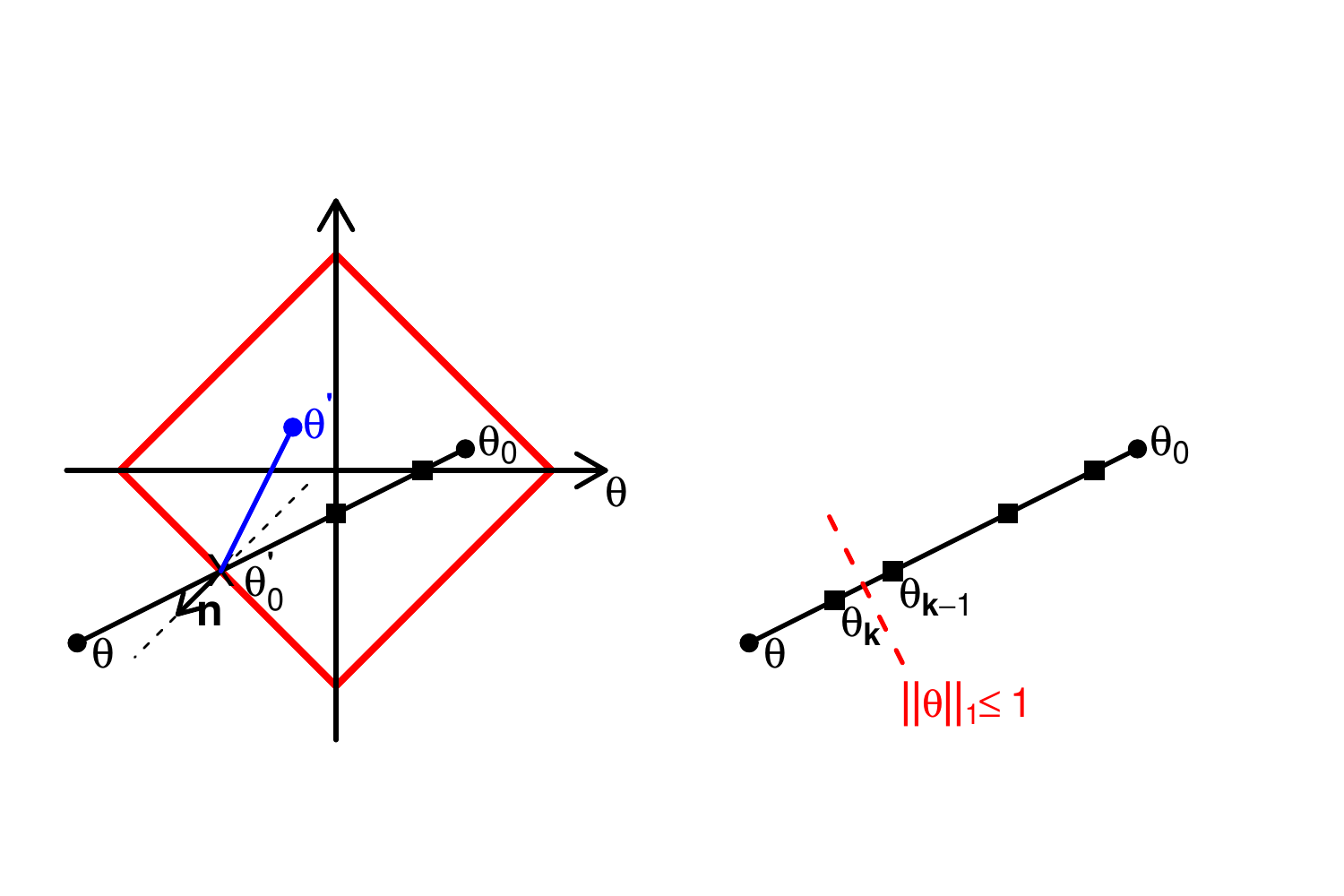}
\vspace{-10pt}
\caption{Wall HMC bounces in the 1-norm constraint domain. Left: given the current
state $\vect\theta_0$, Wall HMC proposes $\vect\theta$, but bounces of the boundary and
reaches $\vect\theta'$ instead. Right: determining the hitting time by
monitoring the first intersection point with coordinate that violates the
constraint.}
\vspace{-20pt}
\label{fig:bounceindiamond}
\end{center}
\end{figure}

\cite{neal11} discusses the \emph{Wall} HMC method for $\infty$-norm constraint only. We can however derive a similar approach for 1-norm constraint. As shown in the left panel of Figure \ref{fig:bounceindiamond}, given the current
state $\vect\theta_0$, HMC makes a proposal $\vect\theta$. It will hit the boundary
to move from $\vect\theta_0$ towards $\vect\theta$. To determine the hit point
`X', we are required to solve for $t\in(0,1)$ such that
\begin{equation}
\Vert\vect\theta_0+(\vect\theta-\vect\theta_0)t\Vert_1=\sum_{d=1}^D|\theta_0^d+(\theta^d-\theta_0^d)t|=1
\end{equation}
One can find the hitting time using the bisection method. However, a more efficient method is to find the orthant in which the sampler hits the boundary, i.e., find the normal direction ${\bf n}$ with elements being $\pm1$. Then, we can find $t$,
\begin{equation}
\Vert\vect\theta_0+(\vect\theta-\vect\theta_0)t\Vert_1 = \tp{\bf n}[\vect\theta_0+(\vect\theta-\vect\theta_0)t]=1 \implies t^* = \frac{1-\tp{\bf n}\vect\theta_0}{\tp{\bf n}(\vect\theta-\vect\theta_0)}
\end{equation}
Therefore the hit point is
$\vect\theta_0'=\vect\theta_0+(\vect\theta-\vect\theta_0)t^*$ and consequently
the reflection point is
\begin{equation}
\vect\theta'=\vect\theta-2{\bf n}^*\langle {\bf n}^*,\vect\theta-\vect\theta_0'\rangle = \vect\theta-2{\bf n}(\tp{\bf n}\vect\theta-1)/D
\end{equation}
where ${\bf n}^*:={\bf n}/\Vert{\bf n}\Vert_2$ and $\tp{\bf n}\vect\theta_0'=1$
because $\vect\theta_0'$ is on the boundary with the normal direction ${\bf n}^*$.

It is in general difficult to directly determine the intersection of $\vect\theta-\vect\theta_0$ with
boundary. Instead, we can find its intersections with coordinate planes
$\{\vect\pi_d\}_{d=1}^D$, where $\vect\pi_d:=\{\vect\theta\in\mathbb R^D|\theta^d=0\}$.
The intersection times are defined as ${\bf T}=\{\theta_0^d/(\theta_0^d-\theta^d)|\theta_0^d\neq \theta^d\}$.
We keep those between 0 and 1 and sort them in ascending order (Figure \ref{fig:bounceindiamond}, right panel). Then, we find the intersection points $\{\vect\theta_k:=\vect\theta_0+(\vect\theta-\vect\theta_0)T_k\}$ that violate the constraint $\Vert\vect\theta\Vert\leq 1$. Denote the first intersection point outside the constrained domain as $\vect\theta_k$. The signs of $\vect\theta_k$ and $\vect\theta_{k-1}$ determine the
orthant of the hitting point $\vect\theta_0'$.

Note, for each $d\in\{1,\cdots D\}$,
$(\sign(\vect\theta_k^d),\sign(\vect\theta_{k-1}^d))$ cannot be $(+,-)$ or
$(-,+)$, otherwise there exists an intersection point $\vect\theta^*:=\vect\theta_0+(\vect\theta-\vect\theta_0)T^*$ with some coordinate plane $\vect\pi_{d^*}$ between $\vect\theta_k$ and $\vect\theta_{k-1}$. Then $T_{k-1}<T^*<T_k$ contradicts the order of ${\bf T}$. \footnote{The same argument applies
when $T_k=1$, i.e. $\vect\theta$ is the first point outside the domain among
$\{\vect\theta_k\}$.}
Therefore any point (including $\vect\theta_0'$) between $\vect\theta_k$ and
$\vect\theta_{k-1}$ must have the same sign as
$\sign(\sign(\vect\theta_k)+\sign(\vect\theta_{k-1}))$; that is
\begin{equation}
{\bf n} = \sign(\sign(\vect\theta_k)+\sign(\vect\theta_{k-1}))
\end{equation}

After moving from $\vect\theta$ to $\vect\theta'$, we examine whether $\vect\theta'$ satisfies the constraint. If it does not satisfy the constraint, we repeat above procedure with
$\vect\theta_0\leftarrow\vect\theta_0'$ and $\vect\theta\leftarrow\vect\theta'$
until the final state is inside the constrained domain. Then we adjust the
velocity direction by
\begin{equation}
{\bf v} \leftarrow (\vect\theta'-\vect\theta_0')\frac{\Vert{\bf v}\Vert}{\Vert\vect\theta'-\vect\theta_0'\Vert}
\end{equation}
Algorithm \ref{Alg:diamond} summarizes the above steps.

%\begin{tiny}
\begin{algorithm}[t]
\caption{Wall HMC for 1-norm constraint (Wall HMC)}
\label{Alg:diamond}
\begin{algorithmic}
\STATE Initialize $\vect\theta^{(1)}$ at the current state $\vect\theta$ after transformation
\STATE Sample a new velocity value ${\bf v}^{(1)}\sim \mathcal N({\bf 0},{\bf I}_D)$
\STATE Calculate $H(\vect\theta^{(1)},{\bf v}^{(1)})=U(\vect\theta^{(1)}) + K({\bf v}^{(1)})$ 
\FOR{$\ell=1$ to $L$}
\STATE ${\bf v}^{(\ell+\frac{1}{2})} = {\bf v}^{(\ell)}-\frac{\eps}{2} \nabla_{\vect\theta} U(\vect\theta^{(\ell)})$
\STATE ${\vect\theta}^{(\ell+1)} = {\vect\theta}^{(\ell)} + \eps{\bf v}^{(\ell+\frac{1}{2})}$
\STATE set hit $\leftarrow$ false
\WHILE{$\Vert\vect\theta^{(\ell)}\Vert>1$}
\STATE find coordinate plane intersecting times:
${\bf T}=\{T_d:=\theta^{(\ell)}_d/(\theta^{(\ell)}_d-\theta^{(\ell+1)}_d)|\theta^{(\ell)}_d\neq\theta^{(\ell+1)}_d\}$
\STATE sort those between 0 and 1 in ascending order: ${\bf T}=\{0\leq T_k\uparrow\leq 1\}$
\STATE find the first point in $\{\vect\theta_k:=\vect\theta^{(\ell)}+(\vect\theta^{(\ell+1)}-\vect\theta^{(\ell)})T_k\}$ that violates $\Vert\vect\theta\Vert\leq 1$ and denote it as $\vect\theta_k$
\STATE set normal direction as ${\bf n} = \sign(\sign(\vect\theta_k)+\sign(\vect\theta_{k-1}))$
\STATE find the wall hitting time $t^* = (1-\tp{\bf n}\vect\theta^{(\ell)})/(\tp{\bf n}(\vect\theta^{(\ell+1)}-\vect\theta^{(\ell)}))$
\STATE $\vect\theta^{(\ell)}\leftarrow\vect\theta^{(\ell)}+(\vect\theta^{(\ell+1)}-\vect\theta^{(\ell)})t^*$
and $\vect\theta^{(\ell+1)}\leftarrow\vect\theta^{(\ell+1)}-2{\bf n}\langle {\bf n},\vect\theta^{(\ell+1)}-\vect\theta^{(\ell)}\rangle/\Vert{\bf n}\Vert_2^2$
\STATE set hit $\leftarrow$ true
\ENDWHILE
\IF{hit}
\STATE ${\bf v}^{(\ell+\frac{1}{2})} \leftarrow (\vect\theta^{(\ell+1)}-\vect\theta^{(\ell)})\Vert{\bf v}^{(\ell+\frac{1}{2})}\Vert/\Vert\vect\theta^{(\ell+1)}-\vect\theta^{(\ell)}\Vert$
\ENDIF
\STATE ${\bf v}^{(\ell+1)} = {\bf v}^{(\ell+\frac{1}{2})}-\frac{\eps}{2} \nabla_{\vect\theta} U(\vect\theta^{(\ell+1)})$
\ENDFOR
\STATE Calculate $H(\vect\theta^{(L+1)},{\bf v}^{(L+1)})=U(\vect\theta^{(L+1)}) + K({\bf v}^{(L+1)})$
\STATE Calculate the acceptance probability $\alpha =\min\{1, \exp[-H(\vect\theta^{(L+1)},{\bf v}^{(L+1)})+H(\vect\theta^{(1)},{\bf v}^{(1)})] \}$
\STATE Accept or reject the proposal according to $\alpha$ for the next state ${\vect\theta}'$
\end{algorithmic}
\end{algorithm}
%\end{tiny}

\clearpage

%%%% Reference %%%%%
%\vskip 0.2in
\bibliography{references}

\begin{thebibliography}{57}
\providecommand{\natexlab}[1]{#1}
\providecommand{\url}[1]{\texttt{#1}}
\expandafter\ifx\csname urlstyle\endcsname\relax
  \providecommand{\doi}[1]{doi: #1}\else
  \providecommand{\doi}{doi: \begingroup \urlstyle{rm}\Url}\fi

\bibitem[Ahmadian et~al.(2011)Ahmadian, Pillow, and Paninski]{ahmadian11}
Y.~Ahmadian, J.~W. Pillow, and L.~Paninski.
\newblock Efficient {Markov Chain Monte Carlo} methods for decoding neural
  spike trains.
\newblock \emph{Neural Computation}, 23\penalty0 (1):\penalty0 46--96, 2011.

\bibitem[Ahn et~al.(2013)Ahn, Chen, and Welling]{ahn13}
S.~Ahn, Y.~Chen, and M.~Welling.
\newblock {Distributed and adaptive darting Monte Carlo through regenerations}.
\newblock In \emph{Proceedings of the 16th International Conference on
  Artificial Intelligence and Statistics (AI Stat)}, 2013.

\bibitem[Ahn et~al.(2014)Ahn, Shahbaba, and Welling]{ahn14}
S.~Ahn, B.~Shahbaba, and M.~Welling.
\newblock {Distributed Stochastic Gradient MCMC}.
\newblock In \emph{International Conference on Machine Learning}, 2014.

\bibitem[Amari and Nagaoka(2000)]{amari00}
S.~Amari and H.~Nagaoka.
\newblock \emph{Methods of Information Geometry}, volume 191 of
  \emph{Translations of Mathematical monographs}.
\newblock Oxford University Press, 2000.

\bibitem[Andrieu and Moulines(2006)]{andrieu06}
C.~Andrieu and E.~Moulines.
\newblock On the ergodicity properties of some adaptive mcmc algorithms.
\newblock \emph{Annals of Applied Probability}, 16\penalty0 (3):\penalty0
  1462--1505, 2006.

\bibitem[Beal(2003)]{beal03}
M.~J. Beal.
\newblock \emph{Variational Algorithms for Approximate {Bayesian} Inference}.
\newblock PhD thesis, University College London, London, UK, 2003.

\bibitem[Beskos et~al.(2011)Beskos, Pinski, Sanz-Serna, and Stuart]{beskos11}
Alexandros Beskos, Frank~J Pinski, Jes{\'u}s~Mar{\i}a Sanz-Serna, and Andrew~M
  Stuart.
\newblock Hybrid monte carlo on hilbert spaces.
\newblock \emph{Stochastic Processes and their Applications}, 121\penalty0
  (10):\penalty0 2201--2230, 2011.

\bibitem[Blei et~al.(2003)Blei, Ng, and Jordan]{blei03}
David~M Blei, Andrew~Y Ng, and Michael~I Jordan.
\newblock Latent dirichlet allocation.
\newblock \emph{the Journal of machine Learning research}, 3:\penalty0
  993--1022, 2003.

\bibitem[Brockwell(2006)]{brockwell06}
A.~E. Brockwell.
\newblock Parallel markov chain monte carlo simulation by {Pre-Fetching}.
\newblock \emph{Journal of Computational and Graphical Statistics}, pages
  246--261, 2006.

\bibitem[Brubaker et~al.(2012)Brubaker, Salzmann, and Urtasun]{brubaker12}
Marcus~A. Brubaker, Mathieu Salzmann, and Raquel Urtasun.
\newblock A family of mcmc methods on implicitly defined manifolds.
\newblock In Neil~D. Lawrence and Mark~A. Girolami, editors, \emph{Proceedings
  of the Fifteenth International Conference on Artificial Intelligence and
  Statistics (AISTATS-12)}, volume~22, pages 161--172, 2012.

\bibitem[{Byrne} and {Girolami}(2013)]{byrne13}
S.~{Byrne} and M.~{Girolami}.
\newblock {Geodesic Monte Carlo on Embedded Manifolds}.
\newblock \emph{ArXiv e-prints}, January 2013.

\bibitem[Calderhead and Sustik(2012)]{calderhead12}
B.~Calderhead and M.~Sustik.
\newblock Sparse approximate manifolds for differential geometric mcmc.
\newblock In P.~Bartlett, F.C.N. Pereira, C.J.C. Burges, L.~Bottou, and K.Q.
  Weinberger, editors, \emph{Advances in Neural Information Processing Systems
  25}, pages 2888--2896. 2012.

\bibitem[Capp\'{e} et~al.(2008)Capp\'{e}, Douc, Guillin, Marin, and
  Robert]{cappe08}
Olivier Capp\'{e}, Randal Douc, Arnaud Guillin, Jean-Michel Marin, and
  Christian~P. Robert.
\newblock Adaptive importance sampling in general mixture classes.
\newblock \emph{Statistics and Computing}, 18\penalty0 (4):\penalty0 447--459,
  2008.

\bibitem[Craiu et~al.(2009)Craiu, R., and Y.]{craiu09}
R.~V. Craiu, Jeffrey R., and Chao Y.
\newblock Learn from thy neighbor: Parallel-chain and regional adaptive mcmc.
\newblock \emph{Journal of the American Statistical Association}, 104\penalty0
  (488):\penalty0 1454--1466, 2009.

\bibitem[de~Freitas et~al.(2001)de~Freitas, H{\o}jen-S{\o}rensen, Jordan, and
  Stuart]{freitas01}
N.~de~Freitas, P.~H{\o}jen-S{\o}rensen, M.~Jordan, and R.~Stuart.
\newblock {Variational MCMC}.
\newblock In \emph{Proceedings of the 17th Conference in Uncertainty in
  Artificial Intelligence}, UAI '01, pages 120--127, San Francisco, CA, USA,
  2001. Morgan Kaufmann Publishers Inc.
\newblock ISBN 1-55860-800-1.

\bibitem[do~Carmo(1992)]{docarmo92}
Manfredo~P. do~Carmo.
\newblock \emph{{Riemannian Geometry}}.
\newblock Birkh\"{a}user Boston, 1 edition, January 1992.
\newblock ISBN 0817634908.

\bibitem[Duane et~al.(1987)Duane, Kennedy, Pendleton, and Roweth]{duane87}
S.~Duane, A.~D. Kennedy, B~J. Pendleton, and D.~Roweth.
\newblock {Hybrid Monte Carlo}.
\newblock \emph{Physics Letters B}, 195\penalty0 (2):\penalty0 216 -- 222,
  1987.

\bibitem[Frank and Friedman(1993)]{frank93}
Ildiko~E. Frank and Jerome~H. Friedman.
\newblock {A Statistical View of Some Chemometrics Regression Tools}.
\newblock \emph{Technometrics}, 35\penalty0 (2):\penalty0 109--135, 1993.

\bibitem[Gelfand et~al.(2010)Gelfand, van~der Maaten, Chen, and
  Welling]{gelfand10}
A.~Gelfand, L.~van~der Maaten, Y.~Chen, and M.~Welling.
\newblock On herding and the cycling perceptron theorem.
\newblock In \emph{Advances in Neural Information Processing Systems 23}, pages
  694--702, 2010.

\bibitem[Geyer(1992)]{geyer92}
C.~J. Geyer.
\newblock {Practical Markov Chain Monte Carlo}.
\newblock \emph{Statistical Science}, 7\penalty0 (4):\penalty0 473--483, 1992.

\bibitem[Gilks et~al.(1998)Gilks, Roberts, and Sahu]{gilks98}
Walter~R. Gilks, Gareth~O. Roberts, and Sujit~K. Sahu.
\newblock Adaptive markov chain monte carlo through regeneration.
\newblock \emph{Journal of the American Statistical Association}, 93\penalty0
  (443):\penalty0 pp. 1045--1054, 1998.
\newblock ISSN 01621459.

\bibitem[Girolami and Calderhead(2011)]{girolami11}
M.~Girolami and B.~Calderhead.
\newblock {Riemann manifold Langevin and Hamiltonian Monte Carlo methods}.
\newblock \emph{Journal of the Royal Statistical Society, Series B}, (with
  discussion) 73\penalty0 (2):\penalty0 123--214, 2011.

\bibitem[Golub and Van~Loan(1996)]{golub96}
Gene~H. Golub and Charles~F. Van~Loan.
\newblock \emph{Matrix computations (3rd ed.)}.
\newblock Johns Hopkins University Press, Baltimore, MD, USA, 1996.
\newblock ISBN 0-8018-5414-8.

\bibitem[Hans(2009)]{hans09}
Chris Hans.
\newblock {Bayesian lasso regression}.
\newblock \emph{Biometrika}, 96\penalty0 (4):\penalty0 835--845, 2009.

\bibitem[Hoffman and Gelman(2011)]{hoffman11}
M.~Hoffman and A.~Gelman.
\newblock {The No-U-Turn Sampler: Adaptively Setting Path Lengths in
  Hamiltonian Monte Carlo}.
\newblock arxiv.org/abs/1111.4246, 2011.

\bibitem[Hoffman et~al.(2010)Hoffman, Bach, and Blei]{hoffman10}
Matthew Hoffman, Francis~R Bach, and David~M Blei.
\newblock Online learning for latent dirichlet allocation.
\newblock In \emph{advances in neural information processing systems}, pages
  856--864, 2010.

\bibitem[Kurihara et~al.(2006)Kurihara, Welling, and Vlassis]{kurihara06}
K.~Kurihara, M.~Welling, and N.~Vlassis.
\newblock Accelerated variational {Dirichlet} process mixtures.
\newblock In \emph{Advances of Neural Information Processing Systems -- NIPS},
  volume~19, 2006.

\bibitem[Lan et~al.(2014)Lan, Stathopoulos, Shahbaba, and Girolami]{lan14a}
Shiwei Lan, Vasileios Stathopoulos, Babak Shahbaba, and Mark Girolami.
\newblock Markov chain monte carlo from lagrangian dynamics.
\newblock \emph{Journal of Computational and Graphical Statistics}, \penalty0
  (just-accepted):\penalty0 00--00, 2014.

\bibitem[Leimkuhler and Reich(2004)]{leimkuhler04}
B.~Leimkuhler and S.~Reich.
\newblock \emph{Simulating Hamiltonian Dynamics}.
\newblock Cambridge University Press, 2004.

\bibitem[M{\o}ller et~al.(2006)M{\o}ller, Pettitt, Berthelsen, and
  Reeves]{moller06}
J.~M{\o}ller, A.~Pettitt, K.~Berthelsen, and R.~Reeves.
\newblock An efficient {M}arkov chain {M}onte {C}arlo method for distributions
  with intractable normalisation constants.
\newblock \emph{Biometrica}, 93, 2006.
\newblock to appear.

\bibitem[Murray et~al.(2010)Murray, Adams, and MacKay]{murray10}
Iain Murray, Ryan~Prescott Adams, and David~J.C. MacKay.
\newblock Elliptical slice sampling.
\newblock \emph{JMLR: W\&CP}, 9:\penalty0 541--548, 2010.

\bibitem[Mykland et~al.(1995)Mykland, Tierney, and Yu]{mykland95}
Per Mykland, Luke Tierney, and Bin Yu.
\newblock Regeneration in markov chain samplers.
\newblock \emph{Journal of the American Statistical Association}, 90\penalty0
  (429):\penalty0 pp. 233--241, 1995.
\newblock ISSN 01621459.

\bibitem[Neal and Roberts(2008)]{pneal08}
Peter Neal and Gareth~O. Roberts.
\newblock Optimal scaling for random walk metropolis on spherically constrained
  target densities.
\newblock \emph{Methodology and Computing in Applied Probability},
  Vol.10\penalty0 (No.2):\penalty0 277--297, June 2008.

\bibitem[Neal et~al.(2012)Neal, Roberts, and Yuen]{pneal12}
Peter Neal, Gareth~O. Roberts, and Wai~Kong Yuen.
\newblock Optimal scaling of random walk metropolis algorithms with
  discontinuous target densities.
\newblock \emph{Annals of Applied Probability}, Volume 22\penalty0 (Number
  5):\penalty0 1880--1927, 2012.

\bibitem[Neal(1993)]{neal93}
R.~M. Neal.
\newblock \emph{Probabilistic Inference Using {Markov} Chain {Monte Carlo}
  Methods}.
\newblock Technical Report CRG-TR-93-1, Department of Computer Science,
  University of Toronto, 1993.

\bibitem[Neal(2005)]{neal05}
R.~M. Neal.
\newblock The short-cut metropolis method.
\newblock Technical Report 0506, Department of Statistics, University of
  Toronto, 2005.

\bibitem[Neal(2011)]{neal11}
R.~M. Neal.
\newblock {MCMC using Hamiltonian dynamics}.
\newblock In S.~Brooks, A.~Gelman, G.~Jones, and X.~L. Meng, editors,
  \emph{Handbook of Markov Chain Monte Carlo}, pages 113--162. Chapman and
  Hall/CRC, 2011.

\bibitem[Neal(1996)]{neal96a}
Radford~M. Neal.
\newblock \emph{Bayesian Learning for Neural Networks}.
\newblock Springer-Verlag New York, Inc., Secaucus, NJ, USA, 1996.
\newblock ISBN 0387947248.

\bibitem[Neal(2003)]{neal03}
Radford~M. Neal.
\newblock Slice sampling.
\newblock \emph{Annals of Statistics}, 31\penalty0 (3):\penalty0 705--767,
  2003.

\bibitem[{Pakman} and {Paninski}(2013)]{pakman13}
A.~{Pakman} and L.~{Paninski}.
\newblock {Exact Hamiltonian Monte Carlo for Truncated Multivariate Gaussians}.
\newblock \emph{ArXiv e-prints}, August 2013.

\bibitem[Park and Casella(2008)]{park08}
Trevor Park and George Casella.
\newblock The bayesian lasso.
\newblock \emph{Journal of the American Statistical Association}, 103\penalty0
  (482):\penalty0 681--686, 2008.

\bibitem[Patterson and Teh(2013)]{patterson13}
Sam Patterson and Yee~Whye Teh.
\newblock Stochastic gradient riemannian langevin dynamics on the probability
  simplex.
\newblock In \emph{Advances in Neural Information Processing Systems}, pages
  3102--3110, 2013.

\bibitem[Propp and Wilson(1996)]{propp96}
J.~G. Propp and D.~B. Wilson.
\newblock Exact sampling with coupled {M}arkov chains and applications to
  statistical mechanics.
\newblock volume~9, pages 223--252, 1996.

\bibitem[Randal and P.(2011)]{douc11}
D.~Randal and Christian~R. P.
\newblock A vanilla rao-blackwellization of metropolis-hastings algorithms.
\newblock \emph{Annals of Statistics}, 39\penalty0 (1):\penalty0 261--277,
  2011.

\bibitem[Randal et~al.(2007)Randal, Arnaud, Jean-Michel, and Christian]{douc07}
D.~Randal, G.~Arnaud, M.~Jean-Michel, and R.~P. Christian.
\newblock Minimum variance importance sampling via population monte carlo.
\newblock \emph{ESAIM: Probability and Statistics}, 11:\penalty0 427--447,
  2007.

\bibitem[Roberts and Sahu(1997)]{roberts97}
G.~O. Roberts and S.~K. Sahu.
\newblock {Updating Schemes, Correlation Structure, Blocking and
  Parameterisation for the Gibbs Sampler}.
\newblock \emph{Journal of the Royal Statistical Society, Series B},
  59:\penalty0 291--317, 1997.

\bibitem[Shahbaba et~al.(2014)Shahbaba, Lan, Johnson, and Neal]{shahbaba14}
Babak Shahbaba, Shiwei Lan, Wesley~O. Johnson, and Radford~M. Neal.
\newblock Split hamiltonian monte carlo.
\newblock \emph{Statistics and Computing}, 24\penalty0 (3):\penalty0 339--349,
  2014.

\bibitem[Sherlock and Roberts(2009)]{sherlock09}
Chris Sherlock and Gareth~O. Roberts.
\newblock Optimal scaling of the random walk metropolis on elliptically
  symmetric unimodal targets.
\newblock \emph{Bernoulli}, Vol.15\penalty0 (No.3):\penalty0 774--798, August
  2009.

\bibitem[Spivak(1979)]{spivak79-1}
Michael Spivak.
\newblock \emph{{A Comprehensive Introduction to Differential Geometry}},
  volume~1.
\newblock Publish or Perish, Inc., Houston, second edition, 1979.

\bibitem[Teh et~al.(2006)Teh, Newman, and Welling]{teh06}
Yee~W Teh, David Newman, and Max Welling.
\newblock A collapsed variational bayesian inference algorithm for latent
  dirichlet allocation.
\newblock In \emph{Advances in neural information processing systems}, pages
  1353--1360, 2006.

\bibitem[Tibshirani(1996)]{tibshirani96}
Robert Tibshirani.
\newblock Regression shrinkage and selection via the lasso.
\newblock \emph{Journal of the Royal Statistical Society, Series B},
  58\penalty0 (1):\penalty0 267--288, 1996.

\bibitem[Wallach et~al.(2009)Wallach, Murray, Salakhutdinov, and
  Mimno]{wallach09}
Hanna~M Wallach, Iain Murray, Ruslan Salakhutdinov, and David Mimno.
\newblock Evaluation methods for topic models.
\newblock In \emph{Proceedings of the 26th Annual International Conference on
  Machine Learning}, pages 1105--1112. ACM, 2009.

\bibitem[Warnes(2001)]{warnes01}
G.~R. Warnes.
\newblock {The normal kernel coupler: An adaptive Markov Chain Monte Carlo
  method for efficiently sampling from multi-modal distributions}.
\newblock Technical Report Technical Report No. 395, University of Washington,
  2001.

\bibitem[Welling(2009)]{welling09}
M.~Welling.
\newblock Herding dynamic weights to learn.
\newblock In \emph{Proc. of Intl. Conf. on Machine Learning}, 2009.

\bibitem[Welling and Teh(2011)]{welling11}
M.~Welling and Y.~W. Teh.
\newblock {B}ayesian learning via stochastic gradient {L}angevin dynamics.
\newblock In \emph{Proceedings of the 28th International Conference on Machine
  Learning (ICML)}, pages 681--688, 2011.

\bibitem[West(1987)]{west87}
M.~West.
\newblock {On scale mixtures of normal distributions}.
\newblock \emph{Biometrika}, 74\penalty0 (3):\penalty0 646--648, 1987.

\bibitem[Zhang and Sutton(2011)]{zhang11}
Yichuan Zhang and Charles Sutton.
\newblock {Quasi-Newton Methods for Markov Chain Monte Carlo}.
\newblock In J.~Shawe-Taylor, R.~S. Zemel, P.~Bartlett, F.~C.~N. Pereira, and
  K.~Q. Weinberger, editors, \emph{Advances in Neural Information Processing
  Systems 24}, pages 2393--2401. 2011.

\end{thebibliography}

\end{document}